%% file: main.tex
\documentclass[]{llncs}

\usepackage{mdframed}
\usepackage{amsmath, amssymb}
\usepackage[scr=boondoxo]{mathalfa}
\usepackage{array, adjustbox}
\DeclareRobustCommand*{\Dashv}{%
  \Relbar\joinrel\mathrel{|}%
}
\usepackage{cancel}
\usepackage{subcaption}
\usepackage{soul}
\usepackage{tikz}
\usetikzlibrary{arrows,shapes}
\usetikzlibrary{positioning}
\usepackage{tabularx, longtable}
\usepackage{bussproofs}
\usepackage{rotating}
\usepackage{longtable}
\usepackage{stmaryrd}
\newtheorem{observation}{Observation}
  {\gdef\scalefactor{#1}\begin{center}\proofSkipAmount \leavevmode}%
  {\scalebox{\scalefactor}{\DisplayProof}\proofSkipAmount \end{center} }

\newcommand{\myeq}[2]{\mathrel{\stackrel{\makebox[0pt]{\mbox{\normalfont\tiny #2}}}{#1}}}

\newcommand{\tuple}[1]{\langle#1\rangle}

\newcommand{\ttt}{\texttt}

\newcommand{\mtt}{\mathtt}
\newcommand{\mrm}{\mathsf}

\newcommand{\Sem}[1]{\llbracket{#1}\rrbracket}

\newcommand{\lst}{\mathscr{l}}
\newcommand{\mrk}{\mathscr{m}}
\newcommand{\blue}{\mathtt{blue}}
\newcommand{\none}{\mathtt{none}}
\newcommand{\red}{\mathtt{red}}
\newcommand{\green}{\mathtt{green}}
\newcommand{\any}{\mathtt{any}}
\newcommand{\dashed}{\mathtt{dashed}}
\newcommand{\grey}{\mathtt{grey}}

\newcommand{\Sat}{\vDash\,}
\newcommand{\Sata}{\vDash^\alpha\,}

\newcommand{\sou}{\mrm{s}}
\newcommand{\tar}{\mrm{t}}
\newcommand{\mV}{\mrm{m_V}}
\newcommand{\mE}{\mrm{m_E}}
\newcommand{\lV}{\mrm{l_V}}
\newcommand{\lE}{\mrm{l_E}}

\newcommand{\RG}{\rho_g(G)}
\newcommand{\RH}{\rho_{g^*}(H)}
\newcommand{\E}[1]{\exists_{\mrm{#1}}}
\newcommand{\A}[1]{\forall_{\mrm{#1}}}

\newcommand{\SLP}{\small{\text{SLP}}}
\newcommand{\SUCCESS}{\small{\text{SUCCESS}}}
\newcommand{\FAIL}{\small{\text{FAIL}}}

\newcommand{\Def}[3]%
	{$~$\begin{definition}[#1]\label{#2}\normalfont
    #3\hfill$\square$
	\end{definition}}%
\newcommand{\Prop}[3]%
	{$~$\begin{proposition}[#1]\label{#2}\normalfont
    #3
	\end{proposition}}%
\newcommand{\Theo}[3]%
	{$~$\begin{theorem}[#1]\label{#2}\normalfont
    #3
	\end{theorem}}%
\newcommand{\Ex}[3]%
	{$~$\begin{example}[#1]\label{#2}\normalfont
    #3
	\end{example}}%
\newcommand{\Lemma}[2]%
	{\begin{lemma}\label{#1}\normalfont
    #2
	\end{lemma}}%
\newcommand{\Col}[2]%
	{$~$\\\begin{corollary}\label{#1}\normalfont
    #2
	\end{corollary}}%	
\newcommand{\Obsv}[2]%
	{$~$\\\begin{observation}\label{#1}\normalfont
    #2
	\end{observation}}%	

\newcommand{\pproof}[1]%
	{$~$\begin{proof}\normalfont #1 \qedhere
	\end{proof}}%
\newcommand{\Remark}[1]%
	{\begin{remark}\normalfont #1
	\end{remark}}%
	
\newcommand{\lkr}{\langle L\leftarrow K\rightarrow R,~\Gamma\rangle}

\DeclareUnicodeCharacter{2212}{-}

\begin{document}

\title{Verifying Graph Programs with\\First-Order Logic (Extended Version)}

\author{Gia S. Wulandari\thanks{Supported by Indonesia Endowment Fund for Education (LPDP)}\inst{1,2}
\and
Detlef Plump\inst{1}
}

\authorrunning{G.S. Wulandari, D. Plump}

\institute{Department of Computer Science, University of York, UK
\and School of Computing, Telkom University, Indonesia
}

\maketitle
\pagestyle{plain}
\begin{abstract}
\input{input/0_abstract}
\end{abstract}

\section{Introduction}
\label{sec:introduction}
\input{input/1_Introduction}

\section{Graph programming language GP\,2}
\label{sec:GP2}
\input{input/2_GP2}

\section{First-Order Formulas for Graph Programs}
\label{sec:FOL}
\input{input/3_FOforGP}

\section{Constructing a Strongest Liberal Postcondition}
\label{sec:SLP}
\input{input/4_SLP2}

\section{Proof Calculus}
\label{sec:proofrules}
\input{input/5_Verification}

\section{Soundness and completeness of proof calculi}
\label{sec:completeness}
\input{input/6_Completeness}

\section{Verification Example}
\label{sec:ex}
\input{input/7_Examples}

\section{Related Work}
\label{sec:related_work}
\input{input/8_RelatedWork}

\section{Conclusion and Future Work}
\label{sec:conclusion}
\input{input/9_Conclusion}

\bibliographystyle{abbrv}
\bibliography{firstorder}

\end{document}

%% file: input/0_abstract.tex
We consider Hoare-style verification for the graph programming language GP\,2. In previous work, graph properties were specified by so-called E-conditions which extend nested graph conditions. However, this type of assertions is not easy to comprehend by programmers that are used to formal specifications in standard first-order logic. In this paper, we present an approach to verify GP\,2 programs with a standard first-order logic. We show how to construct a strongest liberal postcondition with respect to a rule schema and a precondition. We then extend this construction to obtain strongest liberal postconditions for arbitrary loop-free programs. Compared with previous work, this allows to reason about a vastly generalised class of graph programs. In particular, many programs with nested loops can be verified with the new calculus.

%% file: input/1_Introduction.tex
Various Hoare-style proof systems for the graph programming language GP\,2 have been developed by Poskitt and Plump, see for example \cite{PoskittP12,Poskitt13}. These calculi use so-called E-conditions as assertions which extend nested graph conditions \cite{Pennemann09} with support for expressions. However, a drawback of E-conditions and nested graph conditions is that they are not easy to understand by average programmers who are typically used to write formal specifications in standard first-order logic. To give a simple example, the following  E-condition expresses that every node is labelled by an integer: $\forall$(\begin{tikzpicture}[scale=0.5, transform shape, minimum size=.1cm,baseline,thick]
   \node[circle, draw, label=left:\scriptsize 1] (a) at (0, 0) {$\mtt{a}$};
\end{tikzpicture}$\,,\,\exists$(\begin{tikzpicture}[scale=0.6, transform shape, minimum size=.1cm,baseline,thick]
   \node[circle, draw, label=left
   :\scriptsize 1] (a) at (0, 0) {$\mtt{a}$};
\end{tikzpicture}$\,\mid\,\mtt{int(a)}))$ $\land\,\forall$(\begin{tikzpicture}[scale=0.5, transform shape, minimum size=.1cm,baseline,thick]
   \node[circle, draw, label=left:\scriptsize 1, fill=red!50] (a) at (0, 0) {$\mtt{a}$};
\end{tikzpicture}$\,,\,\exists$(\begin{tikzpicture}[scale=0.6, transform shape, minimum size=.1cm,baseline,thick]
   \node[circle, draw, label=left
   :\scriptsize 1, fill=red!50] (a) at (0, 0) {$\mtt{a}$};
\end{tikzpicture}$\,\mid\,\mtt{int(a)}))$ $\land\,\forall$(\begin{tikzpicture}[scale=0.5, transform shape, minimum size=.1cm,baseline,thick]
   \node[circle, draw, label=left:\scriptsize 1, fill=green!50] (a) at (0, 0) {$\mtt{a}$};
\end{tikzpicture}$\,,\,\exists$(\begin{tikzpicture}[scale=0.6, transform shape, minimum size=.1cm,baseline,thick]
   \node[circle, draw, label=left
   :\scriptsize 1, fill=green!50] (a) at (0, 0) {$\mtt{a}$};
\end{tikzpicture}$\,\mid\,\mtt{int(a)}))$ $\land\,\forall$(\begin{tikzpicture}[scale=0.5, transform shape, minimum size=.1cm,baseline,thick]
   \node[circle, draw, label=left:\scriptsize 1, fill=blue!40] (a) at (0, 0) {$\mtt{a}$};
\end{tikzpicture}$\,,\,\exists$(\begin{tikzpicture}[scale=0.6, transform shape, minimum size=.1cm,baseline,thick]
   \node[circle, draw, label=left
   :\scriptsize 1, fill=blue!40] (a) at (0, 0) {$\mtt{a}$};
\end{tikzpicture}$\,\mid\,\mtt{int(a)}))$ $\land\,\forall$(\begin{tikzpicture}[scale=0.5, transform shape, minimum size=.1cm,baseline,thick]
   \node[circle, draw, label=left:\scriptsize 1, fill=gray!50] (a) at (0, 0) {$\mtt{a}$};
\end{tikzpicture}$\,,\,\exists$(\begin{tikzpicture}[scale=0.6, transform shape, minimum size=.1cm,baseline,thick]
   \node[circle, draw, label=left
   :\scriptsize 1, fill=gray!50] (a) at (0, 0) {$\mtt{a}$};
\end{tikzpicture}$\,\mid\,\mtt{int(a)}))$. Having to write two quantifiers that refer to the \emph{same} object appears unnatural from the perspective of standard predicate logic where a single universal quantifier would suffice. In the logic we introduce in this paper, the above condition is simply written as $\mrm{\forall_Vx(int(\lV(x)))}$. Both E-conditions and first-order formulas tend to get lengthy in examples, but our concern with nested graph conditions is that they require a non-standard interpretation. We believe that programmers cannot be expected to think in terms of morphisms and commuting diagrams, but should be allowed to work with a type of logic that they are familiar with.

In this paper we use assertions which are conventional first-order formulas enriched with GP\,2 expressions. We believe that these assertions are easier to comprehend by programmers than E-conditions and also offer the prospect of reusing the large range of tools available for first-order logic.

To use our assertions in Hoare-style verification, we show how to construct a strongest liberal postcondition Slp($c,r$) for a given rule schema $r$ and a precondition $c$. Based on this construction, we are able to construct a strongest liberal postcondition for any loop-free graph programs and preconditions. In addition, we are able to give syntactic conditions on host graphs which for any loop-free program express successful execution resp.\ the existence of a failing execution. With these results we obtain a verification calculus that can handle considerably more programs than the calculi in \cite{PoskittP12,Poskitt13}. In particular, many programs with nested loops can now be formally verified, which has been impossible so far.

Nevertheless, our proof calculus is not relatively complete because first-order logic is not powerful enough to express all necessary assertions. Therefore we present a semantic version of the calculus which turns out to be relatively complete. 

The remainder of this paper is structured as follows. A brief review of the graph programming language GP\,2 can be found in Section 2. In Section \ref{sec:FOL}, we introduce first-order formulas for GP\,2 programs. In Section \ref{sec:SLP}, we outline the construction of a strongest liberal postcondition for a given rule schema and first-order formula. Section \ref{sec:proofrules} presents the proof rules of a semantic and a syntactic verification calculus, and identifies the class of programs that can be verified with the syntactic calculus.
%Section \ref{sec:completeness} then discusses the soundness and completeness of our calculi. 
In Section \ref{sec:ex}, we demonstrate how to verify a graph program for computing a 2-colouring of an input graph. In Section \ref{sec:completeness}, we discuss the soundness and completeness of our proof calculi. Section \ref{sec:related_work} contains a comparison of our approach with other approaches in the literature. Finally, we conclude and give some topics for future work in Section \ref{sec:conclusion}.

%% file: input/2_GP2.tex
GP\,2 is a graph programming language using graph transformation systems with the double-pushout approach, which was introduced in \cite{Plump09}. In this section, we briefly introduce graph transformation systems in GP\,2. For more detail documentation of GP\,2, we refer readers to \cite{Bak15a}.

\subsection{GP\,2 graphs}

A graph is a flexible structure in representing objects and relations between them. Objects are usually represented by nodes, while edges represent relations between them. Additional information about the objects and the relations are usually written as a label of the nodes and edges. Also, sometimes rooted nodes are used to distinguish some nodes with others.

\begin{definition}[Label Alphabet]\label{def:label}
\emph{A label alphabet} $\mathcal{C}=\langle \mathcal{C}_V,\mathcal{C}_E\rangle$ is a pair comprising a set $\mathcal{C}_V$ of node labels and a set $\mathcal{C}_E$ of edge labels. \qed
\end{definition}
% \Def{Label alphabet}{def:label}{
% \emph{A label alphabet} $\mathcal{C}=\langle \mathcal{C}_V,\mathcal{C}_E\rangle$ is a pair comprising a set $\mathcal{C}_V$ of node labels and a set $\mathcal{C}_E$ of edge labels.
% }

\Def{Graph over label alphabet; class of graphs}{def:graphs}{
\emph{A graph over label alphabet $\mathcal{C}$} is a system $G = \langle V_G, E_G, s_G, t_G, l_G, m_G, p_G\rangle$ comprising a finite set $V_G$ of nodes, a finite set $E_G$ of edges, source and target functions $s_G, t_G : E_G \rightarrow V_G$, a partial node labelling function $l_G : V_G \rightarrow \mathcal{C}_V$, an edge labelling function $m_G : E_G \rightarrow \mathcal{C}_E$, and a partial rootedness function $p_G : V_G \rightarrow \{0,1\}$. A \emph{totally labelled graph} is a graph where its node labelling and rootedness functions are total. We then denote by $\mathcal{G(C_\perp)}$ the set of all graphs over $\mathcal{C}$, and $\mathcal{G(C)}$ the set of all totally labelled graphs over $\mathcal{C}$.
}

Graphically, in this paper, we represent a node with a circle, an edge with an arrow where its tail and head represent the source and target, respectively. The label of a node is written inside the node, while the label of an edge is written next to the arrow. The rootedness of a node $v$ is represented by the line of the circle representing $v$, that is, standard circle for an unrooted node ($p(v)=0$) and bold circle for a rooted node ($p(v)=1$). To represent a node $v$ with undefined rootedness  ($p(v)=\perp$), we also use a standard circle. We use the same representation because nodes with undefined rootedness only exist in the interface of GP\,2 rules, and the interface contains only this kind of nodes so that no ambiguity will arise even when we use the same representation.

There are two kinds of graphs in GP\,2, that are host graphs and rule graphs. A label in a host graph is a pair of list and mark, while a label in a rule graph is a pair of expression and mark. Input and output of graph programs are host graphs, while graphs in GP\,2 rules are rule graphs.

\Def{GP\,2 labels}{def:GP2label}{
A set of node marks, denoted by $\mathbb{M}_V$, is the set $\{\none,$ $\red, \blue, \green, \grey\}$, while a set of edge marks, denoted by $\mathbb{M}_E$, is the set $\{\none,$ $\red, \blue, \green, \dashed\}$. A set of lists, denoted by $\mathbb{L}$, consists of all (list of) integers and strings that can be derived from the following abstract syntax:
\begin{center}\small
\begin{tabular}{lcl}
      $\mathbb{L}$ & ::= & $\mtt{empty}$ $\mid$ GraphExp $\mid$ $\mathbb{L}$ `:' $\mathbb{L}$ \\
      GraphExp & ::= & [`-'] Digit \{Digit\} $\mid$ GraphStr\\
      GraphStr & ::= & `~``~' \{Character\} ' " ' $\mid$ GraphStr `.' GraphStr
    \end{tabular}    
\end{center}
\noindent where Character is the set of all printable characters except `"' (i.e. ASCII characters 32, 33, and 35-126), while Digit is the digit set $\{0,\ldots,9\}$. 

\emph{A GP\,2 node label} is a pair $\tuple{\lst^V, \mrk^V}\in\mathbb{L}\times\mathbb{M}_V$, and a \emph{GP\,2 edge label} is a pair $\tuple{\lst^E, \mrk^E}\in\mathbb{L}\times\mathbb{M}_E$. We then denote the set of GP\,2 labels as $\mathcal{L}=\tuple{\mathcal{L}_V,\mathcal{L}_E}$.
}

The colon operator `:' is used to concatenate atomic expressions while the dot operator `.' is used to concatenate strings. The empty list is signified by the keyword $\mtt{empty}$, where it is displayed as a blank label graphically.

Basically, in a host graph, a list consists of (list of) integers and strings which are typed according to hierarchical type system as below:
\begin{center}
    \begin{tikzpicture}[remember picture,
        inner/.style={inner sep=0pt},
        outer/.style={inner sep=2pt}, scale=0.9
        ]
    \node[outer] (A1) at (0,0) {
        \begin{tikzpicture}[scale=1, transform shape]
		    \node[inner] (Aa) at (4.5,-0.5) {\footnotesize{$\mtt{char}$}};	
		    \node[inner] (Ab) at (3,-0.5) {\footnotesize{$\mtt{string}$}};	
	    	\node[inner] (Ac) at (3,0.5) {\footnotesize{$\mtt{int}$}};	
		    \node[inner] (Ad) at (1.5,0) {\footnotesize{$\mtt{atom}$}};	
		    \node[inner] (Ae) at (0,0) {\footnotesize{$\mtt{list}$}};	
    		\draw[white] (Aa) to node[black] {$\supseteq$} (Ab);
    		\draw[white] (Ab) to node[black,rotate=-45] {$\supseteq$} (Ad);
    		\draw[white] (Ac) to node[black,rotate=45] {$\supseteq$} (Ad);
    		\draw[white] (Ad) to node[black] {$\supseteq$} (Ae);
		\end{tikzpicture}};
	\end{tikzpicture}
\end{center}
where the domain for $\mtt{list, atom, int, string,}$ and $\mtt{char}$ is $\mathbb{Z}~\cup$ Char$^*$)$^*, \mathbb{Z}~\cup$ Char$^*, \mathbb{Z},$ \{Char\}$^*$, and Char respectively.

\Def{Labels of rules in GP\,2}{def:RSlabel}{
Let $\mathbb{E}$ be the set of all expressions that can be derived from the syntactic class List in the following grammar:
\begin{center}\small
\begin{tabular}{lcl}
     $\mathbb{E}$ & ::= & List\\
      List & ::= & $\mtt{empty}$ $\mid$ Atom $\mid$ List \lq:' List $\mid$ ListVar \\
      Atom & ::= & Integer $\mid$ String $\mid$ AtomVar \\
      Integer & ::= & [\lq-']~Digit~\{Digit\}~$\mid$~\lq('Integer\lq)'  $\mid$ IntVar \\
       & & $\mid$ Integer (\lq+' $\mid$ \lq-' $\mid$ \lq*' $\mid$ \lq/') Integer \\
       & & $\mid$ ($\mtt{indeg}$ $\mid$ $\mtt{outdeg}$) \lq('NodeId\lq)' \\
       & & $\mid$ $\mtt{length}$ \lq('AtomVar $\mid$ StringVar $\mid$ ListVar\lq)' \\
      String & ::= & Char $\mid$ String `.' String $\mid$ StringVar\\
      Char & ::= & ` `` '\{Character\}` " ' $\mid$ CharVar\\
    \end{tabular}
\end{center}
\noindent where ListVar, AtomVar, IntVar, StringVar, and CharVar represent variables of type $\mtt{list, atom,}$ $\mtt{int, string,}$ and $\mtt{char}$ respectively. Also, NodeId represents node identifiers.

\emph{Label alphabet} for left and right-hand graphs of a GP\,2 rule, denoted by $\mathcal{S}$, contains all pairs node label $\tuple{\lst^V,\mrk^V}\in\mathbb{E}\times(\mathbb{M}_V\cup\{\any\})$ and edge label $\tuple{\lst^E,\mrk^E}\in\mathbb{E}\times(\mathbb{M}_E\cup\{\any\})$}

\Def{GP\,2 host graphs and rule graphs}{def:GP2graphs}{
A \emph{host graph} $G$ is a graph over $\mathcal{L}$, and a \emph{rule graph} $H$ is a graph over $\mathcal{S}$. A host graph (or rule graph) $G$ has a node labelling function $l^V_G=\tuple{\lst^V_G,\mrk^V_G}$ such that for every node $v\in V_G$, $\lst^V_G(v)$ is defined if and only if $\mrk^E_G(v)$ is defined. Similarly, for every edge $e\in E_G$, $\lst^E_G(e)$ is defined iff $\mrk^E_G(e)$ is defined.
}

If we consider the grammars of Definition \ref{def:label} and Definition \ref{def:RSlabel}, it is obvious that $\mathbb{L}$ is part of expressions that can be derived in the latter grammar. Hence, $\mathcal{L}\subset\mathcal{S}$, which means we can consider host graphs as special cases of rule graphs. From here, we may refer `rule graphs' simply as `graphs', which also means host graphs are included.

Syntactically, a graph in GP\,2 is written based on the following syntax:
\begin{center}\small
\begin{tabular}{lcl}
      Graph & ::= & [Position] `$\mid$' {Nodes} `$\mid$' {Edges}\\
      Nodes & ::= & `(' NodeId [`(R)'] `,' Label [ `,' Position ] `)'\\
      Edges & ::= & `(' EdgeId [`(B)']`,' NodeId `,' NodeId `,' Label `)'\\
    \end{tabular}
\end{center}

\noindent where Position is a set of floating-point cartesian coordinates to store layout information for graphical editors, NodeId and EdgeId are sets of node and edge identifiers, and Label is set of labels as defined in Definition \ref{def:label} and Definition \ref{def:RSlabel}. Also, (R) in Nodes is used for rooted nodes while (B) in Edges is used for bidirectional edges. Bidirectional edges may exist in rule graphs but not in host graphs.

The marks \ttt{red}, \ttt{green}, \ttt{blue} and \ttt{grey} are graphically represented by the obvious colours while \ttt{dashed} is represented by a dashed line. The wildcard mark $\mtt{any}$ is represented by the colour magenta. 

Node labels are undefined only in the interface graphs of rule schemata. This allows rules to relabel nodes. Similarly, the root function is undefined only for the nodes of interface graphs. The purpose of root nodes is to speed up the matching of rule schemata \cite{Bak15a,PlumpB12}.

\Ex{A graph}{ex:graph}{Let $G$ be a graph with $V_G=\{1,2,3\}, E_G=\{e1,e2\}, s_G=\{e1\mapsto 1, e2\mapsto 1\}, t_G=\{e1\mapsto 2, e2\mapsto 3\}, l_G=\{1\mapsto \tuple{a,\none}, 2\mapsto \tuple{b,\red}, 3\mapsto \tuple{a+2,\none}\}, m_G=\{e1\mapsto \tuple{d,\none}, e2\mapsto \tuple{e,\dashed}\},$ and $p_G=\{1\mapsto 0, 2\mapsto 1, 3\mapsto 0\}$. Graphically, $G$ can be seen as the following graph:
\begin{center}
    \begin{tikzpicture}[remember picture,
  inner/.style={circle,draw,minimum size=18pt},
  outer/.style={inner sep=2pt}, scale=0.9
  ]
  \node[outer] (A) at (0,0) {
  \begin{tikzpicture}[scale=0.8, transform shape]
		\node[inner, label=below:\tiny 1] (Aa) at (0,0) {$\mtt{a}$};	
		\node[inner, label=below:\tiny 2, fill=red, ultra thick] (Ab) at (1.5,0) {$\mtt{b}$};	
        \node[inner, label=below:\tiny 3] (Ac) at (-1.5,0) {$\mtt{a+2}$};	
		\draw[-latex] (Aa) to node[above] {$\mtt{d}$} (Ab);
		\draw[-latex, dashed] (Aa) to node[above] {$\mtt{e}$} (Ac);
		\end{tikzpicture}};
	\end{tikzpicture}
\end{center}
Syntactically in GP\,2, $G$ is written as follows:
\[
\small{\mtt{~\mid~(1, a)~(2(R), b\#red)~(3, a+2) \mid~(e1, 1, 2, d)~(e2, 1, 3, e\#dashed)}}\]
}        

To show a relation between graphs, which are what we do in graph transformations, we use graph morphism. In GP\,2, in addition to graph morphism, we also have graph premorphisms which is similar to graph morphisms but not considering node and edge labels.

\Def{Graph morphisms}{def:morphisms}{
Given two graphs $G$ and $H$. A graph morphism $g:G\rightarrow H$ is a pair of mapping $g=\langle g_V:V_G\rightarrow V_H, g_E:E_G\rightarrow E_H\rangle$ such that for all nodes and edges in $G$, sources, targets, labels, marks, and rootedness are preserved. That is: $g_V\circ s_G = s_H\circ g_E$, $g_V\circ t_G = t_H\circ g_E$, $l_H(g_V(x)) = l_G(x)$, $m_H(g_E(y)) = m_G(y)$ for all $x\in V_G$ such that $l_G(x)\neq\perp$ and all $y\in E_G$ such that $m_G(y)\neq\perp$. Also, for all $v\in V_G,$ such that $p_G(v)\neq\perp$ $p_H(g_V(v))=p_G(v)$. A graph morphism $g$ is injective (surjective) if both $g_V$ and $g_E$ are injective (surjective). A graph morphism $g : G \rightarrow H$ is an \emph{isomorphism} if $g$ is both injective and surjective, also satisfies $l_H(g_V(v)) = \perp$ for all nodes $v$ with $l_G(v) = \perp$ and $v\in r_G$ iff $g(v)\in r_H$ for all $v\in V_G$. Furthermore. we call a morphism $g$ as an \emph{inclusion} if $g(x)=x$ for all $x$ in $G$.
}

\Def{Premorphisms}{def:premorphism}{
Given a rule graph $L$ and a host graph $G$. A premorphism $g:L\rightarrow G$ consists of two injective functions $g_V:V_L\rightarrow V_G$ and $g_E:E_L\rightarrow E_G$ that preserve sources, targets, and rootedness.
}

\subsection{Conditional rule schemata}
Like traditional rules in graph transformation that use double-pushout approach, rules in GP\,2 (called rule schemata) consists of a left-hand graph, an interface graph, and a right-hand graph. In addition, GP\,2 also allows a condition for the left-hand graph. When a condition exists, the rule is called a conditional rule schema.

% \Def{Rules}{def:rules}{
% A \emph{rule} $r=\langle L\leftarrow K\rightarrow R\rangle$ comprises $L,R\in\mathcal{G(L)}$ and $K\in\mathcal{G(L_\perp)}$, and inclusions $K\rightarrow L$ and $K\rightarrow R$. $L$ is called the \emph{left-hand graph} of $r$, $R$ is the \emph{right-hand graph} of $r$, and $K$ is the interface of $r$.
% }

% While the left and right-hand graphs of a rule are totally labelled host graph, a rule schema may have a rule graph as its left or right-hand graph. In addition, interface of a rule schema containing only unlabelled nodes with undefined rootedness.

\Def{Rule schemata}{def:RS}{
A \emph{rule schema} $r=\langle L\leftarrow K\rightarrow R\rangle$ comprises totally labelled rule graphs $L$ and $R$, a graph $K$ containing only unlabelled nodes with undefined rootedness, also inclusions $K\rightarrow L$ and $K\rightarrow R$. All list expressions in $L$ are simple (i.e. no arithmetic operators, contains at most one occurrence of a list variable, and each occurrence of a string sub-expression contains at most one occurrence of a string variable). Moreover, all variables in $R$ must also occur in $L$, and every node and edge in $R$ whose mark is $\mtt{any}$ has a preserved counterpart item in $L$. An \emph{unrestricted rule schema} is a rule schema without restriction on expressions and marks in its left and right-hand graph.
}

\Remark{
Note that the left and right-hand graph of a rule schema can be rule graphs or host graphs since a host graph is a special case of rule graphs. In GP\,2, we only consider rule schemata (with restrictions). In this paper, we use unrestricted rule schemata to be able to express the properties of the inverse of a rule schema.}  

In GP\,2, a condition can be added to a rule schema. This condition expresses properties that must be satisfied by a match of the rule schema. The variables occur in a rule schema condition must also occur in the left-hand graph of the rule schema.

\Def{Conditional rule schemata}{def:CRS}{
A conditional rule schema is a pair $\langle r,\Gamma\rangle$ with $r$ a rule schema and $\Gamma$ a condition that can be derived from Condition in the grammar below:
\begin{center}\small
    \begin{tabular}{lcl}
        Condition & ::= & ($\mtt{int \mid char \mid string \mid atom}$) `('Var`)' \\
        && $\mid$ List (`$\mtt{=}$' $\mid$ `\ttt{!=}') List\\
        && $\mid$ Integer (`\ttt{>}' $\mid$ `\ttt{>=}' $\mid$ `\ttt{<}' $\mid$ `\ttt{<=}') Integer\\
        && $\mid$ $\mtt{edge}$ `(' NodeId `,' NodeId [`,' List [Mark]] `)'\\
        && $\mid$ $\mtt{not}$ Condition\\
        && $\mid$ Condition ($\mtt{and}$ $\mid$ $\mtt{or}$) Condition\\
        && $\mid$ `(' Condition `)'\\
        Var & ::= & ListVar $\mid$ AtomVar $\mid$ IntVar $\mid$ StringVar $\mid$ CharVar\\
        Mark & ::= & $\mtt{red \mid green \mid blue \mid dashed \mid any}$
    \end{tabular}
\end{center}
such that all variables that occur in $\Gamma$ also occur in the left-hand graph of $r$.
}

Left-hand graph of a rule schema consists of a rule graph, while a morphism is a mapping function from a host graph. To obtain a host graph from a rule graph, we can assign constants for variables in the rule graph. For this, here we define assignment for labels.

A conditional rule schema $\tuple{L\gets K\to R,\, \Gamma}$ is applied to a host graph $G$ in stages: (1) evaluate the expressions in $L$ and $R$ with respect to a premorphism $g\colon L \to G$ and a label assignment $\alpha$, obtaining an instantiated rule $\tuple{L^{g,\alpha}\gets K\to R^{g,\alpha}}$; (2) check that $g\colon L^{g,\alpha} \to G$ is label preserving and that the evaluation of $\Gamma$ with respect to $g$ and $\alpha$ returns true; (3) construct two natural pushouts based on the instantiated rule and $g$.

\Def{Label assignment}{def:assign}{
Consider a rule graph $L$ and the set $X$ of all variables occurring in $L$. For each $x\in X$, let dom$(x)$ denotes the domain of $x$ associated with the type of $x$. A \emph{label assignment} for $L$ is a triple $\alpha=\tuple{\alpha_\mathbb{L},\, \mu_V,\, \mu_E}$ where $\alpha_\mathbb{L}\colon X\rightarrow\mathbb{L}$ is a function such that for each $x\in X$, $\alpha_\mathbb{L}(x)\in$\,dom$(x)$, and $\mu_V\colon V_L\to \mathbb{M}_V\backslash\{\mtt{none}\}$ and $\mu_E\colon E_L\to \mathbb{M}_E\backslash\{\mtt{none}\}$ are partial functions assigning a mark to each node and edge marked with \ttt{any}.
}

For a conditional rule schema $\lkr$ with the set $X$ of all list variables in $L$, set $Y$ (or $Z$) of all nodes (or edges) in $L$ whose mark is $\any$, and label assignment $\alpha_L$, we denote by $L^\alpha$ the graph $L$ after the replacement of every $x\in X$ with $\alpha_\mathbb{L}(x)$, every $\mrk^V_L(i)$ for $i\in Y$ with $\mu_{V}(i)$, and every $\mrk^E_L(i)$ for $i\in Z$ with $\mu_{E}(i)$. Then for an injective graph morphism $g:L^\alpha\rightarrow G$ for some host graph $G$, we denote by $\Gamma^{g,\alpha}$ the condition that is obtained from $\Gamma$ by substituting $\alpha_\mathbb{L}(x)$ for every variable $x$, $g(v)$ for every $v\in V_L$, and $g(e)$ for every $e\in E_L$.

The satisfaction of $\Gamma^{g,\alpha}$ in $G$ is required for the application of a conditional rule schema. In addition, the application also depends on the dangling condition, which is a condition that asserts the production of a graph after node removal. 

\Def{Dangling condition; match}{def:dang}{
Let $r=L\leftarrow K\rightarrow R$ be a rule schema with host graphs $L$ and $R$. Let also $G$ be a host graph, and $g:L\rightarrow G$ be an injective morphism. The \emph{dangling condition} is a condition where no edge in $G-g(L)$ is incident to any node in $g(L-K)$. When the dangling condition is satisfied by $g$, we say that $g$ is a \emph{match} for $r$.}

Since a rule schema has an unlabelled graph as its interface, a natural pushout, i.e. a pushout that is also a pullback, is required in a rule schema application. This approach is introduced in \cite{HabelPlump02c} for unrooted graph programming. The approach is the modified for rooted programming in \cite{Bak15a,Campbell19}.

\Def{Direct derivation; comatch}{def:dder}{
A \emph{direct derivation} from a host graph $G$ to a host graph $H$ via a rule $r = \tuple{L\leftarrow K\rightarrow R}$ consists of a natural double-pushout as in \figurename~\ref{fig:dder}, where $g:L\rightarrow G$ and $g^*:R\rightarrow H$ are injective morphisms. If there exists such direct derivation, we write $G \Rightarrow_{r,g} H$, and we say that $g^*$ is a \emph{comatch} for $r$.  
\begin{figure}
    \centering
    \begin{tikzpicture}[scale=0.7, transform shape]
	\node (2) at (0.75, 0) {(1)};
	\node (a) at (1.5, 0.75) {$K$};
	\node (b) at (0, 0.75) {$L$};
	\node (c) at (3, 0.75) {$R$};
	\node (d) at (1.5, -0.75) {$D$};
	\node (1) at (2.25, 0) {(2)};
	\node (e) at (3, -0.75) {$H$};
	\node (f) at (0, -0.75) {$G$};
	\draw[->] (a) to node {} (b);
	\draw[->] (a) to node {} (c);
	\draw[->] (a) to node {} (d);
	\draw[->] (d) to node {} (e);
	\draw[->] (d) to node {} (f);
	\draw[->] (c) to node [right]{$g^*$} (e);
	\draw[->] (b) to node [left]{$g$} (f);
    \end{tikzpicture}
    \caption{A direct derivation for a rule $r=\tuple{L\leftarrow K\rightarrow R}$}
    \label{fig:dder}
\end{figure}
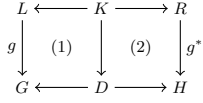
}

Note that we require natural double-pushout in direct derivation. We use a natural pushout to have a unique pushout complement up to isomorphism in relabelling graph transformation\cite{HabelPlump02c,HristakievP16}. In \cite{Bak15a}, a graph morphism preserves rooted nodes while here we require a morphism to preserve unrooted nodes as well. We require the preservation of unrooted nodes to prevent a non-natural pushout as can be seen in \figurename~\ref{fig:NPOPO} \cite{Campbell19}. In addition, we need a natural double-pushout because we want to have invertible direct derivations.

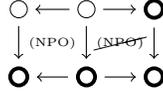
\begin{figure}
    \centering
    \begin{tikzpicture}[scale=0.9, transform shape]
	\node (1) at (0.5, 0) {\tiny{(NPO)}};
	\node (2) at (1.5, 0) {\tiny{\cancel{(NPO)}}};
	\node (K) at (1, 0.5) {
	\begin{tikzpicture}[scale=0.4, transform shape]
		\node[circle,draw,minimum size=18pt] (Aa) at (0,0) {};	
		\end{tikzpicture}
	};
	\node (L) at (0, 0.5) {\begin{tikzpicture}[scale=0.4, transform shape]
		\node[circle,draw,minimum size=18pt] (Aa) at (0,0) {};	
		\end{tikzpicture}};
	\node (R) at (2, 0.5) {\begin{tikzpicture}[scale=0.4, transform shape]
		\node[circle,draw,minimum size=18pt, ultra thick] (Aa) at (0,0) {};	
		\end{tikzpicture}};
	\node (D) at (1, -0.5) {\begin{tikzpicture}[scale=0.4, transform shape]
		\node[circle,draw,minimum size=18pt,ultra thick] (Aa) at (0,0) {};	
		\end{tikzpicture}};
	\node (H) at (2, -0.5) {\begin{tikzpicture}[scale=0.4, transform shape]
		\node[circle,draw,minimum size=18pt,ultra thick] (Aa) at (0,0) {};	
		\end{tikzpicture}};
	\node (G) at (0, -0.5) {\begin{tikzpicture}[scale=0.4, transform shape]
		\node[circle,draw,minimum size=18pt,ultra thick] (Aa) at (0,0) {};	
		\end{tikzpicture}};
	\draw[->] (K) to node {} (L);
	\draw[->] (K) to node {} (R);
	\draw[->] (K) to node {} (D);
	\draw[->] (D) to node {} (H);
	\draw[->] (D) to node {} (G);
	\draw[->] (R) to node {} (H);
	\draw[->] (L) to node {} (G);
    \end{tikzpicture}
    \caption{Non-natural double-pushout}
    \label{fig:NPOPO}
\end{figure}

The natural double-pushout construction such that we have natural double-pushout is described in \cite{Bak15a,Campbell19}, that are:
 \vspace{-\topsep}\begin{enumerate}
    \item To obtain $D$, remove all nodes and edges in $g(L−K)$ from $G$. For all $v\in V_K$ with $l_K(v)=\perp$, define $l_D(g_V(v))=\perp$. Also, define $p_D(g_V(v))=\perp$ for all $v\in V_K$ where $p_K(v)=\perp$.
    \item Add all nodes and edges, with their labels and rootedness, from $R − K$ to D. For $e\in E_R −E_K$, $s_H(e) = s_R(e)$ if $s_R(E) \in V_R −V_K$, otherwise $s_H(e) = g_V (s_R(e))$. Targets are defined analogously.
    \item For all $v \in V_K$ with $l_K (v) = \perp$, define $l_H (g_V (v)) = l_R(v)$. Also, for the injective morphism $R\rightarrow H$ and $v\in V_K$ where $p_K(v)=\perp,$ define $p_H(g^*_V(v)) = p_R(v)$. The resulting graph is $H$.
\end{enumerate}

\noindent Direct derivations transform a host graph via a rule whose the left and right-hand graph are totally labelled host graphs. However, a conditional rule schema contains a condition, and its left or right-hand graph may not be a host graph. Hence, we need some additional requirements for the application of a conditional rule schema on a host graph.

\begin{definition}[Conditional rule schema application]\label{def:ruleapp}\normalfont
Given a conditional rule schema $r=\lkr$, and host graphs $G, H$. $G$ directly derives $r$, denoted by $G\Rightarrow_{r,g} H$ (or $G \Rightarrow_r H$), if there exists a premorphism $g: L \rightarrow G$ and a label assignment $\alpha_L$ such that:
\vspace{-\topsep}\begin{enumerate}
    \item[(i)] $g:L^\alpha\rightarrow G$ is an injective morphism, 
    \item[(ii)] $\Gamma^{g,\alpha}$ is true,
    \item[(iii)] $G \Rightarrow_{r^{g,\alpha},g} H$.\qed
\end{enumerate}\end{definition}

A rule schema $r$ (without condition) can be considered as a conditional rule schema $\tuple{r,\mtt{true}}$, which means in its application, the point (ii) in the definition above is a valid statement for every unconditional rule schema $r$.

Syntactically, a conditional rule schema in GP\,2 is written as follows:
\begin{center}\small
\begin{tabular}{lll}
RuleDecl & ::= & RuleId `(' [ VarList \{`:' VarList\} ] `;' `)'\\
&& Graphs Interface [\ttt{where} Condition]\\
VarList & ::= & Variable \{`,' Variable\} `:' Type\\
Graphs & ::= & `[' Graph `]' `$=>$' `[' Graph `]'\\
Interface & ::= & \ttt{interface} `=' `\{' [NodeId \{`,' NodeId\}]`\}'\\
Type & ::= & $\mtt{int~\mid~char~\mid~string~\mid~atom~\mid~list}$
\end{tabular}
\end{center}

\noindent where Condition is the set of GP\,2 rule conditions as defined in Definition \ref{def:CRS} and Variable represents variables of all types. Graph represent rule graphs, where bidirectional edges may exist. Bidirectional edges and $\any$-marks are allowed in the right-hand graph if there exist preserved counterpart item in the left-hand graph. 

A rule schema with bidirectional edges can be considered as a set of rules with all possible direction of the edges. For example, a rule schema with one bidirectional edge between node $u$ and $v$ can be considered as two rule schemata, where one rule schema has an edge from $u$ to $v$ while the other has an edge from $v$ to $u$.

\subsection{Syntax and operational semantics of graph programs}
A GP\,2 graph program consists of a list of three declaration types: rule declaration, main procedure declaration, and other procedure declaration. A main declaration is where the program starts from so that there is only one main declaration allowed in the program, and it consists of a sequence of commands. For more details on the abstract syntax of GP 2 programs, see \figurename$~$\ref{fig:GP2syntax}, where RuleId and ProcId are identifiers that start with lower case and upper case respectively.

\begin{figure}[!h]
\centering\small
\begin{tabular}{lll}
Prog & ::= & Decl \{Decl\}\\
Decl & ::= & MainDecl $\mid$ ProcDecl $\mid$ RuleDecl\\
MainDecl & ::= & $\mtt{Main}$ `=' ComSeq\\
ProcDecl & ::= & ProcId `=' Comseq\\
ComSeq & ::= & Com \{`;' Com\}\\
Com & ::= & RuleSetCall $\mid$ ProcCall\\
&& $\mid~\mtt{if}$ ComSeq $\mtt{ then }$ ComSeq [$\mtt{else }$ ComSeq]\\
&& $\mid~\mtt{try }$ ComSeq [$\mtt{ then }$ ComSeq] [$\mtt{else }$ ComSeq]\\
&& $\mid~$ComSeq `!'\\
&& $\mid~$ComSeq $\mtt{ or }$ ComSeq\\
&& $\mid~$`(' ComSeq `)'\\
&& $\mid~\mtt{break}~\mid~\mtt{skip}~\mid~\mtt{fail}$\\
RuleSetCall & ::= & RuleId $\mid$ `\{' [RuleId \{ `,' RuleId\}] `\}'\\
ProcCall & ::= & ProcId
\end{tabular}
\caption{Abstract syntax of GP 2 programs}
\label{fig:GP2syntax}
\end{figure}

Other than executing a set of rule schemata, a program can also execute some commands sequentially by using `;'. There also exist \texttt{if} and \texttt{try} as branching commands, where the program will execute command after \texttt{then} when the condition is satisfied or \texttt{else} if the condition is not satisfied. However, as we can see in the syntax of GP 2 in \figurename~\ref{fig:GP2syntax}, we have command sequence as the condition of branching commands instead of a Boolean expression. Here, we say that the condition is satisfied when the execution of command in the condition terminates with a result graph (that is, it neither diverges nor fails) and it is not satisfied if the execution yields failure. 

The difference between \texttt{if} and \texttt{try} lies in the host graph that is used after the evaluation of conditions. For \texttt{if}, the program will use the host graph that is used before the examination of the condition. Otherwise for \texttt{try}, if the condition is satisfied, then the program will execute the graph obtained from applying the condition or the previous graph if the condition is not satisfied. Other than branching commands, there is also a loop command `!' (read as ``as long as possible"). It executes the loop-body as long as the command does not yield failure. Like a loop in other programming languages, a !-construct can result in non-termination of a program.

Configurations in GP 2 represents a program state of program execution in any stage. Configurations are given by (\texttt{ComSeq}$\times\mathcal{G}(\mathcal{L})$)$~\cup~\mathcal{G}(\mathcal{L})~\cup~$(\texttt{fail})), where $\mathcal{G}(\mathbb{L})$ consists of all host graphs. This means that a configuration consists either of unfinished computations, represented by command sequence together with current graph; only a graph, which means all commands have been executed; or the special element \texttt{fail} that represents a failure state. A small step transition relation $\rightarrow$ on configuration is inductively defined by inference rules shown in \figurename~\ref{fig:infRule-core} and \figurename~\ref{fig:infRule-derived} where $\mathcal{R}$ is a rule set call; $C,P,P'$, and $Q$ are command sequences; and $G$ and $H$ are host graphs.

\begin{figure}[!h]
\centering
\scalebox{0.85}{
%\makebox[.7\textwidth]{
\begin{tabular}{lll}
~[Call$_1$]$\displaystyle\frac{G\Rightarrow_RH}{\langle R,G\rangle\rightarrow H}$&$~~~~~~~$&[Call$_2$]$\displaystyle\frac{G\nRightarrow_R}{\langle R,G\rangle\rightarrow \texttt{fail}}$\\~\\
~[Seq$_1$]$\displaystyle\frac{\langle P,G\rangle\rightarrow\langle P',H\rangle}{\langle P;Q,G\rangle\rightarrow\langle P';Q,H\rangle}$&&[Seq$_2$]$\displaystyle\frac{\langle P,G\rangle\rightarrow H}{\langle P;Q,G\rangle\rightarrow\langle Q,H\rangle}$\\~\\
~[Seq$_3$]$\displaystyle\frac{\langle P,G\rangle\rightarrow\texttt{fail}}{\langle P;Q,G\rangle\rightarrow\texttt{fail}}$&&[Break]$\displaystyle\frac{}{\langle \texttt{break};P,G\rangle\rightarrow\langle \texttt{break},G\rangle}$\\~\\
~[If$_1$]$\displaystyle\frac{\langle C,G\rangle\rightarrow ^+ H}{\langle \texttt{if }C\texttt{ then }P\texttt{ else }Q,G\rangle\rightarrow\langle P,G\rangle}$&&[If$_2$]$\displaystyle\frac{\langle C,G\rangle\rightarrow ^+ \texttt{fail}}{\langle \texttt{if }C\texttt{ then }P\texttt{ else }Q,G\rangle\rightarrow\langle Q,G\rangle}$\\~\\
~[Try$_1$]$\displaystyle\frac{\langle C,G\rangle\rightarrow ^+ H}{\langle \texttt{try }C\texttt{ then }P\texttt{ else }Q,G\rangle\rightarrow\langle P,H\rangle}$&&[Try$_2$]$\displaystyle\frac{\langle C,G\rangle\rightarrow ^+ \texttt{fail}}{\langle \texttt{try }C\texttt{ then }P\texttt{ else }Q,G\rangle\rightarrow\langle Q,G\rangle}$\\~\\
~[Loop$_1$]$\displaystyle\frac{\langle P,G\rangle\rightarrow ^+H}{\langle P!,G\rangle\rightarrow\langle P!,H\rangle}$&&[Loop$_2$]$\displaystyle\frac{\langle P,G\rangle\rightarrow ^+\texttt{fail}}{\langle P!,G\rangle\rightarrow H}$\\~\\
~[Loop$_3$]$\displaystyle\frac{\langle P,G\rangle\rightarrow ^*\langle\texttt{break}, H\rangle}{\langle P!,G\rangle\rightarrow H}$&&
\end{tabular}}
\caption{Inference rules for core commands \cite{Plump12a}}
\label{fig:infRule-core}
\end{figure}

\begin{figure}[!h]
\centering
\scalebox{0.85}{
\def\arraystretch{2}\tabcolsep=1.5pt
\begin{tabular}{lll}
~[Or$_1$] $\langle P\texttt{ or }Q,G\rangle\rightarrow\langle P,G\rangle$&~~~~~&[Or$_2$] $\langle P\texttt{ or }Q,G\rangle\rightarrow\langle Q,G\rangle$\\
~[Skip$_1$] $\langle \texttt{skip},G\rangle\rightarrow G$&&[Fail] $\langle \texttt{fail},G\rangle\rightarrow\texttt{fail}$\\
\multicolumn{3}{l}{~[If$_3$] $\langle\texttt{if }C\texttt{ then }P,G\rangle\rightarrow\langle\texttt{if }C\texttt{ then }P\texttt{ else skip},G\rangle$}\\
\multicolumn{3}{l}{~[Try$_3$] $\langle\texttt{try }C\texttt{ then }P,G\rangle\rightarrow\langle\texttt{try }C\texttt{ then }P\texttt{ else skip},G\rangle$}\\
\multicolumn{3}{l}{~[Try$_4$] $\langle\texttt{try }C\texttt{ else }Q,G\rangle\rightarrow\langle\texttt{try }C\texttt{ then skip else }Q,G\rangle$}\\
\multicolumn{3}{l}{~[Try$_4$] $\langle\texttt{try }C,G\rangle\rightarrow\langle\texttt{try }C\texttt{ then skip else skip},G\rangle$}

\end{tabular}
}
\caption{Inference rules for derived commands \cite{Plump12a}}
\label{fig:infRule-derived}
\end{figure}

The semantics of programs is given by the semantic function $\llbracket\_\rrbracket$ that maps an input graph $G$ to the set of all possible results of executing a program $P$ on $G$. The application of $\llbracket P\rrbracket$ to $G$ is written $\llbracket P\rrbracket G$. The result set may contain proper results in the form of graphs or the special values \textit{fail} and $\perp$. The value \texttt{fail} indicates a failed program run while $\perp$ indicates a run that does not terminate or gets stuck. Program $P$ can diverge from $G$ if there is an infinite sequence $\langle P,G\rangle\rightarrow\langle P_1, G_1\rangle\rightarrow \langle P_2, G_2\rangle\rightarrow\ldots$. Also, $P$ can get stuck from $G$ if there is a terminal configuration $\langle Q,H\rangle$ such that $\langle P,G\rangle\rightarrow^*\langle Q,H\rangle$.

\Def{Semantic function \cite{Plump12a}}{def:semanticfunc}{
The semantic function $\llbracket\_\rrbracket$: ComSeq $\rightarrow(\mathcal{G}(\mathbb{L})\rightarrow2^{\mathcal{G}(\mathbb{L})\cup\{fail,\bot\}})$ is defined by
\[\footnotesize{\llbracket P\rrbracket G=\{X\in (\mathcal{G}(\mathbb{L})\cup\{\mrm{fail}\})|\langle P,G\rangle\rightarrow^+X\}\cup\{\bot\mid P \text{ can diverge or get stuck from } G\}.}\] 
}

A program $C$ can get stuck only in two situations, that is either $P$ contains a command \texttt{if $A$ then $P$ else $Q$} or \texttt{try $A$ then $P$ else $Q$} such that $A$ can diverge from a host graph $G$, or 
$P$ contains a loop $B!$ whose body $B$ can diverge from a host graph $G$. The evaluation of such commands gets stuck because none of the inference rules for if-then-else, try-then-else or looping is applicable. Getting stuck always signals some form of divergence.

We sometimes need to prove that a property holds for all graph programs. For this, we use structural induction on graph programs by having a general version of graph programs. That is, ignoring the context condition of the command $\mtt{break}$ such that it can appear outside a loop. However, when $\mtt{break}$ occur outside the context condition, we treat it as a $\mtt{skip}$.

\begin{definition}[Structural induction on graph programs]\label{def:gpinduction}\normalfont
Proving that a property \textit{Prop} holds for all graph programs by induction, is done by:\\
\begin{tabular}{lrp{11.3cm}}
    \multicolumn{3}{l}{Base case.}\\
    ~~~~~& \multicolumn{2}{p{11.7cm}}{
    Show that \textit{Prop} holds for $\mathcal{R}=\{r_1,\ldots,r_n\}$, where $n\geq 0$}\\
    \multicolumn{3}{l}{Induction case.}\\
    ~~~~~& \multicolumn{2}{p{11.7cm}}{
    Assuming \textit{Prop} holds for graph programs $C, P,$ and $Q$, show that \textit{Prop} also holds for:}\\
    & 1. & $P;Q$,\\
    & 2. & $\mtt{if\,}C\mtt{\,then\,}P\mtt{\,else\,}Q$,\\
    & 3. & $\mtt{try\,}C\mtt{\,then\,}P\mtt{\,else\,}Q$, and\\
    & 4. & $P!$.\qed
\end{tabular}
\end{definition}

The commands $\mtt{fail}$ and $\mtt{skip}$ can be considered (respectively) as a call of the ruleset $\mathcal{R}=\{\}$ and a call of the rule schema where the left and right-hand graphs are the empty graphs. Also, the command $P \mtt{~or~} Q$ can be replaced with the program $\mtt{if\, (Delete!;\,\{nothing, add\};\, zero)\, then\,} P\mtt{\, else\,} Q$ where $\mtt{Delete}$ is a set of rule schemata that deletes nodes and edges, including loops. $\mtt{nothing}$ is the rule schema where the left and right-hand graphs are the empty graphs, $\mtt{add}$ is the rule schema where the left-hand graph is the empty graph and the right- hand graph is a single 0-labelled unmarked and unrooted node, and $\mtt{zero}$ is a rule schema that matches with a 0-labelled unmarked and unrooted node.

As mentioned before, the execution of a graph program may yield a proper graph, failure, or diverge/get stuck. The latter only may happen when a loop exists in the program. In some cases, we may want to not considering the possibility of diverging or getting stuck such that we only consider loop-free graph programs. To show that a property holds for a loop-free program, we also introduce structural induction on loop-free programs.

\begin{definition}[Structural induction on loop-free programs]\label{def:noloopinduction}\normalfont
Proving that a property \textit{Prop} holds for all loop-free programs by induction, is done by:\\
\begin{tabular}{lrp{11.3cm}}
    \multicolumn{3}{l}{Base case.}\\
    ~~~~~& \multicolumn{2}{p{11.7cm}}{
    Show that \textit{Prop} holds for $\mathcal{R}=\{r_1,\ldots,r_n\}$, where $n\geq 0$}\\
    \multicolumn{3}{l}{Induction case.}\\
    ~~~~~& \multicolumn{2}{p{11.7cm}}{
    Assuming \textit{Prop} holds for loop-free programs $C, P,$ and $Q$, show that \textit{Prop} also holds for:}\\
    & 1. & $P\mtt{~or~}Q$,\\
    & 2. & $P;Q$,\\
    & 3. & $\mtt{if\,}C\mtt{\,then\,}P\mtt{\,else\,}Q$, and\\
    & 4. & $\mtt{try\,}C\mtt{\,then\,}P\mtt{\,else\,}Q$.\qed
\end{tabular}
\end{definition}

%% file: input/3_FOforGP.tex
In this section, we define first-order formulas which are able to express properties of GP\,2 graphs. Also, we define structural induction on the first-order formulas and replacement graphs which later can be used to show satisfaction of a first-order formula in a morphism.

\subsection{Syntax}

Our first-order (FO) formulas have logical connectives, variables, constants, also auxiliary, predicate, and function symbols.

\Def{Alphabet of a first-order formula}{def:FOLabc}{
The alphabet of a first-order formula consists of the following sets of symbols:
 \vspace{-\topsep}\begin{enumerate}
    \item Logical connectives: $\wedge$ (and), $\vee$ (or), $\neg$ (not), $\mrm{true}$, $\mrm{false}$, equality symbols $=,\neq,>,\geq,<,\leq,$ and quantifiers $\E{V},\E{E},\E{L}$ for nodes, edges, and labels respectively.
    \item Variables: a countably infinite set of lowercase letters.
    % \item Auxiliary symbols: ``('',``)'',``['',``]'',``$\mid$'',``\{'', and ``\}''.
     \item Predicate symbols: $\mrm{int, char, string, atom, edge, root}$.
     \item Function symbols: $\sou$ (source), $\tar$ (target), $\lV$ (node label), $\lE$ (edge label), $\mV$ (node mark), $\mE$ (edge mark), $\mrm{indeg}$, $\mrm{outdeg}$, $\mrm{length}$, integer operators $+,-,*,/$, label operator $:$, and string operator $.$ (concatenation).
    \item Constants: all elements in $\mathbb{L}$, $\mrm{empty}$, $\mrm{none, red, green, blue,}$ $\mrm{green, dashed, grey,}$\\and $\mrm{ any}$.
   
\end{enumerate}
}

Here we differentiate variables in seven kinds, which are first-order variables (single variables) for nodes, edges, and labels where labels are typed as in GP\,2. Table \ref{tab:domkind} shows the seven kinds of variables and their domains in a graph $G$. Note that we assume that node, edge, and list variables are pairwise distinct, while list, atom, integer, string, and character variables have hierarchy based on their domain.

\begin{table}[]
    \centering
    \caption{Kinds of variables and their domain on a graph $G$}
    \begin{tabular}{|c|c|} 
    \hline
        \textbf{kind of variables} & \textbf{domain} \\\hline
        NodeVar & $V_G$ \\
        EdgeVar & $E_G$ \\
        ListVar & $(\mathbb{Z}\cup(\text{Char})^*)^*$ \\
        AtomVar & $\mathbb{Z}\cup\text{Char}^*$ \\
        IntVar & $\mathbb{Z}$ \\
        StringVar & $\text{Char}^*$ \\
        CharVar & Char \\
        \hline
    \end{tabular}
    \label{tab:domkind}
\end{table}

The syntax of FO formulas is given by the grammar of Figure \ref{fig:mso}. In the syntax, NodeVar and EdgeVar represent disjoint sets of first-order node and edge variables, respectively. We use ListVar, AtomVar, IntVar, StringVar, and CharVar for sets of first-order label variables of type $\mrm{list, atom, int, string}$, and $\mrm{char}$ respectively. The nonterminals Character and Digit in the syntax represent the fixed character set of GP\,2, and the digit set $\{0,\ldots,9\}$ respectively.

\begin{figure}
    \centering\footnotesize
    \begin{tabular}{lcl}
      Formula & ::= & $\mrm{true}~\mid~\mrm{false}~\mid$ Cond $\mid$ Equal\\
        && $\mid$ Formula~(`$\mrm{\wedge}$' $\mid$ `$\mrm{\vee}$')~Formula $\mid$ `$\neg$'Formula $\mid$ `('Formula`)'\\
        && $\mid$ `$\exists_\mathtt{V}$' (NodeVar)~`('Formula`)' \\
        && $\mid `\exists_\mathtt{E}$'~(EdgeVar)~`('Formula`)'       \\
        && $\mid$ `$\exists_\mathtt{L}$' (ListVar)~`('Formula`)' \\
      Number & ::= & Digit~\{Digit\}\\
      Cond & ::= & ($\mrm{int \mid char \mid string \mid atom}$) `('Var`)' \\
        && $\mid$ Lst (`$\mrm{=}$' $\mid$ `$\mrm{\neq}$') Lst $\mid$ Int (`$\mrm{>}$' $\mid$ `$\mrm{>=}$' $\mid$ `$\mrm{<}$' $\mid$ `$\mrm{<=}$') Int\\
        && $\mid$ $\mrm{edge}$ `(' Node `,' Node [`,' Lst] [`,' EMark] `)' $\mid$ $\mrm{root}$ `(' Node `)'\\
      Var & ::= & ListVar $\mid$ AtomVar $\mid$ IntVar $\mid$ StringVar $\mid$ CharVar\\
      Lst & ::= & $\mrm{empty}$ $\mid$ Atm $\mid$ Lst \lq:' Lst $\mid$ ListVar $\mid$ $\mrm{l_V}$ `('Node`)' $\mid$ $\mrm{l_E}$ `('EdgeVar`)'  \\
      Atm & ::= & Int $\mid$ String $\mid$ AtomVar \\
      Int & ::= & [\lq-']~Number~$\mid$~\lq('Int\lq)'  $\mid$ IntVar $\mid$ Int (\lq+' $\mid$ \lq-' $\mid$ \lq*' $\mid$ \lq/') Int \\
        && $\mid$ ($\mrm{indeg}$ $\mid$ $\mrm{outdeg}$) \lq('Node\lq)' $\mid$ $\mrm{length}$ \lq('AtomVar $\mid$ StringVar $\mid$ ListVar\lq)' \\
      String & ::= & ` `` ' {Character} ` " ' $\mid$ CharVar $\mid$ StringVar $\mid$ String `.' String \\
      Node & ::= & NodeVar $\mid$ ($\mrm{s}~\mid \mrm{t}$) `(' EdgeVar`)'\\
      EMark & ::= & $\mrm{none~\mid~red~\mid~green~\mid~blue~\mid~dashed~\mid~any}$\\
      VMark & ::= & $\mrm{none~\mid~red~\mid~blue~\mid~green~\mid~grey~\mid~any}$  \\
      Equal & ::= & Node ('$\mrm{=}$' $\mid$ `$\mrm{\neq}$') Node $\mid$ EdgeVar ('$\mrm{=}$' $\mid$ `$\mrm{\neq}$') EdgeVar\\
        && $\mid$ Lst ('$\mrm{=}$' $\mid$ `$\mrm{\neq}$') Lst $\mid$ $\mrm{m_V}$`('Node`)' ('$\mrm{=}$' $\mid$ `$\mrm{\neq}$') VMark \\
        && $\mid$ $\mrm{m_E}$`('EdgeVar`)' ('$\mrm{=}$' $\mid$ `$\mrm{\neq}$') EMark
    \end{tabular}
    \caption{Syntax of first-order formulas}
    \label{fig:mso}
\end{figure}

The quantifiers $\E{V}, \E{E},$ and $\E{L}$ in the grammar are reserved for variables of nodes, edges, and labels respectively. The function symbols $\mrm{indeg, outdeg}$ and $\mrm{length}$ work similar with functions with the same names in GP\,2 rule schema conditions. In addition, we also have unary functions $\mrm{s, t, l_V, l_E, m_V,}$ and ${m_{E}}$. These functions return the mapping result of the argument based on their functions as defined in Definition \ref{def:graphs}. For example, the function $\mrm{s}$ takes an edge variable in the argument and returns the node that is the source of the edge represented by the variable in a host graph. The predicate $\mrm{edge}$ expresses the existence of an edge between two nodes. The predicates $\mrm{int, char, string, atom}$ are typing predicates to specify the type of the variable in their argument. We have the predicate $\mrm{root}$ to express rootedness of a node. For brevity, we sometimes write $\mrm{c\implies d}$ for $\mrm{\neg c\vee d}$, $\mrm{c\iff d}$ for $\mrm{(\neg c\vee d)\wedge(c\vee\neg d)}$, $\mrm{\forall_Vx(c)}$ for $\mrm{\neg\E{V}x(\neg c)}$ and $\mrm{\E{V}x_1,\ldots,x_n(c)}$ for $\mrm{\E{V}x_1(\E{V}x_2(...\E{V}x_n(c)\ldots))}$ (also for edge and label quantifiers). Also, we define 'term' as the set of variables, constants, and functions in first-order formulas.

\subsection{Structural induction on first-order formulas}
To prove properties related to our first-order formulas, we classify first-order formulas into eight cases, based on their forms. To prove that some properties hold for these cases, we define structural induction on first-order formulas. Three cases are defined as base cases since they are formed from terms while the others are defined as inductive cases. As mentioned before, terms can exist as a variable, a constant, or a function (or an operators). For terms, we also define a structural induction on terms with variables and constants are its base cases.

\begin{definition}[Structural induction on terms]\label{def:inductionterms}\normalfont
Given a property \textit{Prop}. Proving that \textit{Prop} holds for all terms by \textit{structural induction on terms} is done by:
\begin{itemize}[nosep]
    \item Base case.\\
    Show that \textit{Prop} holds for all nodes, edges, and lists that may be represented by variables and constants.
    \item Inductive case.\\
    Assuming that \textit{Prop} holds for lists represented by terms $\mrm{x_1, x_2}$, integers represented by terms $\mrm{i_1, i_2}$, strings represented by terms $\mrm{s_1, s_2}$, a node represented by term $\mrm{v}$, and an edge represented by term $\mrm{e}$, show that \textit{Prop} also holds for:
     \vspace{-\topsep}\begin{enumerate}
        \item integers represented by $\mrm{length(x_1),}$ and $\mrm{i_1\oplus i_2}$ for $\oplus\in\{_,-,*,/\}$
        \item lists represented by $\mrm{\lE(e_1)}$ and  $\mrm{\lV(v_1)}$
        \item marks represented by $\mrm{\mE(e_1)}$ and $\mrm{\mV(v_1)}$
        \item strings represented by $\mrm{s_1.s_2}$\qed
    \end{enumerate}
\end{itemize}
\end{definition}

For simplicity, we do not consider a FO formula in the form $(c)$ as it is equivalent to FO formula $c$. We also do not include the predicate $\mrm{edge}$ because we can express it as $\mrm{\E{E}z(s(z)=x\wedge t(z)=y)}$. The optional arguments list and mark of the predicate $\mrm{edge}$ can be conjunct inside the quantifier, e.g. the predicate $\mrm{edge(x,y,5,\mrm{none})}$ can be expressed as
$\mrm{\E{E}z(s(z)=x\wedge t(z)=y\wedge l_E(z)=5\wedge m_E{z}}$\\$\mrm{=none})$.

\begin{definition}[Structural induction on first-order formulas]\label{def:inductionFO}\normalfont
Given a property \textit{Prop}. Proving that \textit{Prop} holds for all FO formulas by \textit{structural induction on FO formulas} is done by:
\begin{itemize}[nosep]
    \item Base case.\\
    Show that \textit{Prop} holds for:
     \vspace{-\topsep}\begin{enumerate}
        \item the formulas $\mrm{true}$ and $\mrm{false}$
        \item predicates $\mrm{int(z), char(z), string(z), atom(z)}$ for a list variable $z$, and $\mrm{root(y)}$ for a term $y$ representing a node
        \item Boolean operations $\mrm{x_1=x_2}$ and $\mrm{x_1\neq x_2}$ where both $x_1, x_2$ are terms representing nodes, edges, or lists, also $\mrm{y_1\ominus y_2},$ for terms $y_1, y_2$ representing integers and $\ominus\in\{=,\neq,<,\leq,>,\geq\}$
    \end{enumerate}
    \item Inductive case.\\
    Assuming that \textit{Prop} holds for FO formulas $c_1, c_2$, show that \textit{Prop} also holds for FO formulas $c_1\wedge c_2$, $c_1\vee c_2$, $\neg c_1$, $\E{V}x(c_1)$, $\E{E}x(c_1)$, and $\E{L}x(c_1)$.
\end{itemize}
\end{definition}

\subsection{Satisfaction of a first-order formula}
The satisfaction of a FO formula $c$ in a host graph $G$ relies on assignments. A formula assignment of $c$ on $G$ is defined in Definition \ref{def:assignment}. Informally, a formula assignment is a function that maps free variables to constants in their own domain. When we have an assignment for a FO formula on a graph, we can check the satisfaction of the FO formula. The satisfaction of a FO formula on a graph is then defined in Definition \ref{def:satisfaction}.

\Def{Formula assignment}{def:assignment}{
Let $c$ be a FO formula, $X$ and $Y$ be the set of free node and edge variables in $c$ respectively, and $Z$ be the set of free list variables in $c$. For a free variable $x$, dom$(x)$ denotes the domain of variable's kind associated with $x$ as in \tablename~\ref{tab:domkind}. A \emph{formula assignment} of $c$ on a host graph $G$ is a tuple $\alpha=\tuple{\alpha_G,\alpha_\mathbb{L}}$  of functions $\alpha_G=\tuple{\alpha_V:X\rightarrow V_G, \alpha_E:Y\rightarrow E_G}$ and  $\alpha_\mathbb{L}=Z\rightarrow\mathbb{L}$ such that for each free variable $x$, $\alpha(x)\in$ dom$(x)$. We then denote by $c^{\alpha}$ the FO formula $c$ after replacement of each term $\mrm{y}$ to $\mrm{y}^{\alpha}$, where $\mrm{y}^{\alpha}$ is defined inductively:
 \vspace{-\topsep}\begin{enumerate}
    \item If $\mrm{y}$ is a free variable, $\mrm{y}^{\alpha}=\alpha(\mrm{x})$;
    \item If $\mrm{y}$ is a constant, $\mrm{y}^{\alpha}=y$;
    \item If $y=\mrm{length(x)}$ for some variable $\mrm{x}$, $y^{\alpha}$ returns the number of characters in $\mrm{x}^{\alpha}$ if $\mrm{x}$ is a string variable, 1 if $\mrm{x}$ is an integer variable, or the number of atoms in $\mrm{x}^{\alpha}$ if $\mrm{x}$ is a list variable;
    \item If $y$ is $\mrm{s}(x), \mrm{t}(x), \mrm{l_E}(x), \mrm{m_E}(x), \mrm{l_V}(x), \mrm{m_V}(x), \mrm{indeg}(x)$, or $\mrm{outdeg}(x)$ for some term $x$, $y^{\alpha}$ is $s_G(x^{\alpha})$, $t_G(x^{\alpha})$, $\lst^E_G(x^{\alpha}), \mrk^E_G(x^{\alpha}), \lst^V_G(x^{\alpha}), \mrk^V_G(x^{\alpha}),$ indegree of $x^{\alpha}$ in $G$ , or outdegree of $x^{\alpha}$ in $G$, respectively;
    \item If $y=x_1\oplus x_2$ for $\oplus\in\{+,-,*,/\}$ and terms $x_1, x_2$ represented integers, $y^{\alpha}=x_1^{\alpha}\oplus_\mathbb{Z} x_2^{\alpha}$;
    \item If $y=x_1.x_2$ for some terms $x_1, x_2$ represented strings, $y^{\alpha}$ is string concatenation $x_1^{\alpha}$ and $x_2^{\alpha}$;
    \item If $y=x_1:x_2$ for some terms $x_1, x_2$ represented lists, $y^{\alpha}$ is list concatenation $x_1^{\alpha}$ and $x_2^{\alpha}$;
\end{enumerate}
}

\begin{remark}

\end{remark}

\Def{Satisfaction}{def:satisfaction}{
Given a graph $G$ and a first-order formula $c$. $G$ satisfies $c$, written $G\Sat c$, if there exists an assignment $\alpha$ such that $c^{\alpha}$ is true in $G$ (denotes by $G\models^\alpha c$), that is, for each Boolean sub-expression $b^{\alpha}$ of $c^{\alpha}$, the value of $b^{\alpha}$ in $\mathbb{B}$ is inductively defined:
 \vspace{-\topsep}\begin{enumerate}
    \item If $b^{\alpha}=\mrm{true}$ (or $b=\mrm{false}$), then $b^{\alpha}$ is true (or false);
    \item If $b^{\alpha}=\mrm{int(x), char(x), string(x), atom(x)}$, or $\mrm{root(x)}$, $b^{\alpha}$ is true if only if $\mrm{x}^{\alpha}\in\mathbb{Z}, \mrm{x}^{\alpha}\in\text{Char}, \mrm{x}^{\alpha}\in\text{Char}^*, \mrm{x}^{\alpha}\in\mathbb{Z}\cup\text{Char}^*,$ or $p_G(\mrm{x}^{\alpha})=1$ respectively.
    \item If $b^{\alpha}$ has the form $t_1\otimes t_2$ where $\otimes\in\{\mrm{>,>=,<,<=}\}$ and $t_1,t_2\in\mathbb{Z}$, $b^{\alpha}$ is true if and only if $t_1\otimes_\mathbb{Z} t_2$ where $\otimes_\mathbb{Z}$ is the integer relation on $\mathbb{Z}$ represented by $\otimes$. Then if $b^{\alpha}$ has the form $t_1\ominus t_2$ where $\ominus\in\{=,\neq\}$ and $t_1,t_2\in V_G\cup E_G\cup\mathbb{L}\cup\mathbb{M}\{\mrm{any}\}$, $b^{\alpha}$ is true if and only if $t_1\ominus_\mathbb{B} t_2$ where $\ominus_\mathbb{B}$ is the Boolean relation represented by $\ominus$. Then for $t_1=\mrm{any}$, $b^{\alpha}$ is true if and only if $\mtt{blue}\ominus_\mathbb{B} t_2~\vee~\mtt{red}\ominus_\mathbb{B} t_2~\vee~\mtt{green}\ominus_\mathbb{B} t_2~\vee~\mtt{grey}\ominus_\mathbb{B} t_2~\vee~\mtt{dashed}\ominus_\mathbb{B} t_2$ is true (and analogously for $t_2=\mtt{any}$).
    \item If $b^{\alpha}$ has the form $b_1\oslash b_2$ where $\oslash\in\{\vee,\wedge\}$ and $b_1,b_2$ are Boolean expressions, $b^{\alpha}$ is true if and only if $b_1\oslash_\mathbb{B} b_2$ where $\oslash_\mathbb{B}$ is the Boolean operation on $\mathbb{B}$ represented by $\oslash$. 
    \item If the form of $b^{\alpha}$ is $\neg b_1$ where $b_1$ is a Boolean expression, $b^{\alpha}$ is true if and only if $b_1$ is false.
    \item If $b^{\alpha}$ has the form $\E{V}e_1(e_2)$ where $e_1$ is a first-order node variable and $e_2$ is a Boolean expression, $b^\alpha$ is true if and only if there exists $v\in V_G$ such that when we add $e_1\mapsto v$ to assignment $\alpha$, $e_2$ is true.
    \item If $b^{\alpha}$ has the form $\E{E}e_1(e_2)$ where $e_1$ is a first-order edge variable and $e_2$ is a Boolean expression, $b^{\alpha}$ is true if and only if there exists $e\in E_G$ such that when we add $e_1\mapsto e$ to assignment $\alpha$, $e_2$ is true.
    \item If $b^{\alpha}$ is in the form $\E{L}e_1(e_2)$ where $e_1$ is a first-order list variable and $e_2$ is a Boolean expression, $b^{\alpha}$ is true if and only if there exists $l\in \mathbb{L}$ such that when we add $e_1\mapsto l$  to assignment ${\alpha}$, $e_2$ is true.
\end{enumerate}
}

\subsection{First-order formulas in rule schema application}
FO formulas we define in this Section does not have a node or edge constant because we want to be able to check the satisfaction of a FO formula on any graph. However, in a rule schema application, we sometimes need to express the properties of the images of the match or comatch, which is dependent on the left-hand graph or right-hand graph. To be able to express properties of the images of a match or comatch, we need to allow some node and edge constants in FO formulas. Hence, we define a condition over a graph.

\Def{Conditions over a graph}{def:condover}{
Given a graph $G$. A \emph{condition over $G$} is obtained from a first-order formula by substituting node (or edge) identifiers in $G$ for free node (or edge) variables in the first-order formula.}

\begin{example}
Let $G$ and $H$ be graphs where $V_G=\{1,2\}$ and $V_H=\{1\}$.
 \vspace{-\topsep}\begin{enumerate}
    \item $c_1=\mrm{\E{E}x(s(x)=1)}$ is a condition over $G$, also over $H$
    \item $c_2=\mrm{\A{V}x(\mrm{edge}(x,1)\wedge\mrm{indeg}(x)=2)}$ is a condition over $G$, but not over $H$
\end{enumerate}
\end{example}

Checking if a graph satisfies a condition over a graph is similar with checking satisfaction of a FO formula in a graph. However, for a condition $c$ over a graph, the satisfaction of $c$ in a graph $G$ can be defined if only if $c$ is a condition over $G$.

With a condition over a graph, we can express properties of left and right-hand graphs with explicitly mentioning node/edge identifiers in the graphs. In graph program verification, we need to express the properties of the initial and output graph with respect to a given rule schema. In \cite{Poskitt13,HP09}, they express them by showing the satisfaction of a condition on a morphism. Here, we define a replacement graph $H$ of a host graph $G$ with respect to an injective morphism $g$, where $H$ is isomorphic to $G$ and there exists an inclusion from the domain of $g$ to $H$.

\Def{Replacement graph}{def:rho}{
Given an injective (pre)morphism $g:L\rightarrow G$ where $V_G\cap V_L=\{v_1,\ldots,v_n\}$ and $E_G\cap E_L=\{e_1,\ldots,e_m\}$. Let also $U=\{u_1,\ldots,u_n\}$ be a set of identifiers not in $V_L$ and $V_G$, and $W=\{w_1,\ldots,w_n\}$ be a set of identifiers not in $E_L$ and $E_G$. \emph{Graph replacement} $\rho_g(G)$ is is obtained from $G$ by renaming every item $g(i)$ to $i$ for $i\in V_G$ and $i\in E_G$, every $v_i$ to $u_i$ for $i=1,\ldots,n$, and every $e_i$ to $w_i$ for $i=1,\ldots,m$, such that $V_{\RG}=V_G-g(V_L)\cup V_L\cup U$ and $E_{\RG}=E_G-g(E_L)\cup E_L\cup W$.}

From the definition above, it is obvious that a host graph and its replacement graph are isomorphic. For a host graph $G$, a host graph $L$, and a morphism $g:L\rightarrow G$, it is also obvious that there exists an inclusion $f:L\rightarrow\RG$, because $g$ preserves identifiers, sources, targets, and labels of $L$.

\begin{example}
Given $g$, a morphism from $L$ to $G$ as follows:\\
\begin{minipage}{0.5\textwidth}
\begin{center}
\begin{tikzpicture}[scale=0.7, 
inner/.style={circle,draw,inner sep=2.5pt}]
	\node[label=below:\footnotesize{$L$}] (A) at (0,0) {
		\begin{tikzpicture}[scale=0.7, transform shape]
			\node[inner, label=below:$1$] (Aa) at (0,0) {$a$};
			\node[inner, label=below:$2$] (Ab) at (1.5,0) {$a$};
			\draw[-latex] (Aa) to [out=210,in=150, looseness=8] node[left] {3} (Aa);
			\draw[-latex] (Ab) to node[above] {$4$} (Aa);			
		\end{tikzpicture}};
	\node[] (B) at (2.25,0) {
		\begin{tikzpicture}[scale=0.7, transform shape]
			\node[] (Ba) at (0,0.25) {};	
			\node[] (Bb) at (1,0.25) {};
			\draw[->] (Ba) to node[above] {$g$} (Bb);
		\end{tikzpicture}};
	\node[label=below:\footnotesize{$G$}] (C) at (4.25,0) {
		\begin{tikzpicture}[scale=0.7, transform shape]
			\node[inner, label=above:$v_1$] (Ca) at (0,0.5) {$b$};
			\node[inner, label=below:$v_2$] (Cb) at (1.5,0) {$a$};
			\node[inner, label=below:$v_3$] (Cc) at (0,-0.5) {$b$};
			\draw[-latex] (Ca) to [out=210,in=150, looseness=8] node[left] {$e_1$} (Ca);
			\draw[-latex] (Cb) to node[above] {$e_2$} (Ca);
			\draw[-latex] (Cb) to node[below] {$e_3$} (Cc);
			\draw[-latex] (Cc) to [out=210,in=150, looseness=8] node[left] {$e_4$} (Cc);
		\end{tikzpicture}};
\end{tikzpicture}

\end{center}

\end{minipage}
\begin{minipage}{0.5\textwidth}\footnotesize
\[g=\left\langle g_V:\left\{\begin{array}{l}
    		1\mapsto v_3 \\
    		2\mapsto v_2
        \end{array}, 
     		g_E:\right\{ \begin{array}{l}
    		3\mapsto e_4 \\
    		4\mapsto e_3
        \end{array} \right\rangle\]
\end{minipage}

\noindent Then, $\rho_g(G)$ is the graph\begin{center}

    \begin{tikzpicture}[scale=0.7, 
inner/.style={circle,draw,inner sep=2.5pt}]
	\node[] (C) at (0,0) {
		\begin{tikzpicture}[scale=0.7, transform shape]
			\node[inner, label=above:$v_1$] (Ca) at (0,0.5) {$b$};
			\node[inner, label=below:$2$] (Cb) at (1.5,0) {$a$};
			\node[inner, label=below:$1$] (Cc) at (0,-0.5) {$a$};
			\draw[-latex] (Ca) to [out=210,in=150, looseness=8] node[left] {$e_1$} (Ca);
			\draw[-latex] (Cb) to node[above] {$e_2$} (Ca);
			\draw[-latex] (Cb) to node[below] {$4$} (Cc);
			\draw[-latex] (Cc) to [out=210,in=150, looseness=8] node[left] {$3$} (Cc);
		\end{tikzpicture}};
\end{tikzpicture}    
\end{center}
\end{example}

Now we have defined a condition over a graph to express properties of a host graph w.r.t the left-hand graph or right-hand graph. Let $\tuple{L\leftarrow K\rightarrow R}$ be a rule schema, and $ac_L$, $ac_R$ denote a condition over a rule graph $L$ and $R$ respectively. To associate $ac_L$ and $ac_R$ with the rule schema, we define a \textit{generalised rule schema}. Unlike a rule schema, a generalised rule schema consists of an unrestricted rule schema that allows both left and right application condition.

\Def{Generalised rule schema}{def:genrule}{
Given an unrestricted rule schema $r=\langle L\leftarrow K\rightarrow R\rangle$.
\emph{A generalised rule} is a tuple $w=\langle r,ac_L,ac_R\rangle$ where $ac_L$ is a condition over $L$ and $ac_R$ is a condition over $R$. We call $ac_L$ the left application condition and $ac_R$ the right application condition. The inverse of $w$, written $w^{-1}$, is then defined as the tuple $\langle r^{-1},ac_R,ac_L\rangle$ where $r^{-1}=\langle R\leftarrow K\rightarrow L\rangle$.
}

The application of a generalised rule schema is essentially the same as the application of a rule schema. But here, we also check the satisfaction of both $ac_L$ and $ac_R$ in the replacement of input graph $G$ and final graph $H$ by match and comatch respectively.

\Def{Application of generalised rule schema}{def:generalisedrPO}{
Given a generalised rule schema $w=\langle r, ac_L, ac_R\rangle$ with an unrestricted rule schema $r=\langle L\leftarrow K\rightarrow R\rangle$. There exists a direct derivation from $G$ to $H$ by $w$, written $G\Rightarrow_{w,g,g^*} H$ (or $G\Rightarrow_wH$) iff there exists premorphisms $g:L\rightarrow G$ and $g^*:R\rightarrow H$ and label assignments $\alpha$ for $L$ and $\beta$ for $R$ where $\beta(i)=\alpha(i)$ for every variable $i$ in $L$ such that $i$ is in $R$ and for every node/edge $i$ where $\mrk_L(i)=\mrk_R(i)=\mtt{any}$, such that:
 \vspace{-\topsep}\begin{enumerate}
\item[(i)] $g:L^\alpha\rightarrow G$ is an injective morphism
\item[(ii)] $g^*:R^\beta\rightarrow H$ is an injective morphism
\item[(iii)] $\rho_g(G)\Sat ac_L^\alpha$,
\item[(iv)] $\rho_{g^*}(H)\Sat ac_R^\beta$,
\item[(v)] $G \Rightarrow_{r^{\alpha,g},g} H$,
\end{enumerate}
where $G \Rightarrow_{r^{\alpha},g} H$ denotes the existence of natural pushouts (1) and (2) as in the diagram of \figurename~\ref{fig:ddergenrule}.}

\begin{figure}
    \centering
    \begin{tikzpicture}[scale=0.9, transform shape]
	\node (2) at (1, 0) {(1)};
	\node (a) at (2, 0.75) {$K$};
	\node (b) at (0, 0.75) {$L^\alpha$};
	\node (c) at (4, 0.75) {$R^\beta$};
	\node (d) at (2, -0.75) {$D$};
	\node (1) at (3, 0) {(2)};
	\node (e) at (4, -0.75) {$H$};
	\node (f) at (0, -0.75) {$G$};
	\node (g) at (-0.85, -0.75) {$\RG\cong$};
	\node (i) at (-2.1, -0.75) {$ac_L^\alpha\Dashv$};
	\node (h) at (4.95, -0.75) {$\cong\RH$};
	\node (j) at (6.25, -0.75) {$\Sat ac_R^\beta$};
	\draw[->] (b) to node[left] {$g$} (f);
	\draw[->] (b) [bend right=40] to node[left] {incl} (g);
	\draw[->] (c) to node[right] {$g^*$} (e);
	\draw[->] (c) [bend left=40]to node[right] {incl} (h);
    \draw[->] (a) to node {} (b);
	\draw[->] (a) to node {} (c);
	\draw[->] (a) to node {} (d);
	\draw[->] (d) to node {} (e);
	\draw[->] (d) to node {} (f);
	\draw[->] (c) to node {} (e);
	\draw[->] (b) to node {} (f);
\end{tikzpicture}
    \caption{Direct derivation for generalised rule schema}
    \label{fig:ddergenrule}
\end{figure}

Recall the application of conditional rule schema in Definition \ref{def:ruleapp}. The condition of the rule schema is clearly can be considered as the left-application condition of the rule schema. Since there is no right-application condition in a conditional rule schema, there is no requirement about the condition such that we can always consider $\mrm{true}$ as the right-application condition of a conditional rule schema.

\Def{Generalised version of a conditional rule schema}{def:rvee}{
Given a conditional rule schema $\tuple{r,\Gamma}$. The generalised version of $r$, denoted by $r^\vee$, is the generalised rule schema $r^\vee=\tuple{r,\Gamma^\vee,\mrm{true}}$ where $\Gamma^\vee$ is obtained from $\Gamma$ by replacing the notations $\mrm{not}, !=, \mrm{and, or}, \#$ with $\neg, \neq, \wedge, \vee, `,$' (comma symbol) respectively.
}

\Lemma{lemma:rvee}{
Given a conditional rule schema $\tuple{r,\Gamma}$ with $r=\tuple{L\leftarrow K\rightarrow R}$. Then for any host graphs $G,H$,
\[G\Rightarrow_r\,H \text{ if and only if } G\Rightarrow_{r^\vee}\!H.\]
}

\begin{proof}
(Only if). Recall the restrictions about variables and any-mark of a rule schema. It is obvious that every variable in $R$ is in $L$ and every node/edge with mark $\any$ in $R$ is marked $\any$ in $L$ as well. From Definition \ref{def:ruleapp}, we know that $G\Rightarrow_rH$ asserts the existence of $\alpha_L$ and premorphism $g:L\rightarrow G$ such that: 1) $g:L^\alpha\rightarrow G$ is an injective morphism, 2) $\Gamma^{\alpha,g}$ is true in $G$, and 3) $G\Rightarrow_{r^{\alpha,g},g}H$. From 3) and the variable restrictions mentioned above, it is obvious that there exists morphism $g^*:R^\alpha\rightarrow H$, and $\RH\Sat ac_R$ because all graphs satisfy $\mrm{true}$. Hence, (ii), (iv), and (v) of Definition \ref{def:generalisedrPO} are satisfied. Point 1) then asserts (i) of Definition \ref{def:generalisedrPO}. The fact that $\Gamma^{\alpha,g}$ is true in $G$ from point 2) is then asserts $\RG\Sat \Gamma^\vee$ because it is obvious that the change of symbols does not change the semantics of the condition. Moreover, $\RG$ is a replacement graph w.r.t. $g$ such that evaluating $\Gamma^{\alpha,g}$ in $G$ is the same as evaluating $\Gamma^\alpha$ in $\RG$.\\
(If). Similarly, from Definition \ref{def:generalisedrPO}, we know that $G\Rightarrow_{r^\vee}H$ asserts the existence of label assignment $\alpha$ for $L$ and premorphism $g:L\rightarrow G$ such that: 1) $g:L^\alpha\rightarrow G$ is an injective morphism; 2)$\RG\Sat\Gamma^\vee$; and 3)$G\Rightarrow_{r^{\alpha,g},g}H$. These obviously assert $G\Rightarrow_rH$ from Definition \ref{def:ruleapp} and the argument about $\Gamma^\vee$ above.
\end{proof}

\Remark{For morphism $g:L^\alpha\rightarrow G$, the semantics of $\Gamma$ in $G$ with respect to $g$ and $\Gamma^\vee$ in $\RG$ is identical. From here, $\Gamma$ also refers to $\Gamma^\vee$ when it obviously refers to a condition over $L$.}

\Lemma{lemma:inverse}{
Given a generalised rule schema $w=\langle r, ac_L, ac_R\rangle$ with an unrestricted rule schema $r=\langle L\leftarrow K\rightarrow R\rangle$ and label assignment $\alpha$ for $L$. Then for host graphs $G$ and $H$ with premorphisms $g:L\rightarrow G$ and $g^*:R\rightarrow H$,
\[G\Rightarrow_{w,g,g^*}\!H \text{ if and only if } H\Rightarrow_{w^{-1},g^*,g}\!G.\]
}

\begin{proof}$~$\\
(Only if.) From Definition \ref{def:generalisedrPO} we know that when $G\Rightarrow_{w,g,g^*}\!H$, it means that there exists label assignment $\alpha$ for $L$ and $\beta$ for $R$ where $\alpha(i)=\alpha(i)$ for every variable $i$ in $L$ such that $i$ is in $R$, and for every node/edge $i$ where $\mrk_L(i)=\mrk_R(i)=\mtt{any}$, such that $g:L^\alpha\rightarrow G$ and $g^*:R^\beta\rightarrow H$ are injective morphisms where
 \vspace{-\topsep}\begin{enumerate}
\item[(i)] $\rho_g(G)\Sat ac_L^\alpha$
\item[(ii)] $\rho_{g^*}(H)\Sat ac_R^\beta$
\item[(iii)] $G \Rightarrow_{r^{\alpha},g} H$.
\end{enumerate}
These are obviously defines direct derivation $H \Rightarrow_{(r^-1)^{g^*,\alpha},g^*}\!G$ such that $H\Rightarrow_{w^{-1},g^*,g}\!G$.\\
(If). We can apply the above proof analogously.\qed
\end{proof}

The application of a rule depends on the existence of morphisms. Showing the existence of a morphism $L\rightarrow G$ for host graphs $L, G$ can be done by checking the existence of the structure of $L$ in $G$. For this, we define a condition over a graph to specify the structure and labels of a graph.

\Def{Specifying a totally labelled graph}{def:spec}{
Given a totally labelled graph $L$ with $V_L=\{v_1,\ldots,v_n\}$ and $E_L=\{e_1,\ldots,e_m\}$. Let $X=\{x_1,\ldots,x_k\}$ be the set of all list variables in $L$, and Type$(x)$ for $x\in X$ is $\mrm{int(x)}$, $\mrm{char(x)}$, $\mrm{string(x)}$, $\mrm{atom(x)}$, or $\mrm{true}$ if $x$ is an integer, char, string, atom, or list variable respectively. Let also Root$_L(v)$ for $v\in V_L$ be a function such that Root$_L(v)=\mrm{root(v)}$ if $p_L(v)=1$, and Root$_L(v)=\mrm{\neg root(v)}$ otherwise. \emph{A specification of $L$}, denoted by Spec$(L)$, is the condition over $L$:
\[\mrm{\bigwedge_{i=1}^k} \text{Type}(x_i) \mrm{~\wedge~\bigwedge_{i=1}^n \lV(v_i)}=\lst_L(v_i)\mrm{~\wedge~\mV(v_i)=}\mrk_L(v_i)\mrm{~\wedge~} \text{Root}_L(v_i)\]
\[\mrm{\wedge~\bigwedge_{i=1}^m s(e_i)=}s_L(e_i)\mrm{~\wedge~ t(e_i)=}t_L(e_i)\mrm{~\wedge~\lE(e_i)}=\lst_L(e_i)\mrm{~\wedge~
\mE(e_i)=}\mrk_L(e_i)\]
}

Since morphisms require the preservation of sources, targets, labels, and rootedness, we need to explicitly state rootedness and label of each node, source and target of each edge. Also, since we also want to specify rule graphs, the type of each variable needs to explicitly stated as well. Note that we only specify totally labelled graphs so that the label and rootedness of a node are always defined.

\Ex{Specification of $L$}{ex:spec}{
Let us consider the graph $L$ below:
\begin{center}
    \begin{tikzpicture}[remember picture,
  inner/.style={circle,draw,minimum size=18pt},
  outer/.style={inner sep=2pt}, scale=0.9
  ]
  \node[outer] (A) at (0,0) {
  \begin{tikzpicture}[scale=0.8, transform shape]
		\node[inner, label=below:\tiny 1] (Aa) at (0,0) {$\mtt{a+b}$};	
		\node[inner, label=below:\tiny 2, fill=red!50] (Ab) at (1.5,0) {$\mtt{a}$};	
        \node[inner, label=below:\tiny 3, ultra thick] (Ac) at (-1.5,0) {$\mtt{b}$};	
		\draw[-latex, dashed] (Aa) to node[above] {$\mtt{d}$} (Ab);
		\draw[-latex] (Aa) to node[above] {$\mtt{7}$} (Ac);
		\end{tikzpicture}};\end{tikzpicture}
\end{center}
where the edge incident to 1 and 2 is edge $e1$ and the other one is edge $e2$, and $a,b,$ are integer variables while $d$ is a list variable. Then, Spec$(L)$ is the condition over $L$:
\begin{small}
\[\mrm{int(a)\wedge\,int(b)\wedge\, \lV(1)=a+b\wedge\, \lV(2)=a\wedge\, \lV(3)=b}\]
\[\mrm{\wedge\, \mV(1)=none\wedge\, \mV(2)=red\wedge\, \mV(3)=none\wedge\, \neg root(1)\wedge\, \neg root(2)\wedge\, root(3)}\]
\[\mrm{\wedge\, s(e1)=1\wedge\, t(e1)=2\wedge\, s(e2)=1\wedge\, t(e2)=3\wedge\, \lE(e1)=d\wedge\, \lE(e2)=7}\]
\[\mrm{\wedge\, \mE(e1)=dashed\wedge\, \mE(e2)=none}\]
\end{small}}

When a graph $G$ satisfying Spec$(L)$, it means $G$ has a subgraph $H$ with identical node and edge identifiers and with the same structure (sources, targets, and rootedness) as $L$. The labels of $H$ and $L$ should also be the same if both are host graphs, but not necessarily if at least one of them is a rule graph. However, if $G$ is a host graph satisfying Spec$(L)$, then there must exist label assignment $\alpha$ for $L$ such that Spec$(L^\alpha)$ is satisfied by $G$, yields to the existence of inclusion $L^\alpha\rightarrow G$.

\Prop{Spec$(L)$ and inclusion}{prop:specL}{
Given a rule graph $L$ and a host graph $G$ where $V_L\subseteq V_G$ and $E_L\subseteq E_G$. Then,
$G \Sat \text{Spec}(L)$ if and only if there exists a label assignment $\alpha$ for $L$ such that there exists inclusion $g:L^\alpha\rightarrow G$.
}

\begin{proof}
Let us consider the construction of Spec$(L)$. It is clear that there is no node or edge variables in the condition. Hence, $G$ satisfies Spec$(L)$ if and only if there exists an assignment $\beta$ for all list variables in Spec($L$) and partial functions $\mu_V:V_L\rightarrow \mathbb{M}_V\backslash\{\none\}$ and $\mu_E:E_L\rightarrow \mathbb{M}_E\backslash\{\none\}$ for every item $i$ whose mark is $\mtt{any}$  such that substituting $\beta(x)$ for every variable $x$ and $\mu_V(i)$ or $\mu_E(i)$ for every $\mtt{any}$-mark associated with $i$ in Spec$(L)$ resulting a valid statement in $G$. 
Let we denote by $V_L=\{v_1,\ldots,v_n\}$, $E_L=\{e_1,\ldots,e_m\}$, and $X=\{x_1,\ldots,x_p\}$ the set of all nodes, edges, label variables in $L$. From the semantics of satisfaction, it is clear that  
\[ \bigwedge_{i=1}^n \lst^V_G(v_i)=(\lst^V_L(v_i))^\beta~\wedge~\mrk^V_G(v_i)=(\mrk^V_L(v_i))^{\mu_V}~\wedge~ \text{Root}_G(v_i)\]
\[\wedge~\bigwedge_{i=1}^m s_G(e_i)=s_L(e_i)~\wedge~ t_G(e_i)=t_L(e_i)~\wedge~\lst^E_G(e_i)=(\lst^E_L(e_i))^\beta~\wedge~
\mrk^E_G(e_i)=(\mrk^E_L(e_i))^{\mu_E}\]
Define $g(i)=i$ for every item $i\in V_L\cup E_L$ (such that identifiers are preserved by $g$), and $\alpha=\langle \beta, \mu_V, \mu_E\rangle$. It is clear that $g$ preserves sources, targets, lists, marks, and rootedness.\qed
\end{proof}

Note that Spec$(L)$ is a condition over $L$, so the a graph satisfying the condition must have node and edge identifiers of $L$ in the graph. It is obviously not practical, but we can make it more general by replacing the identifiers with fresh variables such that a graph satisfying the condition does not necessarily contain identifiers of $L$.

\Def{Variablisation of a condition over a graph}{def:var}{
Given a graph $L$ and a condition $c$ over $L$ where $\{v_1,\ldots,v_n\}$ and $\{e_1,\ldots,e_m\}$ represent the set of node and edge constants in $c$ respectively. Let $x_1,\ldots,x_n$ be node variables not in $c$ and $y_1,\ldots,y_m$ be edge variables not in $c$. \emph{Variablisation of $c$}, denoted by Var$(c)$, is the FO formula
\[\mrm{\bigwedge_{i=1}^n\bigwedge_{j\neq i}x_i\neq x_j~\wedge~ \bigwedge_{i=1}^m\bigwedge_{j\neq i}y_i\neq y_j
~\wedge~c^{[v_1\mapsto x_1]\ldots[v_n\mapsto x_n][e_1\mapsto y_1]\ldots[e_m\mapsto y_m]}}\]
where $c^{[a\mapsto b]}$ is obtained from $c$ by replacing every occurrence of $a$ with $b$, and $c^{[a\mapsto b][d\mapsto e]}=(c^{[a\mapsto b]})^{[d\mapsto e]}$.
}

\Lemma{lemma:var}{
Given a graph $L$ and a condition $c$ over $L$. For every host graph $G$ and morphism $g:L\rightarrow G$,
\[G\Sat\text{Var}(c)\text{ if and only if }\RG\Sat c\text.\]
}

\begin{proof}
Let $V=\{v_1,\ldots,v_n\}$ and $E=\{e_1,\ldots,e_m\}$ represent the set of node and edge constants in $c$ respectively, and $X=x_1,\ldots,x_n$ be node variables not in $c$ and $Y=y_1,\ldots,y_m$ be edge variables not in $c$ such that Var$(c)$ is the FO formula shown in the definition above.\\
Let $\alpha$ be an assignment such that $\alpha(x_i)=v_i$ and $\alpha(y_i)=e_i$ for all $x_i\in X$ and $y_i\in Y$. It is obvious that (Var$(c)$)$^{\alpha}\equiv c$, since we only replace each node/edge variable with the constant that was replaced by the variable to obtain Var$(c)$. Therefore, $\RG\Sat c$ iff $\RG\Sat$Var$(c)^{\alpha}$ iff $G\Sat$Var$(c)^{\alpha}$, which means that $G$ satisfies Var$(c)$.\qed
\end{proof}$~~$\\

If we apply this variablisation to Spec$(L)$ for a rule graph $L$, morphism as in Proposition \ref{prop:specL} should also exist but without necessarily preserves identifiers.

\Lemma{lemma:VarL}{
Given rule graph $L$ and host graph $G$. Then,
$G \Sat \text{Var(Spec}(L))$ if and only if there exists a label assignment $\alpha$ for $L$ such that there exists injective morphism $g:L^\alpha\rightarrow G$.
}

\begin{proof}
$G$ satisfying Form(Spec$(L)$) if and only if there exists formula assignment $\gamma=\tuple{\gamma_V, \gamma_L. \gamma_\mathbb{L}}$ and mappings $\mu_V:V_L\rightarrow \mathbb{M}_V\backslash\{none\}$ and $\mu_E:E_L\rightarrow \mathbb{M}_E\backslash\{none\}$ for every item $i$ whose mark is $\mtt{any}$, such that (Form(Spec$(L))^{\gamma})^{\mu_V,\mu_E}$ is true in $G$. 

If we consider Form(Spec$(L))^{\gamma_G}$, it clearly gives us a condition similar to Spec$(L)$, but with different identifiers. Let $X$ denotes the set of images of $\gamma_G$, and $\beta:(V_L\cup G_L)\rightarrow X$ be a bijective mapping such that Spec$(L)^\beta =$Form(Spec$(L))^{\gamma_G}$.

Let we denote by $V_L=\{v_1,\ldots,v_n\}$, $E_L=\{e_1,\ldots,e_m\}$, and $X=\{x_1,\ldots,x_p\}$ the set of all nodes, edges, label variables in $L$. From the semantics of satisfaction, it is clear that  
\begin{scriptsize}
\[ \bigwedge_{i=1}^n \lst^V_G(\beta(v_i))=(\lst^V_L(v_i))^{\gamma_\mathbb{L}}~\wedge~\mrk^V_G(\beta(v_i))=(\mrk^V_L(v_i))^{\mu_V}~\wedge~ \text{Root}_G(\beta(v_i))\]
\[\wedge~\bigwedge_{i=1}^m s_G(\beta(e_i))=s_L(e_i)~\wedge~ t_G(\beta(e_i))=t_L(e_i)~\wedge~\lst^E_G(\beta(e_i))=(\lst^E_L(e_i))^{\gamma_\mathbb{L}}~\wedge~
\mrk^E_G(\beta(e_i))=(\mrk^E_L(e_i))^{\mu_E}\]
\end{scriptsize}
Define $g(i)=\beta(i)$ for every item $i\in V_L\cup E_L$, and $\alpha=\langle \gamma_\mathbb, \mu_V,\mu_E\rangle$. It is clear that $g:L^\alpha\rightarrow G$ preserves sources, targets, lists, marks, and rootedness.\qed
\end{proof}

\subsection{Properties of first-order formulas}

\Lemma{lemma:isocond}{
Given a FO formula $c$ and two isomorphic host graphs $G$ and $H$ with isomorphism $f:G\rightarrow H$. Let $\alpha=\tuple{\alpha_G,\alpha_\mathbb{L}}$ and $\beta=\tuple{\beta_H,\beta_\mathbb{L}}$ be formula assignments where $\beta_H(x)=f(\alpha_G(x))$ for every node and edge variable $x$ in $c$ and $\beta_\mathbb{L}(x)=\alpha_\mathbb{L}(x)$ for every list variable $x$ in $c$. Then,
\[G\Sata c \text{ if and only if } H\vDash^\beta c\]
}

\begin{proof} Here, we prove the Lemma inductively.\\
(Base case).
 \vspace{-\topsep}\begin{enumerate}
    \item If $c=\mrm{true}$ or $c=\mrm{false}$, it is obvious that $G\Sata c$ iff $H\vDash^\beta c$
    \item If $c$ is a predicate P$(x)$ for $P\in\{\mrm{int,char,string,atom}\}$ and some list variable $x$, the satisfaction of the predicate is independent on host graphs. Also, it is obvious that $x^\alpha=x^\beta$ such that P$(x^\alpha)$ is true in every host graphs iff P$(x^\beta)$ is true in every host graph
    \item If $c=\mrm{root(x)}$ for some term $x$ representing a node, $x^\beta=g(x^\alpha)$. From Definition \ref{def:morphisms}, we know that $p_G(x^\alpha)=p_H(g(x^\alpha))$. Hence, $\mrm{root(x^\alpha}$ is true in $G$ iff $\mrm{root(x^\beta)}$ is true in $H$
    \item If $c=x_1\otimes x_2$ for $\otimes\in\{=,\neq\}$ and terms $x_1,x_2$ representing edges or nodes, $x_1^\beta=g(x_1^\alpha)$ and $x_2^\beta=g(x_2^\alpha)$. It is clear that $x_1^\alpha\otimes x_2^\alpha$ iff $g(x_1^\alpha)\otimes g(x_2^\alpha)$ because $g$ is injective.
    \item If $c=x_1\otimes x_2$ for $\otimes\in\{=,\neq,\leq,\geq\}$ and terms $x_1,x_2$ representing lists, $x_1^\alpha=x_1^\beta$ and $x_2^\alpha=x_2^\beta$ (note that $l_V(x^\alpha)=l_V(g(x^\alpha))=l_V(x^\beta)$ for all node variable $x$ in $c$, and analogously for $l_E(x)$). Since the truth value of $x_1^\alpha\otimes x_2^\alpha$ does not depend on host graphs, $x_1^\alpha\otimes x_2^\alpha$ is true in $G$ iff $x_1^\beta\otimes x_2^\beta$ is true in $H$
\end{enumerate}
(Inductive case).
Next, we prove the Lemma for the inductive cases. Let $c_1, c_2$ be FO formulas such that $G\Sata c_1$ iff $H\vDash^\beta c_1$ and $G\Sata c_2$ iff $H\vDash^\beta c_2$. Also, let $c^{x\mapsto v}$ for some variable $x$ and constant $v$ represents $c$ after replacement of every free variable $x$ in $c$ with $v$.
 \vspace{-\topsep}\begin{enumerate}
    \item If $c=\neg c_1,$ $G\Sata \neg c_1$ iff $c_1^\alpha$ is false in $G$ iff $c_1^\beta$ is false in $H$ iff $H\vDash^\beta\neg c_1$
    \item If $c=c_1\vee c_2,$ $G\Sata c_1\vee c_2$ iff $G\Sata c_1 \vee G\Sata c_2$ iff $H\vDash^\beta c_1 \vee H\vDash^\beta c_2$ iff $H\vDash^\beta c_1\vee c_2$
    \item If $c=c_1\wedge c_2,$ $G\Sata c_1\wedge c_2$ iff $G\Sata c_1 \wedge G\Sata c_2$ iff $H\vDash^\beta c_1 \wedge H\vDash^\beta c_2$ iff $H\vDash^\beta c_1\wedge c_2$
    \item $G\Sata\E{V}x(c_1)$ iff $(c_1^\alpha)^{[x\mapsto v]}$ for some $v\in V_G$ is true in $G$ iff $(c_1^\beta)^{[x\mapsto g(v)]}$ is true in $H$ iff $H\vDash^\beta\E{V}x(c_1)$
    \item $G\Sata\E{E}x(c_1)$ iff $(c_1^\alpha)^{[x\mapsto e]}$ for some $e\in E_G$ is true in $G$ iff $(c_1^\beta)^{[x\mapsto g(e)]}$ is true in $H$ iff $H\vDash^\beta\E{E}x(c_1)$
    \item $G\Sata\E{L}x(c_1)$ iff $(c_1^\alpha)^{[x\mapsto i]}$ for some $i\in \mathbb{L}$ is true in $G$ iff $(c_1^\beta)^{[x\mapsto i]}$ is true in $H$ iff $H\vDash^\beta\E{L}x(c_1)$

\end{enumerate}\qed
\end{proof}

\begin{corollary}\label{cor:isocond}
Given two isomorphic host graphs $G$ and $H$, and a FO formula $c$. It is true that
\[G\models c \text{ if and only if } H\models c\]
\end{corollary}

\begin{proof}
$G\Sat c$ iff there exists an assignment $\alpha=\tuple{\alpha_G,\alpha_\mathbb{L}}$ such that $G\Sata c$. By Lemma \ref{lemma:isocond}, $G\Sata c$ iff $H\vDash^\beta c$ for $\beta=\tuple{\beta_H,\alpha_\mathbb{L}}$ where $\beta_H(x)=g(\alpha_G(x))$ for all node and edge variables $x$ iff $H\Sat c$.\qed
\end{proof}

$~$\\

\Lemma{lemma:BolOperators}{
Given a host graph $G$ and FO formulas $c_1, c_2$. Then, the following holds:
 \vspace{-\topsep}\begin{enumerate}
    \item $G\Sat c_1\vee c_2$ if and only if $G\Sat c_1 \vee G\Sat c_2$
    \item $G\Sat c_1\wedge c_2$ if and only if $G\Sat^\alpha c_1 \wedge G\Sat^\alpha c_2$ for some assignment $\alpha$
    \item $G\Sat\neg c_1$ if and only if $\neg(G\Sat^\alpha c_1)$ for some assignment $\alpha$
    \item $V_G\neq\emptyset \wedge G\Sat\E{V}x(c_1)$ if and only if $G\Sat c_1$
    \item $E_G\neq\emptyset \wedge G\Sat\E{E}x(c_1)$ if and only if $G\Sat c_1$
    \item $G\Sat\E{L}x(c_1)$ if and only if $G\Sat c_1$
\end{enumerate}
Furthermore, the above properties also hold if $c_1, c_2$ are conditions over $G$.
}

\begin{proof}$~$\\
For a condition $d$ over $G$, from the definition of condition over a graph we know that $d=c^{\alpha_G}$ for some FO formula $c$ and node/edge assignment $\alpha_G$. Also, there is no free node and edge variables in $d$ so that $G\Sat^{\alpha_\mathbb{L}}d$ for some list assignment $\alpha_\mathbb{L}$ is equivalent to $G\Sata c$ for $\alpha=\tuple{\alpha_G,\alpha_\mathbb{L}}$. Hence, we can consider a condition over a graph as a FO formula with a fixed node/label assignment. Then for some FO formulas $c_1,c_2,$
 \vspace{-\topsep}\begin{enumerate}
    \item (only if) $G\Sat c_1\vee c_2$ implies $G\Sat^\alpha c_1\vee c_2$ for some assignment $\alpha$ implies $G\Sat^\alpha c_1$ or $G\Sat^\alpha c_2$ implies $G\Sat c_1\vee G\Sat c_2$\\
    (if) $G\Sat c_1 \vee G\Sat c_2$ implies $G\Sat^\alpha c_1 \vee G\Sat^\beta c_2$ for some assignments $\alpha,\beta$. It implies $(G\Sat^\alpha c_1\vee c_2) \vee (G\Sat^\beta c_1\vee c_2)$. Hence, $G\Sat c_1\vee c_2$
    \item $G\Sat c_1\wedge c_2$ iff $G\Sat^\alpha c_1\wedge c_2$ for some assignment $\alpha$ iff $G\Sat^\alpha c_1$ and $G\Sat^\alpha c_2$
    \item $G\Sat\neg c_1$ iff $c_1^\alpha$ is false in $G$ for some assignment $\alpha$ such that $G\Sat^\alpha c_1$ is false. Hence, $\neg(G\Sata c_1)$
    \item $G\Sat \E{V}x(c_1)$ iff $G\Sat c_1^{[x\mapsto v]}$ for some $v\in V_G$ iff $G\Sata c_1$ for some assignment $\alpha$ such that $\alpha(x)=v$ iff $G\Sat c_1$
    \item $G\Sat \E{E}x(c_1)$ iff $G\Sat c_1^{[x\mapsto e]}$ for some $e\in E_G$ iff $G\Sata c_1$ for some assignment $\alpha$ such that $\alpha(x)=e$ iff $G\Sat c_1$
    \item $G\Sat \E{L}x(c_1)$ iff $G\Sat c_1^{[x\mapsto k]}$ for some $k\in \mathbb{L}$ iff $G\Sata c_1$ for some assignment $\alpha$ such that $\alpha(x)=k$ iff $G\Sat c_1$
\end{enumerate}\qed
\end{proof}

\Lemma{lemma:quantifier}{
Given a host graph $G$ and a condition $c$ over $G$. Let $\{v_1,\ldots,v_n\}\subseteq V_G$ and $\{e_1,\ldots,e_m\}\subseteq E_G$. Then,
 \vspace{-\topsep}\begin{enumerate}
    \item $\mrm{\E{V}x(c)\equiv c^{[x\mapsto v_1]}\vee\ldots\vee c^{[x\mapsto v_n]}\vee\E{V}x(x\neq v_1\wedge\ldots\wedge x\neq v_n\wedge c)}$
    \item $\mrm{\E{E}x(c)\equiv c^{[x\mapsto e_1]}\vee\ldots\vee c^{[x\mapsto e_m]}\vee\E{V}x(x\neq e_1\wedge\ldots\wedge x\neq v_m\wedge c)}$
    \item $\mrm{\E{E}x(c)\equiv \E{E}x(\bigvee_{i=1}^n(\bigvee_{j=1}^n s(x)=v_i\wedge t(x)=v_j\wedge c^{[s(x)\mapsto v_i, t(x)\mapsto v_j]})}$\\
    $\mrm{~~~~~~~~~~~~~~~~~~~~~~~~~~\vee(s(x)=v_i\wedge\bigwedge_{j=1}^n t(x)\neq v_j\wedge c^{[s(x)\mapsto v_i]})}$\\
    $\mrm{~~~~~~~~~~~~~~~~~~~~~~~~~~\vee(\bigwedge_{j=1}^n s(x)\neq v_j\wedge t(x)=v_i\wedge c^{[t(x)\mapsto v_i]})}$\\
    $\mrm{~~~~~~~~~~~~~~~~~~~~\vee(\bigwedge_{i=1}^n s(x)\neq v_i\wedge\bigwedge_{i=1}^n t(x)\neq v_i\wedge c))}$
\end{enumerate}
}

\begin{proof}$~$\\
\begin{tabular}{clcl}
     1. & $\mrm{\E{V}x(c)}$ & $\equiv$ & $\mrm{\E{V}x(((x=v_1\vee\ldots\vee x=v_n)\vee\neg(x=v_1\vee\ldots\vee x=v_n))\wedge c)}$\\
     & & $\equiv$ & $\mrm{\E{V}x((x=v_1\wedge c)\vee\ldots\vee(x=v_n\wedge c)\vee(x\neq v_1\wedge\ldots\wedge x\neq v_n\wedge c))}$\\
     & & $\equiv$ & $\mrm{\E{V}x(c^{[x\mapsto v_1]}\vee\ldots\vee c^{[x\mapsto v_n]}\vee(x\neq v_1\wedge\ldots\wedge x\neq v_n\wedge c))}$\\
     & & $\equiv$ & $\mrm{c^{[x\mapsto v_1]}\vee\ldots\vee c^{[x\mapsto v_n]}\vee\E{V}x(x\neq v_1\wedge\ldots\wedge x\neq v_n\wedge c)}$\\
     2. & \multicolumn{3}{l}{Analogous to point 1}\\
     3. & $\mrm{\E{E}x(c)}$ & $\equiv$ & $\mrm{\E{E}x(((s(x)=v_1\vee\ldots\vee s(x)=v_n)\vee\neg(s(x)=v_1\vee\ldots\vee s(x)=v_n))}$\\
     & & & $\mrm{~~~~~~\wedge (t(x)=v_1\vee\ldots\vee t(x)=v_n\vee\neg(t(x)=v_1\vee\ldots\vee t(x)=v_n)) \wedge c)}$\\
     & & $\equiv$ & $\mrm{\E{E}x((s(x)=v_1\wedge (t(x)=v_1\vee\ldots\vee t(x)=v_n)\wedge c)}$\\
     & & & $~~~~~~~~~~~~~~~~\ldots$\\
     & & & $~~~~~~~\mrm{(s(x)=v_n\wedge (t(x)=v_1\vee\ldots\vee t(x)=v_n)\wedge c)}$\\
     & & & $~~~~~~~\mrm{(s(x)\neq v_1\wedge\ldots\wedge s(x)\neq v_n\wedge (t(x)=v_1\vee\ldots\vee t(x)=v_n)\wedge c)}$\\
     & & & $~~~~~~~\mrm{(s(x)=v_i\wedge (t(x)\neq v_1\wedge\ldots\wedge t(x)\neq v_n)\wedge c)}$\\
    & & & $~~~~~~~\mrm{(s(x)\neq v_1\wedge\ldots\wedge s(x)\neq v_n \wedge (t(x)\neq v_1\wedge\ldots\wedge t(x)\neq v_n)\wedge c)}$\\
    & & $\equiv$ & $\mrm{\E{E}x(\bigvee_{i=1}^n(\bigvee_{j=1}^n s(x)=v_i\wedge t(x)=v_j\wedge c^{[s(x)\mapsto v_i, t(x)\mapsto v_j]})}$\\
    & & & $\mrm{~~~~~~~~~~~~~~\vee(\bigwedge_{j=1}^n s(x)\neq v_j\wedge t(x)=v_i\wedge c^{[t(x)\mapsto v_i]})}$\\
    & & & $\mrm{~~~~~~~~~~~~~~\vee(s(x)=v_i\wedge\bigwedge_{j=1}^n t(x)\neq v_j\wedge c^{[s(x)\mapsto v_i]})}$\\
    & & & $\mrm{~~~~~~\vee(\bigwedge_{i=1}^n s(x)\neq v_i\wedge\bigwedge_{i=1}^n t(x)\neq v_i\wedge c))}$
\end{tabular}\qed\end{proof}

%% file: input/4_SLP2.tex
In this section, we introduce a way to construct a strongest liberal postcondition over a graph program. Here, conditions (including pre- and postconditions) refer to closed FO formulas.

\subsection{Calculating strongest liberal postconditions}
A strongest liberal postcondition is one of predicate transformers \cite{DijkstraS90} for forward reasoning. It expresses properties that must be satisfied by every graph result from the application of the input rule schema to a graph satisfying the input precondition. 

\Def{Strongest liberal postcondition over a conditional rule schema}{def:slp}{
An assertion $d$ is a \emph{liberal postcondition} w.r.t. a conditional rule schema $r$ and a precondition $c$, if for all host graphs $G$ and $H$,
\[G\vDash c \text{ and } G\Rightarrow_r H \text{ implies }H\vDash d.\]
A \emph{strongest liberal postcondition} w.r.t. $c$ and $r$, denoted by $\text{SLP}(c,r)$, is a liberal postcondition w.r.t. $c$ and $r$ that implies every liberal postcondition w.r.t. $c$ and $r$.
}
$~$\\$~$\\
Our definition of a strongest liberal postcondition is different with the definitions in \cite{HP09,DijkstraS90,Cousot90} where they define $\text{SLP}(c,r)$ as a condition such that for every host graph $H$ satisfying the condition, there exists a host graph $G$ satisfying $c$ where $G\Rightarrow_rH$. Lemma \ref{lemma:slp} shows that their definition and ours are equivalent.

\begin{lemma}\normalfont\label{lemma:slp}
Given a rule schema $r$, a precondition $c$. Let $d$ be a liberal postcondition w.r.t. $r$ and $c$. Then $d$ is a strongest liberal postcondition w.r.t. $r$ and $c$ if and only if for every graph $H$ satisfying $d$, there exists a host graph $G$ satisfying $c$ such that $G\Rightarrow_{r}H$.
\end{lemma}

\begin{proof}$~$\\
(If).\\
Let $H$ be a host graph satisfying $d$. Then, there must exists a graph $G$ such that $G\vDash c$ and $G\Rightarrow_r H$. Hence, $H\vDash a$ for any liberal postcondition $a$ from the definition of a liberal postcondition.\\
(Only if).\\
Assume that it is not true that for every host graph $H$, $H\vDash d$ implies there exists a host graph $G$ satisfying $c$ such that $G\Rightarrow_{r}H$. We show that a graph satisfying $d$ can not imply the graph satisfying any liberal postcondition w.r.t $r$ and $c$. From the assumption, there exists a host graph $H$ such that every host graph $G$ does not satisfy $c$ or does not derive $H$ by $r$. In the case of $G$ does not derive $H$ by $r$, we clearly can not guarantee characteristic of $H$ w.r.t. $c$. Then for the case where $G$ does not satisfy $c$ but derives $H$ by $r$, we also can not guarantee the satisfaction of any liberal postcondition $a$ over $c$ and $r$ in $H$ because $a$ is dependent of $c$. Hence, we can not guarantee that $H$ satisfying all liberal postcondition w.r.t. $r$ and $c$.\qed
\end{proof}

To construct $\text{SLP}(c,r)$, we use the generalised version of $r$ to open a possibility of constructing a strongest liberal postcondition over the inverse of a rule schema. Since a rule schema has some restriction on the existence of variables and $\any$-mark, a rule schema may not be invertible. By using the generalised version of a rule schema, we omit this limitation so that the generalised version of the inverse of a rule schema is also a generalised rule schema so that we can use the construction for an inverse rule as well.

In this paper, $\text{SLP}(c,r)$ is obtained by defining transformations Lift$(c,r^\vee)$, Shift$(c,r^\vee)$, and Post$(c,r^\vee)$. The transformation Lift transforms the given condition $c$ into a left-application condition w.r.t. the given unrestricted rule schema $r$. Then, we transform the left-application condition to right-application condition by transformation Shift. Finally, the transformation Post transforms the right-application condition to a strongest liberal postcondition (see \figurename~\ref{fig:pretopost}).

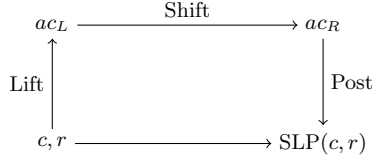
\begin{figure}
    \centering
    \begin{tikzpicture}[scale=0.9, transform shape]
	\node (a) at (4, 0.75) {$ac_R$};
	\node (b) at (0, 0.75) {$ac_L$};
%	\node (c) at (0, -0.25) {$r^\alpha,c$};
	\node (d) at (4, -1) {$\text{SLP}(c,r)$};
	\node (f) at (0, -1) {$c,r$};
	\draw[->] (f) to node {} (d);
	\draw[->] (f) to node[left] {Lift} (b);
	\draw[->] (b) to node[above] {Shift} (a);
	\draw[->] (a) to node[right] {Post} (d);
%	\draw[dots] (f) to node[right] {} (c);
\end{tikzpicture}
    \caption{Constructing $\text{SLP}(c,r)$}
    \label{fig:pretopost}
\end{figure}

For a conditional rule schema $\tuple{r,\Gamma}$ with rule schema $r=\tuple{L\leftarrow K\rightarrow R}$ and a precondition $c$, when a graph $G$ satisfying $c$ and there exists a label assignment $\alpha_L$ such that $G\Rightarrow_{r^\alpha,g}H$ for some host graph $H$ and injective graph morphism $g:L^\alpha\hookrightarrow G$, $ac_L^\alpha=(\text{Lift}(c,r^\vee))^\alpha$ should be satisfied by $G$ w.r.t. $g$. The replacement graph $\RG$ should satisfies $ac_L$ which means $ac_L$ should consist of the precondition $c$, rule schema condition $\Gamma$, and the dangling condition.

$G\Rightarrow_{r^\alpha,g}H$ with injective graph morphism $g:L^\alpha\hookrightarrow G$ and label assignment $\alpha_L$ obviously assert the existence of injective morphism $g^*:R^\beta\rightarrow H$ for some label assignment $\beta_R$ such that $\alpha_L(i)=\beta_R(i)$ for common element $i$ (see \figurename~\ref{fig:ddergenrule}). The graph replacement $\RH$ then should satisfy $ac_R^\beta=(\text{Shift}(c,r^\vee))^\beta$. The graph condition $ac_R$ should describe the elements of the image of the comatch and some properties of $c$ that are still relevant after the rule schema application.

Basically, $ac_R$ is already a strongest property that must be satisfied by a resulting graph. However, it has node/edge constants so that we need to change it into a closed formula so that we finally obtain a strongest liberal postcondition. This part is done by the transformation Post.

To give a better idea of the transformations we define in this chapter, we show examples after each definition. We use the conditional rule schemata $r_1=\mtt{del}$ of \figurename~\ref{fig:ruledel} and $\mtt{copy}$ of \figurename~\ref{fig:rulecopy} and the preconditions $q_1=\mrm{\neg\E{E}x(m_V(s(x))\neq none)}$ and  $q_2=\mrm{\E{V}x(\neg root(x))}$ as running examples. We denote by $\Gamma_1$ and $\Gamma_2$ the GP\,2 rule schema conditions $\mrm{d\geq e}$ and $\mrm{outdeg(1)\neq 0}$ respectively. Also, we denote by $r_1$ and $r_2$ the rule schema of $\mtt{del}$ and $\mtt{copy}$ respectively.

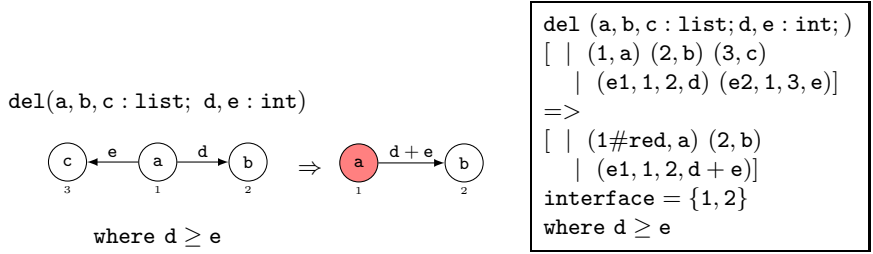
\begin{figure}
\subfloat{%
\begin{minipage}[c]{0.62\textwidth}%
\centering  
  \begin{tikzpicture}[remember picture,
  inner/.style={circle,draw,minimum size=18pt},
  outer/.style={inner sep=2pt}, scale=0.9
  ]
  \node[outer] (AA) at (0,1) {$\mtt{del(a,b,c:list;~ d,e:int)}$};
  \node[outer] (A) at (0,0) {
  \begin{tikzpicture}[scale=0.8, transform shape]
		\node[inner, label=below:\tiny 1] (Aa) at (0,0) {$\mtt{a}$};	
		\node[inner, label=below:\tiny 2] (Ab) at (1.5,0) {$\mtt{b}$};	
        \node[inner, label=below:\tiny 3] (Ac) at (-1.5,0) {$\mtt{c}$};	
		\draw[-latex] (Aa) to node[above] {$\mtt{d}$} (Ab);
		\draw[-latex] (Aa) to node[above] {$\mtt{e}$} (Ac);
		\end{tikzpicture}};
  \node[outer] (B) at (0,-1) {$\mtt{where~d\geq e}$};
  \node[outer] (B) at (2.25,0) {$\Rightarrow$};
  \node[outer] (C) at (3.75,0) {
  \begin{tikzpicture}[scale=0.8, transform shape]
		\node[inner, label=below:\tiny 1, fill=red!50] (Aa) at (0,0) {$\mtt{a}$};	
		\node[inner, label=below:\tiny 2] (Ab) at (1.75,0) {$\mtt{b}$};	
        \draw[-latex] (Aa) to node[above] {$\mtt{d+e}$} (Ab);
		\end{tikzpicture}};
	\end{tikzpicture}
\end{minipage}}
\subfloat{%
\fbox{\begin{minipage}[c]{0.35\textwidth}%
            \begin{tabular}{l}
\small{$\mtt{del~(a,b,c:list; d,e:int;)}$}\\
\small{$\mtt{[~\mid~(1, a)~(2, b)~(3, c)}$}\\
\small{$\mtt{~~~\mid~(e1, 1, 2, d)~(e2, 1, 3, e)]}$}\\
\small{$\mtt{=>}$}\\
\small{$\mtt{[~\mid~(1\#red, a)~(2, b)}$}\\
\small{$\mtt{~~~\mid~(e1, 1, 2, d+e)]}$}\\
\small{$\mtt{interface = \{1, 2\}}$}\\
\small{$\mtt{where~d\geq e}$}
             \end{tabular}
\end{minipage}}}
\caption{GP\,2 conditional rule schema $\mtt{del}$}
\label{fig:ruledel}
\end{figure}

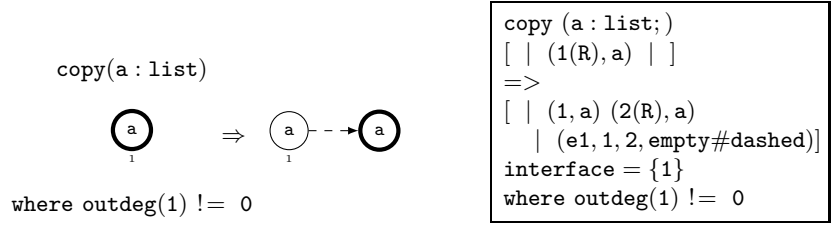
\begin{figure}
\subfloat{%
\begin{minipage}[c]{0.62\textwidth}%
\centering  
  \begin{tikzpicture}[remember picture,
  inner/.style={circle,draw,minimum size=18pt},
  outer/.style={inner sep=2pt}, scale=0.9
  ]
  \node[outer] (AA) at (0,1) {$\mtt{copy(a:list)}$};
  \node[outer] (A) at (0,0) {
  \begin{tikzpicture}[scale=0.8, transform shape]
		\node[inner, label=below:\tiny 1, ultra thick] (Aa) at (0,0) {$\mtt{a}$};	
		\end{tikzpicture}};
  \node[outer] (B) at (0,-1) {$\mtt{where~outdeg(1)~!\!=~0}$};
  \node[outer] (B) at (1.5,0) {$\Rightarrow$};
  \node[outer] (C) at (3,0) {
  \begin{tikzpicture}[scale=0.8, transform shape]
		\node[inner, label=below:\tiny 1] (Aa) at (0,0) {$\mtt{a}$};	
		\node[inner, ultra thick] (Ab) at (1.5,0) {$\mtt{a}$};	
		\draw[-latex, dashed] (Aa) to node[above] {} (Ab);
		\end{tikzpicture}};
\end{tikzpicture}
\end{minipage}}
\subfloat{%
\fbox{\begin{minipage}[c]{0.35\textwidth}%
            \begin{tabular}{l}
\small{$\mtt{copy~(a:list;)}$}\\
\small{$\mtt{[~\mid~(1(R), a)~\mid~]}$}\\
\small{$\mtt{=>}$}\\
\small{$\mtt{[~\mid~(1, a)~(2(R), a)}$}\\
\small{$\mtt{~~~\mid~(e1, 1, 2, empty\#dashed)]}$}\\
\small{$\mtt{interface = \{1\}}$}\\
\small{$\mtt{where~outdeg(1)~!\!=~0}$}
             \end{tabular}
\end{minipage}}}
\caption{GP\,2 conditional rule schema $\mtt{copy}$}
\label{fig:rulecopy}
\end{figure}

\subsection{The dangling condition}
The dangling condition must be satisfied by an injective morphism $g$ if $G\Rightarrow_{r,g}H$ for some rule schema $r=\tuple{L\leftarrow K\rightarrow R}$ and host graphs $G,H$. Since we want to express properties of $\RG$ where such derivation exists, we need to express the dangling condition as a condition over the left-hand graph.

Recall the dangling condition from Definition \ref{def:dang}. $\RG$ satisfies the dangling condition if every node $v\in L-K$ does not incident to any edge outside $L$. This means that the indegree and outdegree of every node $v\in L-K$ in $L$ represent the indegree and outdegree of $v$ in $G$ as well.

\Def{Condition Dang}{def:dang}{
Given an unrestricted rule schema $r=\tuple{L\leftarrow K\rightarrow R}$ where $\{v_1,\ldots,v_n\}$ is the set of all nodes in $L-K$. Let $indeg_L(v)$ and $outdeg_L(v)$ denotes the indegree and outdegree of $v$ in $L$, respectively. The condition Dang$(r)$ is defined as:
 \vspace{-\topsep}\begin{enumerate}
    \item if $V_L-V_K=\emptyset$ then Dang$(r)=\mrm{true}$
    \item if $V_L-V_K\neq\emptyset$ then 
    \[\text{Dang}(r)=\mrm{\bigwedge_{i=1}^n indeg(v_i)=}indeg_L(v_i)\,\wedge\,\mrm{outdeg(v_i)=}outdeg_L(v_i)\]
\end{enumerate}
}

\begin{example}[Condition Dang]\label{ex:dang}$~$
 \vspace{-\topsep}\begin{enumerate}
    \item Dang$(r_1) = \mrm{indeg(3)=1\wedge outdeg(3)=0}$
    \item Dang$(r_2) = \mrm{true}$
\end{enumerate}
\end{example}

\Obsv{obsv:dang}{
Given an unrestricted rule schema $r=\langle L\leftarrow K\rightarrow R\rangle$. Let $G$ be a host graph and $g:L\rightarrow G$ be a premorphism. The dangling condition is satisfied if and only if $\rho_g(G)\Sat$Dang$(r)$.
}

\begin{proof}
From the definition of the dangling condition (see Definition \ref{def:dang}), the dangling condition is satisfied when no edge in $G-g(L)$ is incident to any node in $g(L-K)$. By the definition of replacement graph (see Definition \ref{def:rho}), it is obvious that $G-g(L)$ is equivalent to $\rho_g(G)-L$. Then, evaluating the construct of $g(L-K)$ in $G$ w.r.t. $g$ is the same as evaluating the $L-K$ in $\rho_g(G)$. Hence, the dangling condition is satisfied iff no edge in $\rho_g(G)-L$ incident to any node in $L-K$, which means all nodes in $L-K$ only incident to edges in $L$. Hence, Dang$(r)$ is true.\qed
\end{proof}

\subsection{From precondition to left-application condition}

Now, we start with transforming a precondition $c$ to a left-application condition with respect to a generalised rule $w=\tuple{r,ac_L,ac_r}$. Intuitively, the transformation is done by: 1) Find all possibilities of variables in $c$ representing nodes/edges in an input and form a disjunction from all possibilities, denotes by Split$(c,r)$; 2) Express the dangling condition as a condition over $L$, denoted by Dang$(r)$; 3) Evaluate terms and Boolean expression we can evaluate in Split$(c,r)$, Dang$(r)$, and $\Gamma$, then form a conjunction from the result of evaluation, and simplify the conjunction.

A possibility of variables in $c$ representing nodes/edges in an input graph as mentioned above refers to how variables in $c$ can represent node or edge constants in the replacement of the input graph. A simple example would be for a precondition $c=\E{V}x(c_1)$ for some FO formula $c_1$ with a free variable $x$, $c$ holds on a host graph $G$ if there exists a node $v$ in $G$ such that $c_1^{\alpha}$ where $\alpha(x)=v$ is true in $G$. The node $v$ can be any node in $G$. In the replacement graph of $G$, $v$ can be any node in the left-hand graph of the rule schema, or any node outside it. Split$(c,r)$ is obtained from the disjunction of all these possibilities.

\Def{Transformation Split}{def:split}{
Given an unrestricted rule schema $r=\tuple{L\leftarrow K\rightarrow R}$. where $V_L=\{v_1,\ldots,v_n\}$ and $E_L=\{e_1,\ldots,e_m\}$. Let $c$ be a condition over $L$ sharing no variables with $r$ (note that it is always possible to replace the label variables in $c$ with new variables that are distinct from variables in $r$). We define the condition $\text{Split}(c,r)$ over $L$ inductively as follows:
\begin{tabular}{ll}
    \multicolumn{2}{l}{- Base case.}  \\
     $~$ &  If $c$ is $\mrm{true}$, $\mrm{false}$, a predicate $\mrm{int(t), char(t), string(t), atom(t), root(t)}$ for \\&some term $\mrm{t}$, or in the form $\mrm{t_1\ominus t_2}$ for $\mrm{\ominus\in\{=.\neq.<,\leq,>,\geq\}}$ and some\\& terms $\mrm{t_1, t_2}$,\\
     & \multicolumn{1}{c}{\text{Split$(c,r) = c$}}\\
    \multicolumn{2}{l}{- Inductive case.}  \\
     & Let $c_1$ and $c_2$ be conditions over $L$.\\
     & 1) $\text{Split}(c_1\vee c_2, r) = \text{Split}(c_1, r)\vee\text{Split}(c_2, r)$,\\
    & 2) $\text{Split}(c_1 \wedge c_2, r) = \text{Split}(c_1, r)\wedge\text{Split}(c_2, r)$,\\
    & 3) $\text{Split}(\neg c_1, r) = \neg\text{Split}(c_1, r)$,\\
    & 4) $\text{Split}(\mrm{\E{V}x}(c_1), r)= (\mrm{\bigvee_{i=1}^n}\text{Split}(c_1^{[ x\mapsto v_i]}, r))\vee\mrm{\E{V} x(\bigwedge_{i=1}^n x{\neq}v_i\,\wedge\,} \text{Split}(c_1, r)$,\\
    & 5) $\text{Split}(\mrm{\E{E}x}(c_1), r)=\mrm{(\bigvee_{i=1}^m}\text{Split}(c_1^{[x\mapsto e_i]}, r))\vee\mrm{\E{E}x(\bigwedge_{i=1}^m x{\neq}e_i\,\wedge\,} \text{inc}(c_1, r,x))$,\\
   & ~~~~where\\
   & $~~~~\text{inc}(c_1, r,x)=\mrm{\bigvee_{i=1}^n (\bigvee_{j=1}^n s(x)=v_i\wedge t(x)=v_j\,\wedge\,}\text{Split}(c_1^{[\mrm{s(x)\mapsto v_i, t(x)\mapsto v_j}]}, r))$\\
   & $~~~~~~~~~~~~~~~~~~~~~~~~~~~\mrm{\vee\,
   (s(x)=v_i\,\wedge\,\bigwedge_{j=1}^n t(x)\neq v_j\,\wedge\,}\text{Split}(c_1^{[\mrm{s(x)\mapsto v_i}]}, r))$\\
   & $~~~~~~~~~~~~~~~~~~~~~~~~~~~\mrm{\vee\,
   (\bigwedge_{j=1}^n s(x)\neq v_j\,\wedge\,t(x)= v_i\,\wedge\,}\text{Split}(c_1^{[\mrm{t(x)\mapsto v_i}]}, r))$\\
   & $~~~~~~~~~~~~~~~~~~~~~\mrm{\vee\,
   (\bigwedge_{i=1}^n s(x)\neq v_i\,\wedge\,\bigwedge_{j=1}^n t(x)\neq v_j\,\wedge\,}\text{Split}(c_1
   , r))$\\
   & 6) $\text{Split}(\mrm{\E{L}x}(c_1), r)=\E{L}\mrm{x}(\text{Split}(c_1, r))$
\end{tabular}\\
\noindent where $c^{[a\mapsto b]}$ for a variable $a$ and constant $b$ represents the condition $c$ after the replacement of all occurrence of $a$ with $b$. Similarly, $c^{[d\mapsto b]}$ for $d\in\{\mrm{s(x), t(x)}\}$ is also a replacement $d$ with $b$.
}
$~$\\$~$\\
As can be seen in the definition above, Split of an edge quantifier is not as simple as Split of a node quantifier. For an edge variable $x$ in a precondition, $x$ can represent any edge in $G$. Moreover, the term $\mrm{s(x)}$ or $\mrm{t(x)}$ may represent a node in the image of the match. Hence, we need to check these possibilities as well. However, if the precondition does not contain a term $\mrm{s(x)}$ or $\mrm{t(x)}$ for some edge variable $x$, we do not need to consider nodes that can be represented by the functions.

\Obsv{obsv:eE}{
Given an unrestricted rule schema $r=\tuple{L\leftarrow K\rightarrow R}$ where $V_L=\{v_1,\ldots,v_n\}$ and $E_L=\{e_1,\ldots,e_m\}$. Let $c=\E{E}x(c_1)$ be a condition over $L$. Then, the following holds:
 \vspace{-\topsep}\begin{enumerate}
    \item If $c_1$ does not contain the term $\mrm{s(x)},$\\
    $\text{inc}(c_1, r)=\mrm{\bigvee_{i=1}^n (t(x)=v_i\,\wedge\,}\text{Split}(c_1^{[\mrm{t(x)\mapsto v_i}]}, r))
   \mrm{\,\vee\,
   \bigwedge_{i=1}^n (t(x)\neq v_i\,\wedge\,}\text{Split}(c_1
   , r))$
   \item If $c_1$ does not contain the term $\mrm{t(x)},$\\
    $\text{inc}(c_1, r)=\mrm{\bigvee_{i=1}^n (s(x)=v_i\,\wedge\,}\text{Split}(c_1^{[\mrm{s(x)\mapsto v_i}]}, r))
   \mrm{\,\vee\,
   \bigwedge_{i=1}^n (s(x)\neq v_i\,\wedge\,}\text{Split}(c_1
   , r))$
   \item If $c_1$ does not contain the terms $\mrm{s(x)}\text{ and }\mrm{t(x)},$\\
    $\text{inc}(c_1, r)=\text{Split}(c_1
   , r)$
\end{enumerate}
}

\begin{proof}$~$\\
 \vspace{-\topsep}\begin{enumerate}
    \item If $c_1$ does not contain the term $\mrm{s(x)},$ then for any $i,j$, $c_1^{[\mrm{s(x)\mapsto v_i, t(x)\mapsto v_j}]}=c_1^{[\mrm{t(x)\mapsto v_j}]}$, and $c_1^{[\mrm{s(x)}]}=c_1$. The first and the third line of inc$(c_1, r)$ is the disjunction of all possibilities of $\mrm(t(x))$ is one of nodes in $L$ while the second and forth line is about $\mrm(t(x))$ is outside the match.
    \item Analogously to above.
    \item If $c_1$ does not contain the terms $\mrm{s(x)}\text{ and }\mrm{t(x)},$ it is obvious that\\ $c_1^{[\mrm{s(x)\mapsto v_i, t(x)\mapsto v_j}]}$ $=\,c_1^{[\mrm{t(x)\mapsto v_j}]}$ $=\,c_1^{[\mrm{s(x)\mapsto v_j}]}$ $=\,c_1$.\qed
\end{enumerate}
\end{proof}

\Ex{Transformation Split}{ex:Split}{$~$\\
\begin{tabular}[t]{lcl}
     Split$(q_1, r_1)$ & = & $\neg$ Split$(\mrm{\E{E}x(s(x)= t(x))}, r_1)$  \\
     & = & $\mrm{\neg(m_V(s(e1))\neq none \vee m_V(s(e2))\neq none\,\vee}$\\
     && $\mrm{~~~\E{E}x(x\neq e1\wedge x\neq e2\,\wedge\,((s(x)=1\wedge m_V(1)\neq none)}$\\
     && $~~~~~~~~~~~~~~~~~~~~~~~~~~~~~~~~~~~\mrm{\vee\,(s(x)=2\wedge m_V(2)\neq none)}$\\
     && $~~~~~~~~~~~~~~~~~~~~~~~~~~~~~~~~~~~\mrm{\vee\,(s(x)=3\wedge m_V(3)\neq none)}$\\
     && $~~~~~~~~~~~~~~~~~~~~~~~~~~~~~~~~~~~\mrm{\vee\,(s(x)\neq 1\wedge s(x)\neq 2\wedge s(x)\neq 3}$\\
     && $~~~~~~~~~~~~~~~~~~~~~~~~~~~~~~~~~~~~~~\mrm{\wedge m_V(s(x))\neq none))))}$\\
     Split$(q_2, r_2)$& = & $\mrm{\neg root(1) \vee \E{V}x(x\neq 1 \wedge \neg root(x))}$\\
    \end{tabular}
}

Since Split$(c,r)$ only disjunct all possibilities of nodes and edges that can be represented by node and edge variables in $c,$ it should not change the semantic of $c$. However, we transform a condition $c$ to a condition over $L$ such that we may not be able to check satisfaction of Split$(c,r)$ in $G$. However, we can always check its satisfaction in $\RG$ for some premorphism $g:L\rightarrow G$.

\Lemma{lemma:split}{
Given a condition $c$ and an unrestricted rule schema $r=\tuple{L\leftarrow K\rightarrow R}$, sharing no variables with $c$. For a host graph $G$, let $g : L \rightarrow G$ be a premorphism. Then,
\[G \models c \text{ if and only if } \rho_g(G)\models \text{Split}(c, r).\]
}

\begin{proof}Here, we prove the lemma inductively on conditions. The texts above the symbol $\Leftrightarrow$ bellow refer to lemmas that imply the associated implication, e.g. L4 refers to Lemma 4.
\begin{longtable}{llcp{10.5cm}}
   \multicolumn{4}{l}{(Base case).} \\
    \multicolumn{2}{l}{$G\models c$} & $\myeq{\Leftrightarrow}{L\ref{lemma:isocond}}$ & $\rho_g(G)\models c$ \\
    & & $\Leftrightarrow$ & $\rho_g(G)\models \text{Split}(c, r)$\\
   \multicolumn{4}{l}{(Inductive case).}\\
   \multicolumn{4}{l}{Assuming that for some conditions $c_1$ and $c_2$ over $L$, the lemma holds.} \\
     1) & $G\models c_1\vee c_2$ & $\myeq{\Leftrightarrow}{L\ref{lemma:BolOperators}}$ & $G\models c_1\vee G\models c_2$ \\
    & & $\Leftrightarrow$ & $\rho_g(G)\models \text{Split}(c_1, r)\vee \rho_g(G)\models \text{Split}(c_2, r)$\\
     & & $\myeq{\Leftrightarrow}{L\ref{lemma:BolOperators}}$ & $\rho_g(G)\models \text{Split}(c_1, r)\vee \text{Split}(c_2, r)$ \\
     2) & $G\models c_1\wedge c_2$ & $\myeq{\Leftrightarrow}{L\ref{lemma:BolOperators}}$ & $G\models^\alpha c_1\wedge G\models^\alpha c_2$ for some assignment $\alpha$ \\
    & & $\Leftrightarrow$ & $\rho_g(G)\models^\beta \text{Split}(c_1, r)\vee \rho_g(G)\models^\beta \text{Split}(c_2, r)$\\
    &&& where $\beta(x)=\alpha(x)$ if $x\notin V_L$; $\beta(x)=g^{-1}(\alpha(x))$ otherwise\\
     & & $\myeq{\Leftrightarrow}{L\ref{lemma:BolOperators}}$ & $\rho_g(G)\models \text{Split}(c_1, r)\vee \text{Split}(c_2, r)$ \\
     3) & $G\models \neg\, c_1$  & $\myeq{\Leftrightarrow}{L\ref{lemma:BolOperators}}$ & $\neg(G\Sata c_1)$ for some assignment $\alpha$\\
     && $\Leftrightarrow$ & $\neg(\rho_g(G)\models^\beta \text{Split}(c_1, r))$ \\
     &&& where $\beta(x)=\alpha(x)$ if $x\notin V_L$; $\beta(x)=g^{-1}(\alpha(x))$ otherwise\\
     && $\myeq{\Leftrightarrow}{L\ref{lemma:BolOperators}}$ & $\rho_g(G)\models \neg\text{Split}(c_1, r)$  \\
     4) & $G\models \mrm{\E{V}x}(c_1)$ & $\myeq{\Leftrightarrow}{L\ref{lemma:quantifier}}$ & $G\models \mrm{\bigvee_{i=1}^n} {c_1}^{[ \mrm{x\mapsto v_i}]}\vee\mrm{\E{V} x(\bigwedge_{i=1}^n x\neq v_i\wedge} c_1)$  \\
     & & $\Leftrightarrow$ & $\rho_g(G)\models \mrm{\bigvee_{i=1}^n \text{Split}(c_1^{[ x\mapsto v_i]}, r)\vee\E{V} x(\bigwedge_{i=1}^n x\neq v_i\wedge \text{Split}(c_1, r))}$  \\
     5) & $G\models \E{E} \mrm{x}(c_1)$ & $\myeq{\Leftrightarrow}{L\ref{lemma:quantifier}}$ & $G\models \mrm{\bigvee_{i=1}^m} {c_1}^{[ \mrm{x\mapsto e_i}]}\mrm{\vee\E{V} x(\bigwedge_{i=1}^m x\neq e_i\wedge} c_1)$  \\
     & & $\myeq{\Leftrightarrow}{L\ref{lemma:quantifier}}$ & $\rho_g(G)\models \mrm{\bigvee_{i=1}^m \text{Split}(c_1^{[ x\mapsto e_i]}, r)\vee\E{V} x(\bigwedge_{i=1}^m x\neq v_i\wedge \text{Split}(c_1, r))}$  \\
     & & $\myeq{\Leftrightarrow}{L\ref{lemma:quantifier}}$ & $\rho_g(G)\models \mrm{\bigvee_{i=1}^m \text{Split}(c_1^{[ x\mapsto e_i]}, r)\vee\E{V} x(\bigwedge_{i=1}^m x\neq v_i\wedge \text{inc}(c_1,r))}$  \\
     6) & $G\models \E{L}\mrm{x}(c_1)$ & $\myeq{\Leftrightarrow}{L\ref{lemma:BolOperators}}$ & $G\models c_1$ \\
     &  & $\Leftrightarrow$ & $\rho_g(G)\models \text{Split}(c_1, r)$  \\
     & & $\myeq{\Leftrightarrow}{L\ref{lemma:BolOperators}}$ & $\rho_g(G)\models \E{L}\mrm{x}(\text{Split}(c_1, r))$
\end{longtable}\qed
\end{proof}

After splitting the precondition into all possibilities of representations, we check the value of some functions and Boolean operators to check if any possibility violates the precondition such that we can omit the possibility.

\Def{Valuation of $c$}{def:val}{
Given an unrestricted rule schema $r=\tuple{L\leftarrow K\rightarrow R}$, a condition $c$ over $L$, a host graph $G$, and premoprhism $g:L\rightarrow G$. Let $c$ shares no variable with $L$ unless $c$ is a rule schema condition. Let also $F=\{\mrm{s,t,l_V,l_E,m_V,m_E}$,$\mrm{indeg,outdeg,length}$ be the set of function syntax. Let also $y\oplus_L z$ for $\oplus\in\{+,-,*,/,:,.\}$ and $y,z\in\mathbb{L}$ denotes the value of $y\oplus z$ as desribed in Section 3.3, and $f_L(z)$ for a constant $z$ and $f\in F$ denotes the value of $f(y)$ in $L$. \emph{Valuation of $c$ w.r.t. $r$}, written Val$(c, r)$, is constructed by applying the following steps to $c$:
 \vspace{-\topsep}\begin{enumerate}\footnotesize
    \item Obtain $c'$ by changing every term $x$ in $c$ with $T(x)$, where
    \begin{enumerate}
        \item If $x$ is a constant or variable, $T(x)=x$
        \item If $x=f(y)$ for $f\in F,$\\
        $T(x)=\begin{cases}
        f_L(y)&\text{if $f\in F\backslash\{\mrm{indeg,outdeg}\}$ and $y$ is a constant}\\
        & ~~\text{or $f\in\{\mrm{indeg,outdeg}\}$ and $y\in V_L-V_K$}
        \\
        f_L(T(x)) &\text{if $f\in\{\mrm{l_V,m_V}\}$ and $(y=\mrm{s(e)}$ or $y=\mrm{t(e)})$ for $e\in E_L$}\\
        &~~\text{or $f\in\{\mrm{indeg,outdeg}\}$ and $T(y)\in V_L-V_K$}\\
        incon(T(y))+f_L(T(y)) &\text{if $f=\mrm{indeg}$ and $y\in V_K$}\\
        outcon(T(y))+f_L(T(y)) &\text{if $f=\mrm{outdeg}$ and $y\in V_K$}\\
f(y) &\text{otherwise}
        \end{cases}$
    \item If $x\oplus z$ for $\oplus\in\{+,-,/,*,:,.\},$\\
    $T(x)=\begin{cases}
    y\oplus_L z & \text{if $y,z\in\mathbb{L}$}\\
    T(y)\oplus T(z) &\text{if $T(y)=\notin\mathbb{L}$ or $T(z)=\notin\mathbb{L}$}\\
    T(T(y)\oplus T(z)) &\text{otherwise}
    \end{cases}$
    \end{enumerate}
\item Obtain $c"$ by replacing predicates and Boolean operators $x$ in $c'$ with $B(x)$, where\\
$B(x)=\begin{cases}
y\otimes_\mathbb{B} z & \text{if $x=y\otimes z$ for $\otimes\in \{=,\neq,\leq,\geq\}$ and constants $y,z$}\\
\mrm{true} &\text{if $x=\mrm{root(v)}$ for $v\in r_L$}\\
\mrm{false} &\text{if $x=\mrm{root(v)}$ for $v\notin r_L$}\\
x &\text{otherwise}
\end{cases}$
\item Simplify $c"$ such that there are no subformulas in the form $\mrm{\neg\, true,} \neg(\neg\,a)$ ${\neg(a\vee b),}$ ${\neg(a\wedge b)}$ for some conditions $a, b$. We can always simplify them to $\mrm{false}, a, \neg a\wedge\neg b, \neg a\vee\neg b$ respectively.
\end{enumerate}
}

\begin{example}[Valuation of a graph condition]\label{ex:val}\normalfont
For rules $r_1$ and $r_2$,
 \vspace{-\topsep}\begin{enumerate}
    \item \begin{tabular}[t]{lcl}
    \multicolumn{3}{l}{Val$(\text{Split}(q, r_1), r_1)$}\\
    & = & $\mrm{\neg(none\neq none \vee none\neq none\,\vee}$\\
     && $\mrm{~~~\E{E}x(x\neq e1\wedge x\neq e2\,\wedge\,((s(x)=1\wedge none\neq none)}$\\
     && $~~~~~~~~~~~~~~~~~~~~~~~~~~~~~~~~~~~\mrm{\vee\,(s(x)=2\wedge none\neq none)}$\\
     && $~~~~~~~~~~~~~~~~~~~~~~~~~~~~~~~~~~~\mrm{\vee\,(s(x)=3\wedge none\neq none)}$\\
     && $~~~~~~~~~~~~~~~~~~~~~~~~~~~~~~~~~~~\mrm{\vee\,(s(x)\neq 1\wedge s(x)\neq 2\wedge s(x)\neq 3\wedge m_V(s(x))\neq none))))}$\\
    
     & $\equiv$ & $\mrm{\neg\E{E}x(x\neq e1\wedge x\neq e2\wedge s(x)\neq 1\wedge s(x)\neq 2\wedge s(x)\neq 3\wedge m_V(s(x))\neq none)}$
    \end{tabular}\\
    Here, we replace the terms $\mrm{s(e1),s(e2)}$ with node constant $\mrm{1}$, then replace $\mrm{m_V(1), m_V(2), m_V(3)}$ with $\mrm{none}$. Then, we simplify the resulting condition by evaluating $\mrm{none\neq none}$ which is equivalent to $\mrm{false}$.
    
    \item \begin{tabular}[t]{lcl}
    Val$(\text{Split}(s, r_2), r_1)$ & = & $\mrm{false \vee \E{V}x(x\neq 1\wedge\neg root(x))}$\\
    & $\equiv$ & $\mrm{\E{V}x(x\neq 1\wedge\neg root(x))}$\end{tabular}\\
     Here, we substitute $\mrm{false}$ for $\neg root(1)$ since the node 1 in $L$ is a rooted node.

\item \begin{tabular}[t]{lcl}
    Val$(\Gamma_1, r_1)$ & = & $\mrm{d\geq e}$\\
\end{tabular}\\
For this case, we change nothing.

\item \begin{tabular}[t]{lcl}
    Val$(\Gamma_2, r_2)$ & = & $\mrm{outcon(1)\neq 0}$\\
\end{tabular}\\
In this case, we change $\mrm{outdeg(1)}$ with $\mrm{outcon(1)+0}$ because the outdegree of node 1 in $L$ is 0.
\end{enumerate}
\end{example}

Intuitively, Val gives some terms with node/edge constants their value in $L$. Recall that if there exists injective morphism $g:L^\alpha\rightarrow G$ for some label assignment $\alpha_L$, then there must be an inclusion $L^\alpha\rightarrow \RG$. This should assert that the value of terms we valuate in $L$ is equal to their value in $\RG$.

\Lemma{lemma:val}{
Given an unrestricted rule schema $r=\tuple{L\leftarrow K\rightarrow R}$, a host graph $G$, and an injectiva morphism $g:L^\alpha\rightarrow G$ for a label assignment $\alpha_L$. For a graph condition $c$,
\[\rho_g(G)\Sat c \text{ if and only if } \rho_g(G)\Sat (\text{Val}(c,r))^\alpha\]
}

\begin{proof} 
Let us consider the construction of Val$(c)$ step by step. In step 1, we change terms $x$ in $c$ with $T(x)$. Here, we change functions $\mrm{s(e),t(e),l_V(v),m_V(v)}$,\\$\mrm{l_E(e),m_E(e),l_V(s(e)),l_V(t(e)),m_V(s(e)),m_V(t(e))}$ for $e\in E_L$ and $v\in V_L$ with their values in $L$. Since $L^\alpha\rightarrow\RG$ is an inclusion, then $s_L(e)=s_{\RG}(e)$ and $t_L(e)=t_{\RG}(e)$. Also, $(\lst_L(i))^\alpha=\lst_{\RG}(i)$, $(\mrk_L(i))^\alpha=\mrk_{\RG}(i)$ for all $i\in V_G$ and $i\in E_G$ such that the replacement does not change the satisfaction of $c$ in $\RG$. Then for function $\mrm{indeg(v)}$ for $v\in V_L-V_K$, we change it to $indeg_L(x)$ due to the dangling condition, and for $v\in V_K$, we change it to $incon(v)+indeg_L(v)$ which is equivalent to $indeg_G(v)=indeg_{\RG}(v)$ because $incon(v)=indeg_G(v)-indeg_L(v)$ (and analogously for $\mrm{outdeg(v)}$).
In step 2, changing Boolean operators whose arguments are constants to their Boolean value clearly does not change the satisfaction in $\RG$. Also, by the definition of morphism, $p_L(v)=p_{\RG}(v)$ for all $v\in V_L$ so that the Boolean value of $\mrm{root(v)}$ in $L$ is equivalent to the Boolean value of $\mrm{root(v)}$ in $\RG$.
Finally, in step 3, simplification clearly does not change satisfaction.\qed
\end{proof}

Finally, we define the transformation Lift, which takes a precondition and a generalised rule schema as an input and gives a left-application condition as an output. The output should express the precondition, the dangling condition, and the existing left-application condition of the given generalised rule schema.

\Def{Transformation Lift}{def:lift}{
Given a generalised rule $w=\tuple{r,ac_L,ac_R}$ for an unrestricted rule schema $r=\tuple{L\leftarrow K\rightarrow R}$. Let $c$ be a precondition. A left application condition w.r.t. $c$ and $w$, denoted by Lift$(c,w)$, is the condition over $L$:
\[\text{Lift}(c, w)=\text{Val}(\text{Split}(c\wedge ac_L, r)\wedge \text{Dang}(r), r).\]
}

\begin{example}[Transformation Lift]$~$
 \vspace{-\topsep}\begin{enumerate}
    \item $\text{Lift}(q_1, \mtt{del}^\vee)$ \\
    =  $\mrm{\neg\E{E}x(x\neq e1\wedge x\neq e2\wedge s(x)\neq 1\wedge s(x)\neq 2\wedge s(x)\neq 3\wedge m_V(s(x))\neq none)}$\\
        $~~~\wedge\,\mrm{d\geq e}$
         
    \item \begin{tabular}[t]{lcl}
        $\text{Lift}(q_2, \mtt{copy}^\vee)$ & = & $\mrm{\E{V}x(x\neq 1\wedge\neg root(x)) \wedge outcon(1)\neq 0  \wedge true}$\\
        & $\equiv$ & $\mrm{\E{V}x(x\neq 1\wedge\neg root(x)) \wedge outcon(1)\neq 0}$
         \end{tabular}
\end{enumerate}
\end{example}

\Prop{Left-application condition}{prop:lift}{
Given a host graph $G$ and a generalised rule $w=\tuple{r,ac_L,ac_R}$ for an unrestricted rule schema $r=\tuple{L\leftarrow K\rightarrow R}$. Let $c$ be a precondition and $\alpha_L$ be a label assignment such that there exists an injective morphism $g:L^\alpha\rightarrow G$. For some host graph $H$,
\[\text{$G\Sat c$ and $G\Rightarrow_{w,g,g^*}H$ implies $\RG\Sat$(Lift$(c,w))^\alpha$}\]
}

\begin{proof}From Lemma \ref{lemma:split}, we know that $G\Sat c$ implies $\RG\Sat$Split$(c,r)$. Then $G\Rightarrow_{w,g,g^*}H$ implies $\RG\Sat ac_L$, which implies $\RG\Sat$Split$(ac_L,r)$, and the existence of natural double-pushout with match $g:L^\alpha\rightarrow G$. The latter implies the satisfaction of the dangling condition. The satisfaction of the dangling condition implies $\RG\Sat$Dang$(r)$ based on Observation \ref{obsv:dang}, such that $\RG\Sat$Split$(c\wedge ac_L,r)\wedge$Dang$(r)$, and $\RG\Sat$Val(Split$(c,r)\wedge ac_L\wedge$Dang$(r),r)^\alpha$ from Lemma \ref{lemma:val}. \qed
\end{proof}

Recall the construction of Split$(c,r)$ for a precondition $c$ and an unrestricted rule schema $r$. A node/edge quantifier is preserved in the result of the transformation with additional restriction about $x$ not representing any node/edge in $L$. Hence in the resulting condition over $L$ from transformation Lift, every node/edge variable should not represent any node/edge in $L$.

\Obsv{obsv:lift}{
Given a host graph $G$ and a generalised rule $w=\tuple{r^\alpha,ac_L,ac_R}$ for an unrestricted rule schema $r=\tuple{L\leftarrow K\rightarrow R}$, and a precondition $c$. For every node/edge variable $x$ in Lift$(c,w)$, $x$ does not represent any node/edge in $L$.
}

\begin{proof}
Here we show that for every node/edge variable $x$, there exists an existential quantifier over $x$ such that there exists constraint $\bigwedge_{i\in V_{L}x\neq i}$ or $\bigwedge_{i\in E_{L}x\neq i}$ inside the quantifier.

Lift$(c,w)$ is a conjunction of Val$(\text{Split}(c,r), r)$, Dang$(r)$, and Val$(\Gamma, r)$. The transformation Val clearly does not remove or change subformulas in the form $x\neq i$ and does not add any new node/edge variable. Hence, we just need to show that for every node/edge variable $x$ in $\text{Split}(c,r)$, Dang$(r)$, and $\Gamma$, there exists constraint $\bigwedge_{i\in V_{L}x\neq i}$ or $\bigwedge_{i\in E_{L}x\neq i}$.

It is obvious that $\Gamma$ does not have node and edge variable from its syntax. For Dang$(r)$, it clearly only has one edge variable and there exists constraint $\bigwedge_{i\in E_{L}x\neq i}$ inside the existential quantifier for the variable. Finally for Split$(c,r)$, since $c$ is a closed formula, every node/edge variable must be bounded by existential quantifier, such that from Definition \ref{def:split}, the variable must be bounded by existential quantifier with constraint $\bigwedge_{i\in V_{L}x\neq i}$ or $\bigwedge_{i\in E_{L}x\neq i}$ inside.\qed
\end{proof}

\subsection{From left to right-application condition}

To obtain a right-application condition from a left-application condition, we need to consider what properties could be different in the initial and the result graphs. Recall that in constructing a left-application condition, we evaluate all functions with a node/edge constant argument and change them with constant, including the constant $\mrm{incon(v)}$ and $\mrm{outcon(v)}$ when evaluating $\mrm{indeg(v)}$ and $\mrm{outdeg(v)}$ for node $v$ in the interface. In the result graph $H$, $indeg_H(v)$ is clearly equal to $incon(v)+indeg_R(H)$, and analogous for $outdeg_H(v)$. 

The Boolean value for $\mrm{x=i}$ for any node/edge variable $x$ and node/edge constant $i$ not in $R$ must be false in the resulting graph. Analogously, $\mrm{x=i}$ is always true. Also, all variables in the left-application condition should not represent any new nodes and edges in the right-hand side.

\begin{definition}[Adjusment]\label{def:adj}\normalfont
Given an unrestricted rule schema $r=\tuple{L\leftarrow K\rightarrow R}$ and a condition $c$ over $L$. Let $c'$ be a condition over $L$ that is obtained from $c$ by changing every term $\mrm{incon(x)}$ (or $\mrm{outcon(x)}$) for $x\in V_K$ with $\mrm{indeg(x)-}indeg_R(x)$ (or $\mrm{outdeg(x)-}outdeg_R(x)$). Let also $\{v_1,\ldots,v_n\}$ and $\{e_1,\ldots,e_m\}$ denote the set of all nodes and edges in $R-K$ respectively. 
The \textit{adjusted} condition of $c$ w.r.t $r$, denoted by Adj$(c, r)$,  is a condition over $R$ that is defined inductively, where $c_1,c_2$ are conditions over $L$:
 \vspace{-\topsep}\begin{enumerate}
    \item If $c$ is $\mrm{true}$ or $\mrm{false}$, Adj$(c,r)=c'$;
    \item If $c$ is the predicates $\mrm{int(x),char(x),string(x)}$ or $\mrm{atom(x)}$ for a list variable $x$, Adj$(c,r)=c'$;
    \item If $c=\mrm{root(x)}$ for some term $x$ representing a node, Adj$(c,r)=c'$
    \item If $c=x_1\ominus x_2$ for some terms $x_1, x_2$ and $\ominus\in\{=,\neq,<,\leq,>,\geq\}$,\\
    Adj$(c,r)$ =
            $\begin{cases} 
      \mrm{false} & ,\text{if $\ominus\in\{=\}$ and $x_1\in V_L-V_K\cup E_{L}$ or $x_2\in V_L-V_K\cup E_{L}$}, \\
        \mrm{true} & ,\text{if $\ominus\in\{\neq\}$ and $x_1\in V_L-V_K\cup E_{L}$ or $x_2\in V_L-V_K\cup E_{L}$}, \\
      c' & ,\text{otherwise}
   \end{cases}$
    \item Adj$(c_1\vee c_2,r) = \text{Adj}({c_1},r)\vee\text{Adj}({c_2},r)$
    \item  Adj$(c_1\wedge c_2,r) = \text{Adj}({c_1},r)\wedge\text{Adj}({c_2},r)$
    \item Adj$(\neg c_1,r)=\neg\text{Adj}({c_1},r)$
    \item Adj$(\E{V}\mrm{x}(c_1),r)=\E{V}\mrm{x}(x\neq v_1\wedge\ldots\wedge x\neq v_n\wedge\text{Adj}(c_1,r))$
    \item Adj$(\E{E}\mrm{x}(c_1),r)=\E{E}\mrm{x}(x\neq e_1\wedge\ldots\wedge x\neq e_m\wedge\text{Adj}(c_1,r))$
    \item Adj$(\E{L}\mrm{x}(c_1),r)=\E{L}\mrm{x}(\text{Adj}({c_1},r))$\qed
\end{enumerate}
\end{definition}

\begin{example}$~$\\
Let $p_1$ denotes Lift$(q_1,\mtt{del^\vee})$ and $p_2$ denotes Lift$(q_2,\mtt{copy^\vee})$.\\ \begin{tabular}{lcp{10.5cm}}
    1. Adj$(p_1,r_1)$ & = & $\mrm{\neg\E{E}x(x\neq e1\wedge s(x)\neq 1\wedge s(x)\neq 2\wedge m_V(s(x))\neq none)\wedge\,\mrm{d\geq e}}$\\
    2. Adj$(p_2,r_2)$ & = & $\mrm{\E{V}x(x\neq 1\wedge x\neq 2\wedge\neg root(x)) \wedge outdeg(1)\neq 1}$\\
\end{tabular}
\end{example}
         
The main purpose of transformation Adj is to adjust the obtained left-application condition such that it can be satisfied by the replacement graph of the resulting graph.

\Lemma{lemma:adj}{
Given a host graph $G$, a generalised rule $w=\tuple{r,ac_L,ac_R}$ for an unrestricted rule schema $r=\tuple{L\leftarrow K\rightarrow R}$, an injective morphism $g:L^\alpha\rightarrow G$ for some label assignment $\alpha_L$, and a precondition $d$. Let $H$ be a host graph such that $G\Rightarrow_{w,g,g^*}H$ for some injective morphism $g^*:R^\beta\rightarrow H$ where $\beta_R(i)=\alpha_L(i)$ for all common item $i$ in domain $\beta_R$ and $\alpha_L(i)$. Then,
\[\text{$\RG\Sat($Lift$(d,w))^\alpha$ implies $\RH\Sat$(Adj(Lift($d,w),r))^\beta$}\]
}

\begin{proof}
Note that Adj$(c,r)$ does not change any term representing label in $c$ such that Adj$(c^\alpha,r)\equiv$Adj$(c,r)^\alpha$ for all label assignment $\alpha_L$. Also, note that Adj$(c,r)$ does not contain any variable $x$ in $R$ that does not exist in $L$. Hence, Adj$(c,r)^\alpha=$Adj$(c,r)^\beta$.
Assuming $\RG\Sat c^\alpha$ for $c=$Lift$(c,w)$, we prove that $\RH\Sat$(Adj$(c,r))^\beta$ inductively bellow:\\
\noindent Base case.
 \vspace{-\topsep}\begin{enumerate}
    \item If $c$ is $\mrm{true}$ or $\mrm{false}$, it is obvious that the lemma holds as every graph satisfies $\mrm{true}$ and no graph satisfies $\mrm{false}$
    \item If $c$ is the predicate $\mrm{int(x),char(x),string(x)}$ or $\mrm{atom(x)}$ for a list variable $x$, $c'\equiv c$ and satisfaction of $c$ is independent on the host graph such that $\RG\Sat c^\alpha$ implies $\RH\Sat c'^\alpha$ and $c'^\alpha=c'^\beta$.
    \item If $c$ is the predicate $\mrm{root(x)}$ for some term $x$ representing a node, then $x\notin V_L$ (see Definition $\ref{def:val}$ point 2), $x$ is a variable representing $V_{\RG}-(V_L)=V_{\RH}-V_R$ (see Observation \ref{obsv:lift}), or $x$ is the function $\mrm{s(x)}$ or $\mrm{t(x)}$ for some edge variable $x$ representing an edge in $E_{\RG}-E_L=E_{\RH}-E_R$ (see Definition $\ref{def:val}$ point 1(b) and \ref{obsv:lift}). Hence, $x$ representing a node in $\RG-L$, which is also in $\RH-R$ so that if $\mrm{root(x)}$ is true in $\RG$, $\mrm{root(x)}$ must be true in $\RH$, and label assignment has nothing to do with this.
    \item If $c=x_1\ominus x_2$, if $x_1$ and $x_2$ are terms representing lists, then $x_1$ and $x_2$ independent to nodes and edges in $V_L$ unless $x_1$ or $x_2$ is in the form $\mrm{incon(v)} or \mrm{outcon(v)}$ for some $v\in V_K$ (see Definition $\ref{def:val}$ point 1(b) and \ref{obsv:lift}). However, because $outcon(v)=outdeg_{\RG}(v)-outdeg_L(v)=outdeg_{\RH}(v)-outdeg_R(v)$, then semantics of $\mrm{outcon(v)}$ in $\RG$ is equivalent to semantics of $\mrm{indeg(v)-}indeg_R(v)$ in $\RH$. Hence, $c$ is either independent to nodes and edges in $V_L$ or contain $\mrm{outcon(x)}$ or $\mrm{incon(x)}$, $\RG\Sat c$ implies $\RH\Sat c'=$Adj$(c,r)$, or $c$. If $c$ is $x_1=x_2$ and $x_1$ or $x_2$ is a constant in $(V_L-V_K)$ or in $E_L$, it is obvious that there in no node/edge in $\RH$ that is equal to the constant such that $\RH\Sat\mrm{false}=$Adj$(c,r)$. Analogously, if $c$ is $x_1\neq x_2$ and $x_1$ or $x_2$ is a constant in $(V_L-V_K)$ or in $E_L$, every node/edge in $\RH$ does not equal to the node or edge such that $\RH\Sat\mrm{true}=$Adj$(c,r)$.
    \end{enumerate}
    Inductive case. Assuming $\RG\Sat c_1^\alpha$ implies $\RH\Sat$Adj$(c_1,r)^\beta$ and  $\RG\Sat c_2^\alpha$ implies $\RH\Sat$Adj$(c_2,r)^\beta$ for some conditions $c_1,c_2$ over $L$,
    \begin{enumerate}
    \item $\RG\Sat (c_1\vee c_2)^\alpha$ implies $\RG\Sat c_1^\alpha$ or $\RG\Sat c_2^\alpha$ implies 
    $\RH\Sat$Adj$(c_1,r)^\beta$ or $\RH\Sat$Adj$(c_2,r)^\beta$, implies $\RH\Sat$(Adj$(c_1,r)\vee$Adj$(c_2,r))^\beta$.
    \item $\RG\Sat (c_1\wedge c_2)^\alpha$ implies $\RG\Sat^{\mu_V,\mu_E} c_1^\alpha$ and $\RG\Sat^{\mu_V,\mu_E} c_2^\alpha$ for some assignments $\mu_V,\mu_E$ which implies 
    $\RH\Sat^{\mu_V,\mu_E}$Adj$(c_1,r)^\beta$ and $\RH\Sat^{\mu_V,\mu_E}$Adj$(c_2,r)^\beta$ implies $\RH\Sat$(Adj$(c_1,r)\wedge$Adj$(c_2,r))^\beta$
    \item $\RG\Sat \neg c_1^\alpha$ implies $\neg(\RG\Sat{\mu_V,\mu_E} c_1^\alpha)$ for some assignments $\mu_V,\mu_E$ which implies 
    $\neg(\RH\Sat^{\mu_V,\mu_E}$(Adj$(c_1,r))^\beta)$, implying $\RH\Sat\neg$(Adj$(c_1,r))^\beta$
    \item If $c=\E{V}x(c_1)$, recall that every node variable $x$ in $c$ does not represent node in $L$. $\RG\Sat(\E{V}x(c_1))^\alpha$ implies $\RG\Sat (c_1^{x\mapsto v})^\alpha$ for some $v\in V_{\RG}-V_L=V_{\RH}-V_R$ which implies $\RH\Sat$(Adj$(c_1^{[x\mapsto v]}, r))^\beta$. Since $v\notin V_{R},$ $\RH\Sat(\E{V}x(x\neq v_1\wedge\ldots\wedge x\neq v_n\wedge$Adj$(c_1, r)))^\beta$
    \item If $c=\E{V}x(c_1)$, the proof is analogous to above
    \item $\RG\Sat(\E{L}x(c_1))^\alpha$ implies $\RG\Sat (c_1^{x\mapsto k})^\alpha$ for some $k\in \mathbb{L}$ which implies $\RH\Sat$(Adj$(c_1^{[x\mapsto k]}, r))^\beta$=(Adj$(c_1, r)^{[x\mapsto k]})^\beta$, which means $\RH\Sat(\E{L}x($Adj$(c_1, r)))^\beta$.\qed
    \end{enumerate}
\end{proof}

Note that any unrestricted rule schema $r$ is invertible. The transformation Adj adjusts a left-application condition to the properties of the resulting graph w.r.t the given unrestricted rule schema. This means, adjusting the properties of the resulting graph w.r.t the inverse of the unrestricted rule schema should resulting in the initial left-application condition.

\Lemma{lemma:adj2}{
Given host graph $G$, a generalised rule $w=\tuple{r,ac_L,ac_R}$ for an unrestricted rule schema $r=\tuple{L\leftarrow K\rightarrow R}$, and a precondition $d$. Let $g:L^\alpha\rightarrow R$ for some label assignment $\alpha_L$ be an injective morphism satisfying the dangling condition. Then
\[\text{$\RG\Sat$ Adj(Adj$($Lift$(d,w),r),r^{-1})^\alpha$ if and only if $\RG\Sat$Lift$(d,w)^\alpha$}\]
}

\begin{proof}
Here we prove that $\RG\Sat$ Adj(Adj$(c,r),r^{-1})$ if and only if $\RG\Sat c$ inductively, where $c=$Lift$(d,r)$:\\
Base case.
 \vspace{-\topsep}\begin{enumerate}
    \item If $c$ is $\mrm{true}$ or $\mrm{false}$, Adj$(c,r)=c'=$Adj(Adj$(c,r),r^{-1})$
    \item If $c$ is the predicate $\mrm{int(x),char(x),string(x)}$ or $\mrm{atom(x)}$ for a list variable $x$, $c'\equiv c$ such that Adj$(c,r)=c'=$Adj(Adj$(c,r),r^{-1})$
    \item If $c$ is the predicate $\mrm{root(x)}$, $c'\equiv c$ such that Adj$(c,r)=c'=$Adj(Adj$(c,r),r^{-1})$
    \item If $c$ is $x_1=x_2$ for $x_1$ or $x_2$ a node or edge constant in $L-K$, both $x_1$ and $x_2$ cannot be constants (see construction of Val which is used to construct $c$). Then, one of them must be a node or edge variable (which does not represent node in $L$ - see Observation \ref{obsv:lift}), or the function $\mrm{s(x)}$ or $\mrm{t(x)}$ for some edge variable $x$. Observation \ref{obsv:lift} shows us that $x$ does not representing edge in $L$, and $g$ satisfies the dangling condition implies $\mrm{s(x)}$ and $\mrm{t(x)}$ do not represent nodes in $\RG-(L-K)$. Hence, $x_1=x_2$ is always false in $\RG$. Otherwise for $c=x_1\ominus x_2$, Adj$(c,r)=c'=$Adj(Adj$(c,r),r^{-1})$. 
    \end{enumerate}
    Inductive case.\\
    Assume that $c_1\equiv$Adj(Adj$(c_1,r),r^{-1})$ and $c_2\equiv$Adj(Adj$(c_2,r),r^{-1})$ for conditions $c_1,c_2$ over $L$.
    \begin{enumerate}
    \item $\RG\Sat c_1\vee c_2$ iff $\RG\Sat c_1$ or $\RG\Sat c_2$ iff $\RG\Sat$Adj(Adj$(c_1,r),r^{-1})$ or $\RG\Sat$Adj(Adj$(c_2,r),r^{-1})$ iff $\RG\Sat$Adj(Adj$(c_1,r),r^{-1})\vee$\\ Adj(Adj$(c_2,r),r^{-1})\equiv$Adj(Adj$(c,r),r^{-1})$
    \item $\RG\Sat c_1\wedge c_2$ implies $\RG\Sat^\beta c_1 \wedge \RG\Sat^\beta c_2$ for some assignment $\beta$ iff 
    $\RG\Sat^\beta$Adj(Adj$(c_1,r),r^{-1})\wedge\RG\Sat^\beta$Adj(Adj$(c_2,r),r^{-1}))$ iff $\RG\Sat$Adj(Adj$(c_1,r),r^{-1})\wedge$Adj(Adj$(c_2,r),r^{-1}))\equiv$Adj(Adj$(c,r),r^{-1})$
    \item $\RG\Sat \neg c_1$ iff $\neg(\RG\Sat\beta c_1)$ for some assignment $\beta$ iff
    $\neg(\RG\Sat^\beta$Adj(Adj$(c_1,r),r^{-1})$, iff $\RG\Sat\neg$Adj(Adj$(c_1,r),r^{-1})\equiv$Adj(Adj$(c,r),r^{-1})$
    \item If $c=\E{V}x(c_1)$, Adj$(c,r)=\E{V}\mrm{x}(x\neq v_1\wedge\ldots\wedge x\neq v_n\wedge\text{Adj}(c_1,r))$, so that Adj(Adj$(c,r),r^{-1})=\E{V}\mrm{x}(\text{Adj(Adj}(c_1,r),r^{-1}))$. Hence, $\RG\Sat$Adj(Adj$(c,r),r^{-1})$ iff $\RG\Sat\E{V}\mrm{x}(\text{Adj(Adj}(c_1,r),r^{-1}))$ iff $\RG\Sat\E{V}\mrm{x}(c_1)=c$
    \item If $c=\E{V}x(c_1)$, the proof is analogous to above
    \item If $c=\E{L}x(c_1)$, $\RG\Sat$Adj(Adj$(c,r),r^{-1})$ iff $\RG\Sat\E{L}\mrm{x}(\text{Adj(Adj}(c_1,r),r^{-1}))$ iff $\RG\Sat\E{L}\mrm{x}(c_1)=c$
    \end{enumerate}
Since the construction of Adj(Adj($c,r),r^{-1})$ does not any term representing labels, Adj(Adj($c^\alpha,r),r^{-1})\equiv$Adj(Adj($c,r)^\alpha,r^{-1})\equiv$Adj(Adj($c,r),r^{-1})^\alpha$. Hence, the lemma is valid.\qed
\end{proof}

Actually, from the transformation Adj we already obtain a right-application condition. However, we want a stronger condition such that we add the specification of the right-hand graph. In addition, since the resulting graph should also satisfy the existing right-application of the given generalised rule schema, and the comatch should also satisfy the dangling condition.

\Def{Shifting}{def:shift}{
Given a generalised rule $w=\tuple{r,ac_L,ac_R}$ for an unrestricted rule schema $r=\tuple{L\leftarrow K\rightarrow R}$, and a precondition $c$. Right application condition w.r.t. $c$ and $w$, denoted by Shift$(c,w)$, is defined as:
\[\text{Shift$(c,w)=$Adj(Lift$(c,w),r)\wedge\, ac_R\, \wedge\,$Spec$(R)\,\wedge\,$Dang$(r^{-1})$}.\]
}

\begin{example}$~$\\
\begin{tabular}{lcp{10.5cm}}
    Shift$(q_1,\mtt{del}^\vee)$&=&$\mrm{\neg\E{E}x(x\neq e1\wedge s(x)\neq 1\wedge s(x)\neq 2\wedge m_V(s(x))\neq none)\wedge\,\mrm{d\geq e}}$\\
    &&$\mrm{\wedge\, \lV(1)=a\wedge\lV(2)=b\wedge\lE(e1)=d+e\wedge\mV(1)=red}$\\
    &&$\mrm{\wedge\mV(2)=none\wedge\mE(e1)=none\wedge s(e1)=1\wedge t(e1)=2}$\\
    &&$\mrm{\wedge \neg root(1)\wedge\neg root(2)\wedge int(d)\wedge int(e)}$\\
    Shift$(q_2,\mtt{del}^\vee)$&=&$\mrm{\E{V}x(x\neq 1\wedge x\neq 2\wedge\neg root(x)) \wedge outdeg(1)\neq 1}$\\
    &&$\mrm{\wedge\,\lV(1)=a\wedge\lV(2)=a\wedge\lE(e1)=empty\wedge\mV(1)=none}$\\
    &&$\mrm{\wedge\,\mV(2)=none\wedge\mE(e1)=dashed\wedge s(e1)=1\wedge t(e1)=2}$\\
    &&$\mrm{\wedge\,\neg root(1)\wedge root(2)\wedge indeg(2)=1\wedge outdeg(2)=0}$
\end{tabular}
\end{example}

\Prop{Shifting}{prop:shift}{
Given a host graph $G$, a generalised rule $w=\tuple{r,ac_L,ac_R}$ an unrestricted rule schema $r=\tuple{L\leftarrow K\rightarrow R}$, an injective morphism $g:L^\alpha\rightarrow G$ for some label assignment $\alpha_L$, and a precondition $d$.
Then for host graphs $H$ such that $G\Rightarrow_{w,g,g^*}H$ with an right morphism $g^*:R^\beta\rightarrow H$ where $\beta_R(i)=\alpha_L(i)$ for every variable $i$ in $L$ such that $i$ in $R$, and for every node/edge $i$ where $\mrk_L(i)=\mrk_R(i)=\mtt{any}$,
\[\text{$\RH\Sat($Adj(Lift$(d,w)),r)^\beta$ if and only if $\RH\Sat$(Shift$(d,w))^\beta$}\]
}

\begin{proof}
It is obvious that Adj(Lift$(d,w)),r)^\beta$ is implied by Shift$(d,w)^\beta$, so now we show that Adj(Lift$(d,w)),r)^\beta$ implies Shift$(d,w)^\beta$. That is, $ac_R^\beta\, \wedge\,$Spec$(R)^\beta\,\wedge\,$ Dang$(r^{-1})^\beta$ is satisfied by $\RH$. From Definition \ref{def:generalisedrPO}, $G\Rightarrow_{w,g,g^*}H$ implies $\RH\Sat ac_R^\beta$. From the construction of Spec$(R)$, Spec$(R)^\beta\equiv$Spec$(R^\beta)$ such that $\RH\Sat$Spec$(R)^\beta$ is implied by the injective morphism $g^*$. Finally, there is no label variable in Dang$(r^{-1})$ such that Dang$(r^{-1})\equiv$Dang$(r^{-1})^\beta$, which is implied by $G\Rightarrow_{w,g,g^*}H$ because nodes in $R-K$ must not incident to any edge in $\RH-R$ so that their indegree and outdegree in $R$ represents their indegree and outdegree in $\RH$.\qed
\end{proof}

\subsection{From right-application condition to postcondition}

The right-application condition we obtain from transformation Shift is strong enough to express properties of the replacement graph of any resulting graph. However, since we need a condition (without node/edge constant), we define transformation Post.

\Def{Formula Post}{def:post}{
Given a generalised rule $w=\tuple{r,ac_L,ac_R}$ for an unrestricted rule $r=\tuple{L\leftarrow K\rightarrow R}$ and a precondition $c$. A postcondition w.r.t. $c$ and $w$, denoted by Post$(c,w)$, is the FO formula:
\[\text{Post}(c, w)=\E{V}x_1,\ldots,x_n(\E{E}y_1,\ldots,y_m(\E{L}z_1,\ldots,z_k(\text{Var(Shift}(c,w))))).\]
where $\{x_1,\ldots,x_n\}$, $\{y_1,\ldots,y_m\}$, and $\{z_1,\ldots,z_k\}$ denote the set of free node, edge, and label (resp.) variables in Var(Shift$(c,w)$). We then denote by Slp$(c,r)$ the formula Post$(c,r^\vee)$, and  Slp$(c,r^{-1})$ for the formula Post$(c,(r^\vee)^{-1})$.}

To obtain a closed FO formula from the obtained right-application condition, we only need to variablise the node/edge constants in the right-application condition, then put an existential quantifier for each free variable in the resulting FO formula.

\begin{example}$~$\\
\begin{tabular}{lcp{10.5cm}}
    Post$(q_1,\mtt{del}^\vee)$&=&$\mrm{\E{V}u,v(u\neq v\wedge \E{E}w(\E{L}a,b,d,e(}$\\
    &&$\mrm{\neg\E{E}x(x\neq w\wedge s(x)\neq u\wedge s(x)\neq v\wedge m_V(s(w))\neq none)\wedge\,\mrm{d\geq e}}$\\
    &&$\mrm{\wedge\, \lV(u)=a\wedge\lV(v)=b\wedge\lE(w)=d+e\wedge\mV(u)=red}$\\
    &&$\mrm{\wedge\mV(v)=none\wedge\mE(w)=none\wedge s(w)=u\wedge t(w)=v}$\\
    &&$\mrm{\wedge \neg root(u)\wedge\neg root(v)\wedge int(d)\wedge int(e))))}$\\
    Post$(q_2,\mtt{del}^\vee)$&=&
    $\mrm{\E{V}u,v(u\neq v\wedge\E{E}w(\E{L}a(}$
    \\
    &&$\mrm{\E{V}x(x\neq u\wedge x\neq v\wedge\neg root(x)) \wedge outdeg(u)\neq 1}$\\
    &&$\mrm{\wedge\,\lV(u)=a\wedge\lV(v)=a\wedge\lE(w)=empty\wedge\mV(u)=none}$\\
    &&$\mrm{\wedge\,\mV(v)=none\wedge\mE(w)=dashed\wedge s(w)=u\wedge t(w)=v}$\\
    &&$\mrm{\wedge\,\neg root(u)\wedge root(v)\wedge indeg(v)=1\wedge outdeg(v)=0})))$
\end{tabular}
\end{example}

\Prop{Post}{prop:post}{
Given a host graph $G$, a generalised rule $w=\tuple{r,ac_L,ac_R}$ for an unrestricted rule schema $r=\tuple{L\leftarrow K\rightarrow R}$, and a precondition $c$. Then for all host graph $H$ such that there exists an injective morphism $g^*:R^\beta\rightarrow H$ for a label assignment $\beta_R$,
\[\text{$\RH\Sat$(Shift$(c,w))^\beta$ if and only if $H\Sat$Post$(c, w)^\beta$}\]}

\begin{proof}From Lemma \ref{lemma:var}, $\RH\Sat$(Shift$(c,w))^\beta$ iff $H\Sat$Var(Shift$(c,w))^\beta$. If there is no node (or edge) in $H$, then there is no node (or edge) constant in $\RH$ since they are isomorphic. Hence, there is no free node (or edge) variable in Var(Shift$(c,w)))^\beta$ so that there is no additional node (or edge) quantifier for Var(Shift$(c,w)))^\beta$. If there exists a node (or edge) in $H$, then from Lemma \ref{lemma:BolOperators}, adding an existential quantifier will not change its satisfaction on $H$. Hence, $H\Sat$Var(Shift$(c,w)))^\beta$ iff $H\Sat$Post$(c, w)^\beta$.\qed
\end{proof}

Finally, we show that Post$(c,r^\vee)$ is a strongest liberal postcondition w.r.t. $c$ and $r$. That is, by showing that for all host graph $G$, $G\Sat c$ and $G\Rightarrow_rH$ implies $H\Sat$Post$(c,r^\vee)$, and showing that for all host graph $H$, $H\Sat$Post$(c,r^\vee)$ implies the existence of host graph $G$ such that $G\Sat c$ and $G\Rightarrow_rH$.

\Theo{Strongest liberal postconditions}{theo:slp}{
Given a precondition $c$ and a conditional rule schema $r=\tuple{\langle L \leftarrow K\rightarrow R\rangle,\Gamma}$. Then, Slp$(c,r)$ is a strongest liberal postcondition w.r.t. $c$ and $r$.}

\begin{proof}
From Lemma \ref{lemma:rvee}, $G\Rightarrow_r H$ iff $G\Rightarrow_{w,g,g*} H$ for some injective morphisms $g:L^\alpha\rightarrow G$ and $g^*:R^\beta\rightarrow H$ with label assignment $\alpha_L$ and $\beta_R$ where $\beta_R(i)=\alpha_L(i)$ for every variable $i$ in $L$ such that $i$ is in $R$, and for every node/edge $i$ where $\mrk_L(i)=\mrk_R(i)=\mtt{any}$. From Proposition \ref{prop:lift}, Lemma \ref{lemma:adj}, Proposition \ref{prop:shift}, and Proposition \ref{prop:post},
$G\Sat c$ and $G\Rightarrow_{r^\vee,g,g^*} H$ implies $\RG\Sat$(Lift$(c,r^\vee))^\alpha$ implies $\RH\Sat$Shift$(c,r^\vee)^\beta$ implies $H\Sat$Post$(c,r^\vee)$. Hence, Post$(c,r^\vee)$ is a liberal postcondition w.r.t. $c$ and $r$.

To show that Post$(c,r^\vee)$ is a strongest liberal postcondition, based on Lemma \ref{lemma:slp}, we need to show that  for every graph $H$ satisfying Post$(c,r^\vee)$, there exists a host graph $G$ satisfying $c$ such that $G\Rightarrow_{r}H$. 

Recall the construction of Shift$(c,r^\vee)$. A graph satisfying Shift$(c,r^\vee)$ must satisfying Spec$(R)$ such that $H\Sat$(Post$(c,r^\vee))$ implies $H\Sat$(Post$(c,r^\vee))^\beta$  for some label assignment $\beta_R$, which implies $H\Sat$Var(Spec$(R))^\beta\equiv$Var(Spec$(R^\beta))$. From Lemma \ref{lemma:var}, this implies the existence of an injective morphism $g^*:R^\beta\rightarrow H$. 
From Proposition \ref{prop:post}, $H\Sat$Post$(c,r^\vee)^\beta$ implies $\RH\Sat$Shift$(c,r^\vee)^\beta$. From the construction of Shift$(c,r^\vee)$, Dang$(r^{-1})$ asserts that the dangling condition is satisfied by $g^*$. Hence, there exists a natural double-pushout bellow where every morphism is inclusion:

\begin{center}
        \begin{tikzpicture}[scale=0.9, transform shape]
	\node (2) at (1, 0) {(1)};
	\node (a) at (2, 0.75) {$K$};
	\node (b) at (0, 0.75) {$R^\beta$};
	\node (c) at (4, 0.75) {$L^\alpha$};
	\node (d) at (2, -0.75) {$D$};
	\node (1) at (3, 0) {(2)};
	\node (e) at (4, -0.75) {$A$};
	\node (f) at (0, -0.75) {$\RH$};
    \draw[->] (a) to node {} (b);
	\draw[->] (a) to node {} (c);
	\draw[->] (a) to node {} (d);
	\draw[->] (d) to node {} (e);
	\draw[->] (d) to node {} (f);
	\draw[->] (c) to node {} (e);
	\draw[->] (b) to node {} (f);
\end{tikzpicture}\end{center}

Since $\RH\Sat$Adj(Lift($c,r^\vee),r)^\beta$, from Lemma \ref{lemma:adj} this implies $A$ satisfies Adj(Adj(Lift($c,r^\vee),r),r^{-1})^\alpha$. From Lemma \ref{lemma:adj2}, this implies $A\Sat c^\alpha$. Since direct derivations are invertible, $A\Rightarrow_{r^\vee,g,g^*}H$. Hence, $A\Rightarrow_r H$.\qed
\end{proof}

%% file: input/5_Verification.tex
In this section, we introduce semantic and syntactic partial correctness calculus. The former consider arbitrary assertion language as pre- and postconditions, while the latter consider conditions (i.e. closed firs-order formulas) as the pre- and postconditions.

\subsection{Semantic partial correctness calculus}
\label{sec:sem}
For a graph program $P$ and assertions $c$ and $d$, triple $\{c\}\,P\,\{d\}$ is partially correct iff for all graph satisfying $c$, $H\in\Sem{P}G$ implies $H\Sat d$. 

\begin{definition}[Partial correctness \cite{PoskittP12}]\normalfont \label{def:par}
A graph program $P$ is \emph{partially correct} with respect to a precondition $c$ and a postcondition $d$, denoted by $\vDash \{c\}~P~\{d\}$ if for every host graph $G$ and every graph $H$ in $\llbracket P\rrbracket G$, $G\models c$ implies $H\models d$. \qed
\end{definition}  

To prove that $\vDash \{c\}~P~\{d\}$ holds for some assertions $c,d$, and a graph program $P$, we use two methods: 1) finding a strongest liberal postcondition w.r.t $c$ and $P$ and prove that the strongest liberal postcondition implies $d$, and 2) using proof rules for graph programs, create a proof tree to show the partial correctness. The first method has been done in classical programming \cite{DijkstraS90,Jones-Roscoe-Wood10a}, while the second has been done in graph programming \cite{Poskitt13} but without the special command $\mtt{break}$.

In the previous section, we have defined a strongest liberal postcondition w.r.t. a precondition and a conditional rule schema. In this section, we extend the definition from conditional rule schemata to graph programs. In addition, we also introduce a weakest liberal precondition over a graph program.

\Def{Strongest liberal postconditions}{def:slpP}{
A condition $d$ is a \emph{liberal postcondition} w.r.t. a precondition $c$ and a graph program $P$, if for all host graphs $G$ and $H$,
\[G\vDash c \text{ and } H\in\Sem{P}G \text{ implies }H\vDash d.\]
A \emph{strongest liberal postcondition} w.r.t. $c$ and $P$, denoted by SLP$(c,P)$, is a liberal postcondition w.r.t. $c$ and $P$ that implies every liberal postcondition w.r.t. $c$ and $P$.
}

\Def{Weakest liberal preconditions}{def:wlpP}{
A condition $c$ is a \emph{liberal precondition} w.r.t. a postcondition $d$ and a graph program $P$, if for all host graphs $G$ and $H$,
\[G\vDash c \text{ and } H\in\Sem{P}G \text{ implies }H\vDash d.\]
A \emph{weakest liberal precondition} w.r.t. $d$ and $P$, denoted by WLP$(P,d)$, is a liberal precondition w.r.t. $d$ and $P$ that is implied by every liberal postcondition w.r.t. $d$ and $P$.
}

\begin{lemma}\normalfont\label{lemma:slpP}
Given a graph program $P$ and a precondition $c$. Let $d$ be a liberal postcondition w.r.t. $c$ and $P$. Then $d$ is a strongest liberal postcondition w.r.t. $c$ and $P$ if and only if for every graph $H$ satisfying $d$, there exists a host graph $G$ satisfying $c$ such that $H\in\Sem{P}G$.
\end{lemma}

\begin{proof}$~$\\
(If).\\
Assuming it is true that for every graph $H$ satisfying $d$, there exists a host graph $G$ satisfying $c$ such that $H\in\Sem{P}G$.
Let $H$ be a host graph satisfying $d$. From the assumption, there exists a graph $G$ such that $G\vDash c$ and $H\in\Sem{P}G$. Since $H\in\Sem{P}G$, $H\vDash a$ for all liberal postcondition $a$ over $c$ and $P$. Hence, $H\Sat d$ implies $H\Sat a$ for all liberal postcondition $a$ over $c, P$ such that $d$ is a strongest postcondition w.r.t. $c$ and $P$\\
(Only if).\\
Assume that it is not true that for every host graph $H$, $H\vDash d$ implies there exists a host graph $G$ satisfying $c$ such that $H\in\Sem{P}G$. We show that a graph satisfying $d$ can not imply the graph satisfying all liberal postcondition w.r.t $r$ and $P$. From the assumption, there exists a host graph $H$ such that every host graph $G$ does not satisfy $c$ or $H\notin\Sem{P}G$. In the case of $H\notin\Sem{P}G$, we clearly can not guarantee characteristic of $H$ w.r.t. $P$. Then for the case where $G$ does not satisfy $c$, we also can not guarantee the satisfaction of any liberal postcondition $a$ over $c$ in $H$ because $a$ is dependent of $c$. Hence, we can not guarantee that $H$ satisfying all liberal postcondition w.r.t. $r$ and $c$.\qed
\end{proof}

\begin{lemma}\normalfont\label{lemma:wlpP}
Given a graph program $P$ and a postcondition $d$. Let $c$ be a liberal precondition w.r.t. $P$ and $d$. Then $c$ is a weakest liberal precondition w.r.t. $P$ and $d$ if and only if for every graph $G$ $G\Sat c$ if and only if for all host graphs $H$, $H\in\Sem{P}G$ implies $H\Sat d$.
\end{lemma}

\begin{proof}
(If).\\
Suppose that $G\Sat c$ iff for all host graphs $H$, $H\in\Sem{P}G$ implies $H\Sat d$. It implies for all host graphs $H$, $G\Sat c$ and $H\in\Sem{P}G$ implies $H\Sat d$. From Definition \ref{def:wlpP}, $c$ is a liberal precondition. Let $a$ be a liberal precondition w.r.t. $P$ and $d$ as well. From Definition \ref{def:wlpP}, for all host graphs $H$, $H\in\Sem{P}G$ implies $H\Sat d$, and from the premise, $G\Sat c$. Hence, $c$ is a weakest liberal precondition.
(Only if).\\
Suppose that $c$ is a weakest liberal precondition. From Definition \ref{def:wlpP}, if $G\Sat c$ then $H\in\Sem{P}G$ implies $H\Sat d$. Let $a$ be a liberal precondition w.r.t $P$ and $d$. From Definition \ref{def:wlpP}, $G\Sat a$ implies for all $H$, $H\in\Sem{P}G$ implies $H\Sat d$. Since for all $a$, $G\Sat a$ must imply $G\Sat c$, then $H\in\Sem{P}G$ implies $H\Sat d$ must imply $G\Sat c$ as well.
\qed
\end{proof}

SLP and WLP for a loop $P!$ is not easy to construct because $P!$ may get stuck or diverge. In \cite{Pennemann09}, the divergence is represented by infinite formulas while in \cite{Jones-Roscoe-Wood10a}, it is represented by recursive equation that is not well-defined. In this paper, for practical reason we only consider strongest liberal postconditions over loop-free graph programs.

For the conditional commands $\mtt{if/try-then-else}$, the execution of the command depends on the existence of a proper host graph as a result of executing a graph program. In \cite{Poskitt13}, there is an assertion representing a condition that must be satisfied by a graph such that there exists a path to successful execution, and there is also an assertion representing a condition that must be satisfied by a host graph such that there exist a path to a failure. Here, we define assertion SUCCESS for the former and FAIL for the latter.

\Def{Assertion SUCCESS}{def:assSE}{
For a graph program $P$, SUCCESS$(P)$ is the predicate on host graphs where for all host graph $G$,
\[G\Sat\text{SUCCESS}(P)\text{~if and only if there exists a host graph $H$ with~} H\in\llbracket P\rrbracket G.\]
}

\Def{Assertion FAIL}{def:assFE}{
Given a graph program $P$. FAIL$(P)$ is the predicate on host graphs where for all host graph $G$,
\[G\Sat\text{FAIL}(P)\text{~if and only if~} \text{fail}\in\llbracket P\rrbracket G.\qed\]
}

Note that for a graph program $C$, $\FAIL(C)$ does not necessarily equivalent to $\neg\SUCCESS(C)$, e.g. if $C=\mtt{\{nothing,add\};zero}$ where
$\mtt{nothing}$ is the rule schema where the left and right-hand graphs are the empty graph, $\mtt{add}$ is the rule schema where the left-hand graph is the empty graph and the right-hand graph is a single 0-labelled unmarked and unrooted node, and $\mtt{zero}$ is a rule schema that match with the a 0-labelled unmarked and unrooted node.
For a host graph $G$ where there is no 0-labelled unmarked unrooted node, there is a derivation $\tuple{C,G}\rightarrow^* H$ for some host graph $H$ but also a derivation $\tuple{C,G}\rightarrow^* \mtt{fail}$ such that $G\Sat\SUCCESS(C)$ and $G\Sat\FAIL(C)$.

Having a strongest liberal postcondition over a loop-free program $P$ w.r.t a precondition $c$ allows us to prove that the triple $\{c\}\,P\,\{d\}$ for an assertion $d$ is partially correct. That is, by showing that $d$ is implied by the strongest liberal postcondition. 

\Prop{Strongest liberal postcondition for loop-free programs}{prop:slpP}{
Given a precondition $c$ and a loop-free program $S$. Then, the following holds:\\
 \vspace{-\topsep}\begin{enumerate}
    \item If $S$ is a set of rule schemata $\mathcal{R}=\{r_1,\ldots,r_n\}$,\\
    SLP$(c,\mathcal{R})=\begin{cases}
    \SLP(c,r_1)\vee\ldots\vee\SLP(c,r_n) &, \text{, if~}n>0,\\
    \mrm{false}&\text{, otherwise}\end{cases}$
    \item For loop-free programs $C, P,$ and $Q$,
    \begin{enumerate}
        \item[(i)] If $S=P\mtt{~or~}Q,$\\
        $\SLP(c,S)=\SLP(c,P)\,\vee\,\SLP(c,Q)$
        \item[(ii)] If $S=P;Q,$\\
        $\SLP(c,S)=\SLP(\SLP(c,P),Q)$
        \item[(iii)] If $S=\mtt{if\,}C\mtt{\,then\,}P\mtt{\,else\,}Q,$\\
        $\SLP(c,S)=\SLP(c\wedge \SUCCESS(C),P)\vee\SLP(c\wedge \FAIL(C),Q)$
        \item[(iv)] If $S=\mtt{try\,}C\mtt{\,then\,}P\mtt{\,else\,}Q,$\\
        $\SLP(c,S)=\SLP(c\wedge \SUCCESS(C),C;P)\vee\text{SLP}(c\wedge \FAIL(C),Q)$
    \end{enumerate}
\end{enumerate}
}

Computing SLP$(c,\mathcal{R})$ for a set of rule schemata $\mathcal{R}$ is basically disjunct all strongest liberal postcondition w.r.t $c$ and each rule schema in $\mathcal{R}$. If the rule set is empty, then SLP$(c,\mathcal{R})$ is $\mrm{false}$ since there is nothing to disjunct. Computing SLP$(c,P;Q)$ is constructed by having SLP$(c,P)$ and then find strongest liberal postcondition w.r.t. $Q$ and the resulting formula. 

The equation for program composition is the same with the equation for program composition in \cite{DijkstraS90,Jones-Roscoe-Wood10a}. However, for $\mtt{if-then-else}$ command, the command $\mtt{if}$ in classical programming is followed by an assertion while in graph programs it is followed by a graph program. Hence, instead of checking the truth value of the assertion on the input graph, we check the check if the satisfaction of $\SUCCESS$ and $\FAIL$ of the associated program on the input graph. Then for $\mtt{try-then-else}$ command, it does not exist in classical programming, but we have the equation for the command based on its similarity with $\mtt{if-then-else}$.

The execution of if/try commands yields two possibilities of results, so we need to check the strongest liberal postcondition for both possibilities and disjunct them. For the graph program $\mtt{if~}C\mtt{~then~}P\mtt{~else~}Q$, $P$ can be executed if $\SUCCESS(C)$ holds and $Q$ can be executed if $\FAIL(C)$ holds. Similarly for $\mtt{try~}C\mtt{~then~}P\mtt{~else~}Q$, $C;P$ can be executed if $\SUCCESS(C)$ holds and $Q$ can be executed if $\FAIL(C)$ holds.

\begin{proof}[of Proposition \ref{prop:slpP}]
Here, we show that the proposition holds by induction on loop-free programs.\\
Base case.
 \vspace{-\topsep}\begin{enumerate}
    \item If $S=\mathcal{R}=\{\}$,\\
        It is obvious that for all host graph $G$, $G\nRightarrow$ such that every condition is a liberal postcondition w.r.t. $c$ and $\mathcal{R}$, and $\mrm{false}$ is the strongest among all.
    \item If $S=\mathcal{R}=\{r_1,\ldots,r_n\}$ where $n>0$,\\
        \begin{tabular}{llcp{10cm}}
            (a) & $H\Sat\SLP(c,\mathcal{R})$ & $\myeq{\Leftrightarrow}{L\ref{lemma:slpP}}$ & $\exists G. G\Rightarrow_\mathcal{R}H \wedge G\Sat c$\\
             &&  $\Leftrightarrow$ & $\exists G. (G\Rightarrow_{r_1}H\vee\ldots\vee G\Rightarrow_{r_n}H)\wedge G\Sat c$\\
             &&  $\Leftrightarrow$ & $(\exists G. G\Rightarrow_{r_1}H\wedge G\Sat c)\vee\ldots\vee( \exists G. G\Rightarrow_{r_n}H\wedge G\Sat c)$\\
             && $\myeq{\Leftrightarrow}{L\ref{lemma:slpP}}$ & $H\Sat\text{SLP}(c,r_1)\vee\ldots\vee\text{SLP}(c,r_n)$
        \end{tabular}
\end{enumerate}
Inductive case. Assume the proposition holds for loop-free programs $C, P$, and $Q$.
 \vspace{-\topsep}\begin{enumerate}\small
    \item If $S=P\mtt{~or~}Q,$\\
    \begin{tabular}{llcp{11cm}}
            & $H\Sat\SLP(c,S)$ & $\myeq{\Leftrightarrow}{L\ref{lemma:slpP}}$ & $\exists G. H\in\Sem{P\mtt{~or~}Q}G \wedge G\Sat c$\\
             &&  $\Leftrightarrow$ & $\exists G. (H\in\Sem{P}G \vee H\in\Sem{Q}G)\wedge G\Sat c$\\
             &&  $\Leftrightarrow$ & $(\exists G. H\in\Sem{P}G \wedge G\Sat c) \vee (\exists G. H\in\Sem{Q}G\wedge G\Sat c)$\\
             &&  $\myeq{\Leftrightarrow}{L\ref{lemma:slpP}}$ & $G\Sat\SLP(c,P)\vee\SLP(c,Q)$
        \end{tabular}
    \item If $S=P;Q$,\\
        \begin{tabular}{llcp{10cm}}
            & $H\Sat\SLP(c,S)$ & $\myeq{\Leftrightarrow}{L\ref{lemma:slpP}}$ & $\exists G.~ H\in\Sem{P;Q}G \wedge G\Sat c$\\
             &&  $\Leftrightarrow$ & $\exists G,G'.~ G'\in\Sem{P}G \wedge H\in\Sem{Q}G'\wedge G\Sat c$\\
             &&  $\myeq{\Leftrightarrow}{L\ref{lemma:slpP}}$ & $\exists G'.~G'\Sat\SLP(c,P) \wedge H\in\Sem{Q}G'$\\
             &&  $\myeq{\Leftrightarrow}{L\ref{lemma:slpP}}$ & $H\Sat\SLP(\SLP(c,P),Q)$
        \end{tabular}
    \item If $S=\mtt{if\,}C\mtt{\,then\,}P\mtt{\,else\,}Q$,\\
    \begin{tabular}{llcp{10.7cm}}
            & \multicolumn{3}{l}{$H\Sat\SLP(c,S)$}\\
            && $\myeq{\Leftrightarrow}{L\ref{lemma:slpP}}$ & $\exists G.~ G\Sat c\wedge H\in\Sem{\mtt{if\,}C\mtt{\,then\,}P\mtt{\,else\,}Q}G$\\
             &&  $\Leftrightarrow$ & $\exists G.~ G\Sat c\wedge((G\Sat\SUCCESS(C)\wedge H\in\Sem{P}G) \vee (G\Sat\FAIL(C)\wedge H\in\Sem{Q}G))$\\
             && $\Leftrightarrow$ & $(\exists G.~G\Sat c\wedge\SUCCESS(C)\wedge H\in\Sem{P}G) \vee (\exists G.~ G\Sat c\wedge\FAIL(C)\wedge H\in\Sem{Q}G)$\\
             &&  $\myeq{\Leftrightarrow}{L\ref{lemma:slpP}}$ & $G\Sat\SLP(c\wedge\SUCCESS(C),P)\vee\SLP(c\wedge\FAIL(C),Q)$
        \end{tabular}
    \item If $S=\mtt{try\,}C\mtt{\,then\,}P\mtt{\,else\,}Q$,\\
    \begin{tabular}{llcp{10.7cm}}
            & \multicolumn{3}{l}{$H\Sat\SLP(c,S)$}\\
            && $\myeq{\Leftrightarrow}{L\ref{lemma:slpP}}$ & $\exists G.~ G\Sat c\wedge H\in\Sem{\mtt{try\,}C\mtt{\,then\,}P\mtt{\,else\,}Q}G$\\
             && $\Leftrightarrow$ & $(\exists G,G'.~G\Sat c\wedge G'\in\Sem{C}G\wedge H\in\Sem{P}G') \vee (\exists G.~ G\Sat c\wedge\FAIL(C)\wedge H\in\Sem{Q}G)$\\
             &&  $\myeq{\Leftrightarrow}{L\ref{lemma:slpP}}$ & $(\exists G'.~G'\Sat\SLP(c,C)\wedge H\in\Sem{P}G')\vee\SLP(c\wedge\FAIL(C),Q)$\\
             &&  $\myeq{\Leftrightarrow}{L\ref{lemma:slpP}}$ & $G\Sat\SLP(\SLP(c,C),P)\vee\SLP(c\wedge\FAIL(C),Q)$\qed
        \end{tabular}
\end{enumerate}\end{proof}

To prove the triple $\{c\}~P~\{d\}$ is partially correct for a graph program $P$, we only need to show that SLP$(c,P)$ implies $d$. However for graph programs $P$ containing a loop, obtaining the assertion SLP$(c,P)$ is not easy. Alternatively, we can create a proof tree (see Definition \ref{def:prooftree}) with proof rules to show that $\{c\}~P~\{d\}$ is partially correct. Before we define the proof rules for partial correctness, we define predicate Break which shows relation between a graph program and assertions.

\Def{Provability; proof tree\cite{Poskitt13}}{def:prooftree}{
Given an proof system $I$, a triple $\{c\}~P~\{d\}$ is provable in $I$, denoted by $\vdash_I\{c\}~P~\{d\}$, if one can construct a \emph{proof tree} from the axioms and inference rules of $I$ with that triple as the root.
If $\{c\}~P~\{d\}$ is an instance of an axiom $X$ then 
\[X~\frac{}{\{c\}~P~\{d\}}\]
is a proof tree, and $\vdash_I\{c\}~P~\{d\}$. If $\{c\}~P~\{d\}$ can be instantiated from the conclusion of an inference rule $Y$, and there are proof trees $T_1,\ldots ,T_n$ with conclusions that are instances of the $n$ premises of $Y$, then 
\[Y~\frac{T_1~~~\ldots~~~T_n}{\{c\}~P~\{d\}}\]
is a proof tree, and $\vdash_I\{c\}~P~\{d\}$.}

\Def{Predicate Break}{def:Break}{
Given a graph program $P$ and assertions $c$ and $d$. Break$(c,P,d)$ is the predicate defined by:
\[\small{\text{Break$(c,P,d)$ holds iff for all derivations $\langle P, G\rangle\rightarrow^*\langle\mtt{break},H\rangle, G\Sat c$ implies $H\Sat d$}}.\]
}

Intuitively, when Break$(c,P,d)$ holds, the execution of $\mtt{break}$ that yields to termination of $P!$ will result a graph satisfying $d$.

\Lemma{lemma:breakno}{
Given a graph program $P$ with invariant $c$. If $P$ does not contain the command $\mtt{break}$, then the following triple holds:
\[\{c\}~P!~\{c\wedge\text{FAIL}(P)\}\]
}

\begin{proof}
If $P$ does not contain the command $\mtt{break}$, then the derivation $\langle P, G\rangle\rightarrow^*\langle\mtt{break},H\rangle$ must not exist for any host graphs $G$ and $H$. Hence, Break$(c,P,d)$ is true for any $c$ and $d$. Hence, Break$(c,P,\mathrm{false})$ must be true. Since $c$ in an invariant, $\{c\}~P~\{c\}$ is true. If $\tuple{P!,G}\rightarrow^*H$ for some host graph $H$, from the semantics of graph programs, $\tuple{P!,G}\rightarrow\tuple{P!,H}\rightarrow^+\mtt{fail}$. $H$ must satisfy $c$ because $c$ is the invariant of $P$, and $H$ must satisfy FAIL$(P)$ because $fail\in\Sem{P}H$. Hence, the triple holds.\qed
\end{proof}

\Def{Semantic partial correctness proof rules}{def:semproofrule}{
The semantic partial correctness proof rules for core commands, denoted by \textsf{SEM}, is defined in \figurename~\ref{fig:semprules}, where $c, d,$ and $d'$ are any assertions, $r$ is any conditional rule schema, $\mathcal{R}$ is any set of rule schemata, and $C, P$, and $Q$ are any graph programs.}

\begin{figure}
    \centering
    \footnotesize
\def\arraystretch{3}\tabcolsep=2pt
~[ruleapp]$_{\text{slp}}~\displaystyle\frac{}{\{c\}~r~\{\text{SLP}(c,r)\}}$\\$~$\\
~[ruleapp]$_{\text{wlp}}~\displaystyle\frac{}{\{\text{WLP}(r,d)\}~r~\{d\}}$\\$~$\\
~[ruleset]$~\displaystyle\frac{\{c\}~r~\{d\}\text{ for each }r\in\mathcal{R}}{\{c\}~\mathcal{R}~\{d\}}$\\$~$\\
~[comp]$\displaystyle\frac{\{c\}~P~\{e\}~~~~\{e\}~P~\{d\}}{\{c\}~P;Q~\{d\}}$
\\$~$\\
~[cons]~$\displaystyle\frac{c\text{ implies }c'~~~\{c'\}~P~\{d'\}~~~d'\text{ implies }d}{\{c\}~P~\{d\}}$\\$~$\\
{~[if]$~\displaystyle\frac{\{c\wedge\text{SUCCESS}(C)\}~P~\{d\}~~~\{c\wedge\text{FAIL}(C)\}~Q~\{d\}}{\{c\}~\mtt{if~}C\mtt{~then~}P\mtt{~else~}Q~\{d\}}$}\\$~$\\
{~[try]$~\displaystyle\frac{\{c\wedge\text{SUCCESS}(C)\}~C;P~\{d\}~~~\{c\wedge\text{FAIL}(C)\}~Q~\{d\}}{\{c\}~\mtt{try~}C\mtt{~then~}P\mtt{~else~}Q~\{d\}}$}
\\$~$\\
~[alap]$~\displaystyle\frac{\{c\}~P~\{c\}~~~~~\text{Break}(c,P,d)}{\{c\}~P!~\{(c\wedge\text{FAIL}(P))\vee d\}}$
    \caption{Calculus \textsf{SEM} of semantic partial correctness proof rules}
    \label{fig:semprules}
\end{figure}

The inference rule [ruleset] tells us about the application of a set of rule schemata $\mathcal{R}$. The rule set $\mathcal{R}$ is applied to a graph by nondeterministically choose an applicable rule schema from the set and apply it to the input graph. Hence, to derive a triple about $\mathcal{R}$, we need to prove the same triple for each rule schema inside $\mathcal{R}$.

The inference rule [comp] is similar to [comp] in traditional programming. In executing $P;Q$, the graph program $Q$ is not executed until after the execution of $P$ has terminated. So to show a triple about $P;Q$, we need to prove a triple about each $P$ and $Q$ and show that they are connected to some midpoint such that the midpoint is satisfied after the execution of $P$ and before the execution of $Q$.

Like in conventional Hoare logic \cite{Apt19}, the rule [cons] is aimed to strengthen the precondition and weaken the postcondition, or to replace the condition to another condition that semantically equivalent but syntactically different. To show that $c'$ can be strengthened to $c$, we only need to show that $c$ implies $c'$, and to weaken $d'$ to $d$, we need to show that $d'$ implies $d$.

The assertions SUCCESS and FAIL are needed to prove a triple about if command. Recall that in the execution of $\mtt{if~}C\mtt{~then~}P\mtt{~else~}Q$, the program $C$ is first executed on a copy of $G$. If it terminates and yields a proper graph as a result, $P$ is executed on $G$. If $C$ terminates and results in a fail state, then $Q$ is executed on $G$. 

Similarly, for a triple about try command, we use the two assertions. But for $\mtt{try~}C\mtt{~then~}P\mtt{~else~}Q$, $C$ is not executed on a copy of $G$, but $G$ itself. When the execution of $C$ on $G$ terminates and yields a proper graph, $P$ is executed on the result graph. Hence, the difference with [if] is located in the first of the premises, where we use the sequential composition of $C$ and $P$.

As in traditional programming, we need an invariant to show a triple about loop $P!$. When we have proven the existence of an invariant for $P$, the invariant will hold after any number of successful executions of $P$. If $P!$ terminates, from the semantics of ``!" we know that the last execution of $P$ either yields a fail state (see [Loop$_2$] of \figurename~\ref{fig:infRule-core}), such that FAIL$(P)$ must hold, or executing the command $\mtt{break}$ (see [Loop$_3$] of \figurename~\ref{fig:infRule-core}). In the former case, it is clear that the invariant and FAIL$(P)$ must hold. Then in the latter case, we use Break$(c,P,d)$ which is defined in Definition \ref{def:Break}. The triple for loops is then captured by the rule [alap].

\subsection{Syntactic partial correctness calculus}

Section \ref{sec:sem} introduces us to the semantic of partial correctness calculus. Now that we already have first-order formulas for some properties in graph programming, in this section we define the construction of SLP, SUCCESS, and FAIL in first-order formulas. In addition, we also define the syntactic version of partial correctness proof rules where possible (it will turn out that this is not always can be done). First, we define the first-order formula App$(r)$ which should represent the first-order formula of SUCCESS$(r)$.

\Def{App$(r)$}{def:appr}{
Given a conditional rule schema $r:\lkr$. The formula App$(r)$ is defined as
\[\text{App}(r)=\text{Var}(\text{Spec}(L)\wedge\text{Dang}(r)\wedge{\Gamma}).\]
}

The definition of Var$(c)$, Spec$(L)$, and Dang$(r)$ for a condition $c$, rule graph $L$, and rule schema $r$, can be found in Definition \ref{def:var}, \ref{def:spec}, and \ref{def:dang} respectively.

\Lemma{lemma:appr}{
Given a conditional rule schema $r:\lkr,$ and a host graph $G$,
\[G\Sat\text{SUCCESS}(r) \text{ if and only if } G\Sat\text{App}(r).\]
}

\begin{proof}
(If).\\
$G\Sat$App$(r)$ implies $G\Sat$Var(Spec($L$)), such that from Lemma \ref{lemma:VarL}, we know that there exists injective morphism $g:L^\alpha\rightarrow G$ for some label assignment $\alpha_L$. Then from Lemma \ref{lemma:var}, $G\Sat$App$(r)$ implies $\RG\Sat$Dang$(r)$ and $\RG\Sat\Gamma^\alpha$. From Observation \ref{obsv:dang}, $\RG\Sat$Dang$(r)$ implies $g$ satisfies the dangling condition, and $\RG\Sat\Gamma^\alpha$ clearly implies $\Gamma^{\alpha,g}$ is satisfied by $G$. Hence, from the definition of conditional rule schema application, we know that $G\Rightarrow_{r,g}H$ for some host graph $H$ such that $G\Sat$SUCCESS$(r)$.\\
(Only if).\\
$G\Sat$SUCCESS$(r)$ implies $G\Rightarrow H$ for some host graph $H$, which implies the existence of injective morphism $g:L^\alpha\rightarrow G$ for some label assignment $\alpha_L$ such that $g$ satisfies the dangling condition and $G\Sat \Gamma^{\alpha,g}$. The existence of the injective morphism implies $G\Sat$Var(Spec$(L)$) from Lemma \ref{lemma:VarL}, the satisfaction of the dangling condition implies $\RG\Sat$Dang$(r)$, and the $G\Sat\Gamma^{\alpha,g}$ implies $\RG\Sat\Gamma$. Hence, $\RG\Sat$Spec$(L)$ since $L^\alpha\rightarrow\RG$ is inclusion (see Proposition \ref{prop:specL}). Hence, $\RG\Sat$Spec$(L)\wedge$Dang$(r)\wedge\Gamma$ so that from Lemma \ref{lemma:var}, $G\Sat$App$(r)$.\qed
\end{proof}

Defining a first-order formula for SUCCESS$(r)$ with a rule schema $r$ is easier than defining FO formula for SUCCESS$(P)$ with an arbitrary loop-free program $P$. This is because we need to express properties of the initial graph after checking the existence of derivations. To determine the properties of the initial graph, we introduce the condition Pre$(P,c)$ for a postcondition $c$ and a loop-free program $P$. 
Intuitively, Pre$(P,c)$ expresses the properties of the initial graph such that we can assert the existence of a host graph $H$ such that $H\Sat c$ and $H\in\Sem{P}G$. For an example, if there exists host graphs $G'$ and $H$ for a given host graph $G$ and rule schemata $r_1$ and $r_2$ such that $G\Rightarrow_{r_1}G'\Rightarrow_{r_2}H$ and $H\Sat\mrm{true}$ (which also means that $G\Sat\SUCCESS(P)$), then $G'$ should satisfy Pre$(r_2, \mrm{true})$ and $G$ should satisfy Pre($r_1$,Pre$(r_2,\mrm{true})$) such that Pre($r_1$,Pre$(r_2,\mrm{true})$) can be considered as SUCCESS$(r_1;r_2)$ in first-order formula. For more general cases, see Definition \ref{def:slpnoloop}. In the definition, $(r^\vee)^{-1}$ refers to the inverse of the generalised $r$ (see Definition \ref{def:rvee}).

\Def{Slp, Success, Fail, Pre of a loop-free program}{def:slpnoloop}{
Given a condition $c$ and a loop-free program $S$. The first-order formulas Slp$(c,S)$, Pre$(c,S)$, Success$(S)$, and Fail$(S)$ are defined inductively:
 \vspace{-\topsep}\begin{enumerate}
    \item If $S$ is a set of rule schemata $\mathcal{R}=\{r_1,\ldots,r_n\}$,
    \begin{enumerate}
        \item[(a)] Slp$(c,S)=\begin{cases}\text{Post}(c,r_1^\vee)\vee\ldots\vee\text{Post}(c,r_n^\vee)&\text{if~} n>0,\\
        \mrm{false}&\text{otherwise}\end{cases}$
        \item[(b)] Pre$(S,c)=\begin{cases}\text{Post}(c,(r_1^\vee)^{-1})\vee\ldots\vee\text{Post}(c,(r_n^\vee)^{-1})&\text{if~} n>0,\\
        \mrm{false}&\text{otherwise}\end{cases}$
        \item[(c)] Success$(S)=\begin{cases}\text{App}(r_1)\vee\ldots\vee\text{App}(r_n)&\text{if~} n>0,\\
        \mrm{false}&\text{otherwise}\end{cases}$
        \item[(d)] Fail$(S)=\begin{cases}\neg(\text{App}(r_1)\vee\ldots\vee\text{App}(r_n))&\text{if~} n>0,\\
        \mrm{false}&\text{otherwise}\end{cases}$
    \end{enumerate}
    \item For loop-free programs $C, P,$ and $Q$,
    \begin{enumerate}
    \item[(i)] If $S=P\mtt{~or~}Q$, 
        \begin{enumerate}
        \item[(a)] Slp$(c,S)$=Slp$(c,P)\vee$Slp$(c,Q)$
        \item[(b)] Pre$(S,c)$=Pre$(P,c)\vee$Pre$(Q,c)$
        \item[(c)] Success$(S)=$Success$(P)\vee$Success$(Q)$
        \item[(d)] Fail$(S)=$Fail$(P)\vee$Success$(Q)$
        \end{enumerate}
    \item[(ii)] If $S=P;Q$,
        \begin{enumerate}
        \item[(a)] Slp$(c,S)$=Slp(Slp$(c,P),Q$)
        \item[(b)] Pre$(S,c)$=Pre($P$,Pre$(Q,c)$)
        \item[(c)] Success$(S)=$Pre$(P,\text{Success}(Q))$
        \item[(d)] Fail$(S)=$Fail$(P)\vee$Pre($P$,Fail$(Q)$)
        \end{enumerate}
    \item[(iii)] If $S=\mtt{if\,}C\mtt{\,then\,}P\mtt{\,else\,}Q$,
        \begin{enumerate}
        \item[(a)] Slp$(c,S)$=Slp$(c\wedge\text{Success}(C),P)\,\vee\,$Slp$(c\wedge\text{Fail}(C),Q)$
        \item[(b)] Pre$(S,c)$=(Success$(C)\wedge$Pre($P,c))\,\vee\,$(Fail$(C)\wedge$Pre($Q,c))$
        \item[(c)] Success$(S)=($Success$(C)\wedge$Success($P)\,\vee\,$(Fail$(C)\wedge$Success($Q))$
        \item[(d)] Fail$(S)$=(Success$(C)\wedge$Fail$(P))\,\vee\,$(Fail$(C)\wedge$Fail$(Q))$
        \end{enumerate}
    \item[(iv)] If $S=\mtt{try\,}C\mtt{\,then\,}P\mtt{\,else\,}Q$,
        \begin{enumerate}
        \item[(a)] Slp$(c,S)$=Slp$(c\wedge\text{Success}(C),C;P)\,\vee\,$Slp$(c\wedge\text{Fail}(C),Q)$
        \item[(b)] Pre$(S,c)=$Pre($C,\text{Pre}(P,c))\,\vee\,$(Fail$(C)\wedge$Pre($Q,c))$
        \item[(c)] Success$(S)=$Pre(C,Success($P))\,\vee\,$(Fail$(C)\wedge$Success($Q))$
        \item[(d)] Fail$(S)=$Pre(Fail$(P),C))\,\vee\,$(Fail$(C)\wedge$Fail$(Q))$
    \end{enumerate}
    \end{enumerate}
\end{enumerate}
}

For a precondition $c$ and a loop-free program $S$, Slp$(c,S)$ is basically constructed based on Proposition $\ref{prop:slpP}$. For Pre$(S,c)$, since we want to know the property of the initial graph based on $c$ that is satisfied by the final graph and $S$, it works similar with constructing a weakest liberal precondition from a given postcondition and a program. Here we use \cite{Pennemann09} as a reference. However, in the reference the conditional part of $\mtt{if-then-else}$ command contains an assertion instead of a graph program such that if $C$ is an assertion, following their setting we will have Pre$(C,c)=C\implies$Pre$(P,c)\wedge\neg C\implies$Pre$(Q,c)$. The difference between assertions and graph programs as condition of a conditional program is, the satisfaction of the assertion on the initial graph implies that $Q$ can not be executed, while in our case, $G\Sat$Success$(C)$ does not always imply that $Q$ can not be executed. Hence, we change the equation to what we have in the definition above.

Success$(S)$ should express the existence of a proper graph as a final result, which means it should express the property of the initial graph based on $S$ and the final graph satisfying $\mrm{true}$. This is exactly what Pre$(\mrm{true},S)$ should express. Finally, Fail$(S)$ should express the property of the initial graph where failure is a result of the execution of $S$. Since we can yield failure anywhere is the subprogram of $S$, we need to disjunct all possibilities.

\Theo{Slp, Pre, Success, and Fail}{theo:FOL}{
For all condition $c$ and loop-free program $S$, the following holds:
 \vspace{-\topsep}\begin{enumerate}
    \item[(a)] Slp$(c,S)$ is a strongest liberal postcondition w.r.t. $c$ and $S$
    \item[(b)] For all host graph $G$, $G\Sat$Pre$(S,c)$ if and only if there exists host graph $H$ such that $H\in\Sem{S}G$ and $H\Sat c$
    \item[(c)] $G\Sat$Success$(S)$ if and only if $G\Sat$SUCCESS$(S)$
    \item[(d)] $G\Sat$Fail$(S)$ if and only if $G\Sat$FAIL$(S)$
\end{enumerate}
}

\begin{proof}
Here, we prove the theorem by induction on loop-free graph programs.\\
Base case.
 \vspace{-\topsep}\begin{enumerate}
    \item For $\mathcal{R}=\{\}$,\begin{enumerate}
        \item It is obvious that for all host graph $G$, $G\nRightarrow$ such that every condition is a liberal postcondition w.r.t. $c$ and $\mathcal{R}$, and $\mrm{false}$ is the strongest among all.
        \item Statement (b) is valid because nothing satisfies $\mrm{false}$.
        \item Both $G\Sat$Success$(\mathcal{R})$ and $G\Sat$SUCCESS$(\mathcal{R})$ always false such that $G\Sat$Success$(\mathcal{R})$ iff $G\Sat$SUCCESS$(\mathcal{R})$ holds.
        \item Similarly, this point holds because both $G\Sat$Fail$(\mathcal{R})$ and $G\Sat$FAIL$(\mathcal{R})$ always true.
    \end{enumerate}
    \item If $S=\mathcal{R}=\{r_1,\ldots,r_n\}$ where $n>0$,\\
        \begin{tabular}{llcp{10cm}}
            (a) & $H\Sat\SLP(c,\mathcal{R})$ & $\myeq{\Leftrightarrow}{P\ref{prop:slpP}}$ & $H\Sat\text{SLP}(c,r_1)\vee\ldots\vee\text{SLP}(c,r_n)$\\
             && $\myeq{\Leftrightarrow}{T\ref{theo:slp}}$ & $H\Sat\text{Post}(c,r_1^\vee)\vee\ldots\vee\text{Post}(c,r_n^\vee)$
        \end{tabular}
        \begin{tabular}{llcp{10cm}}
             (b) & \multicolumn{3}{l}{$\exists H. H\in\Sem{\mathcal{R}}G\wedge H\Sat c$}\\
             &~~~~~& $\Leftrightarrow$ & $\exists H. (G\Rightarrow_{r_1}H\vee\ldots\vee G\Rightarrow_{r_n}H)\wedge H\Sat c$ \\
             && $\Leftrightarrow$ & $(\exists H. G\Rightarrow_{r_1}H\wedge H\Sat c)\vee\ldots\vee (\exists H.G\Rightarrow_{r_n}H\wedge H\Sat c$)\\
             && $\myeq{\Leftrightarrow}{D\ref{def:generalisedrPO}}$ & $(\exists H. G\Rightarrow_{r_1^\vee}H\wedge H\Sat c)\vee\ldots\vee (\exists H.G\Rightarrow_{r_n^\vee}H\wedge H\Sat c$)\\
             && $\myeq{\Leftrightarrow}{L\ref{lemma:inverse}}$ & $(\exists H. H\Rightarrow_{(r_1^\vee)^{-1}}H\wedge H\Sat c)\vee\ldots\vee (\exists H.H\Rightarrow_{(r_n^\vee)^{-1}}G\wedge H\Sat c$)\\
             && $\myeq{\Leftrightarrow}{T\ref{theo:slp}}$ & $G\Sat\text{Post}(c,(r_1^\vee)^{-1})\vee\ldots\vee \text{Post}(c,(r_n^\vee)^{-1})$
        \end{tabular}
        \begin{tabular}{llcp{10cm}}
             (c) & $H\Sat\SUCCESS(\mathcal{R})$&$\Leftrightarrow$&$\exists H. H\in\Sem{\mathcal{R}}G$\\
             && $\Leftrightarrow$ & $\exists H. G\Rightarrow_{r_1}H\vee\ldots\vee G\Rightarrow_{r_n}H$ \\
             && $\Leftrightarrow$ & $(\exists H. G\Rightarrow_{r_1}H)\vee\ldots\vee (\exists H.G\Rightarrow_{r_n}H$)\\
             && $\myeq{\Leftrightarrow}{D\ref{def:assSE}}$ & $G\Sat\SUCCESS(r_1)\vee\ldots\vee \SUCCESS(r_n)$\\
             && $\myeq{\Leftrightarrow}{L\ref{lemma:appr}}$& $G\Sat\text{App}(r_1)\vee\ldots\vee \text{App}(r_n)$\\
        \end{tabular}
        \begin{tabular}{llcp{10cm}}
             (d) & $G\Sat\FAIL(\mathcal{R})$&$\Leftrightarrow$&$\text{fail}\in\Sem{\mathcal{R}}G$\\
             && $\Leftrightarrow$ & $(\neg\exists H. G\Rightarrow_{r_1}H)\wedge\ldots\wedge\ (\neg\exists H.G\Rightarrow_{r_n}H$)\\
             && $\myeq{\Leftrightarrow}{D\ref{def:assSE}}$ & $G\Sat\neg(\SUCCESS(r_1)\vee\ldots\wedge \SUCCESS(r_n))$\\
             && $\myeq{\Leftrightarrow}{L\ref{lemma:appr}}$& $G\Sat\neg(\text{App}(r_1)\vee\ldots\vee \text{App}(r_n))$
        \end{tabular}
\end{enumerate}
Inductive case. Assume (a), (b), (c), and (d) hold for loop-free programs $C, P$, and $Q$.
 \vspace{-\topsep}\begin{enumerate}\small
    \item If $S=P\mtt{~or~}Q,$\\
    \begin{tabular}{llcp{11cm}}
            (a) & $H\Sat\SLP(c,S)$ & $\myeq{\Leftrightarrow}{P\ref{prop:slpP}}$ & $G\Sat\SLP(c,P)\vee\SLP(c,Q)$\\
             && $\myeq{\Leftrightarrow}{Ind.}$ & $G\Sat\text{Slp}(c,P)\vee\text{Slp}(c,Q)$
        \end{tabular}
        \begin{tabular}{llcp{11cm}}
             (b) & {$\exists H. H\in\Sem{S}G\wedge H\Sat c$} & $\Leftrightarrow$ & $\exists H. (H\in\Sem{P}G \vee H\in\Sem{Q}G)\wedge H\Sat c$ \\
             && $\Leftrightarrow$ & $(\exists H. H\in\Sem{P}G\wedge H\Sat c)\vee(\exists H.H\in\Sem{Q}G\wedge H\Sat c$)\\
             && $\myeq{\Leftrightarrow}{Ind.}$ & $G\Sat\text{Pre}(P,c)\vee\text{Pre}(Q,c)$
        \end{tabular}
        \begin{tabular}{llcp{10cm}}
             (c) & $G\Sat\SUCCESS(S)$
             & $\Leftrightarrow$ & $\exists H. H\in\Sem{P\mtt{~or~}Q}G$\\
             && $\Leftrightarrow$ & $\exists H. H\in\Sem{P}G \vee H\in\Sem{Q}G$ \\
             && $\myeq{\Leftrightarrow}{D\ref{def:assSE}}$ & $G\Sat\SUCCESS(P)\vee\SUCCESS(Q)$\\
             && $\myeq{\Leftrightarrow}{Ind.}$ & $G\Sat\text{Success}(P)\vee\text{Success}(Q)$
        \end{tabular}
        \begin{tabular}{llcp{10cm}}
             (d) & $G\Sat\FAIL(S)$
             & $\Leftrightarrow$ & $\text{fail}\in\Sem{P\mtt{~or~}Q}G$\\
             && $\Leftrightarrow$ & $\text{fail}\in\Sem{P}G \vee \text{fail}\in\Sem{Q}G$ \\
             && $\myeq{\Leftrightarrow}{D\ref{def:assFE}}$ & $G\Sat\FAIL(P)\vee\FAIL(Q)$\\
             && $\myeq{\Leftrightarrow}{Ind.}$ & $G\Sat\text{Fail}(P)\vee\text{Fail}(Q)$
        \end{tabular}
    \item If $S=P;Q$,\\
        \begin{tabular}{llcp{10cm}}
            (a) & $H\Sat\SLP(c,S)$ & $\myeq{\Leftrightarrow}{P\ref{prop:slpP}}$ & $H\Sat\SLP(\SLP(c,P),Q)$\\
             && $\myeq{\Leftrightarrow}{Ind.}$ & $H\Sat\text{Slp}(\text{Slp}(c,P),Q)$
        \end{tabular}
        \begin{tabular}{llcp{10cm}}
             (b) & {$\exists H. H\in\Sem{S}G\wedge H\Sat c$}& $\Leftrightarrow$ & $\exists H,G'.~G'\in\Sem{P}G \wedge H\in\Sem{Q}G' \wedge H\Sat c$ \\
             && $\myeq{\Leftrightarrow}{Ind.}$ & $\exists G'.~ G'\Sat\text{Pre}(Q,c)\wedge G'\in\Sem{P}G$\\
             && $\myeq{\Leftrightarrow}{Ind.}$ & $G\Sat\text{Pre}(P,\text{Pre}(Q,c))$
        \end{tabular}
        \begin{tabular}{llcp{10cm}}
             (c) & $G\Sat\SUCCESS(S)$ & $\Leftrightarrow$ & $\exists H. H\in\Sem{P;Q}G$\\
             && $\Leftrightarrow$ & $\exists H,G'.~G'\in\Sem{P}G \wedge H\in\Sem{Q}G'$ \\
             && $\myeq{\Leftrightarrow}{D\ref{def:assSE}}$ & $\exists G'.~ G'\Sat\text{SUCCESS}(Q)\wedge G'\in\Sem{P}G$\\
             && $\myeq{\Leftrightarrow}{Ind.}$ & $\exists G'.~ G'\Sat\text{Success}(Q)\wedge G'\in\Sem{P}G$\\
             && $\myeq{\Leftrightarrow}{Ind.}$ & $G\Sat\text{Pre}(P,\text{Success}(Q))$
        \end{tabular}
        \begin{tabular}{llcp{10cm}}
             (d) & $G\Sat\FAIL(S)$ & $\Leftrightarrow$ & $\text{fail}\in\Sem{P;Q}G$\\
             && $\Leftrightarrow$ & $\text{fail}\in\Sem{P}G \vee \exists H.~H\in\Sem{P}G \wedge \text{fail}\in\Sem{Q}H$ \\
             && $\myeq{\Leftrightarrow}{D\ref{def:assFE}}$ & $G\Sat\FAIL(P) \vee \exists H.~H\in\Sem{P}G \wedge H\Sat\FAIL(Q)$ \\
             && $\myeq{\Leftrightarrow}{Ind.}$ & $G\Sat\text{Fail}(P) \vee \text{Pre}(P,\text{Fail}(Q))$ 
        \end{tabular}
    \item If $S=\mtt{if\,}C\mtt{\,then\,}P\mtt{\,else\,}Q$,\\
    \begin{tabular}{llcp{10.7cm}}
            (a) & \multicolumn{3}{l}{$H\Sat\SLP(c,S)$}\\
            && $\myeq{\Leftrightarrow}{P\ref{prop:slpP}}$ &  $G\Sat\SLP(c\wedge\SUCCESS(C),P)\vee\SLP(c\wedge\FAIL(C),Q)$\\
             && $\myeq{\Leftrightarrow}{Ind.}$ & $G\Sat\text{Slp}(c\wedge\text{Success}(C),P)\vee\text{SLP}(c\wedge\text{Fail}(C),Q)$
        \end{tabular}
        \begin{tabular}{llcp{10.7cm}}
             (b) & \multicolumn{3}{l}{$\exists H. H\in\Sem{S}G\wedge H\Sat c$}\\
             && $\Leftrightarrow$ & $\exists H.~((G\Sat\SUCCESS(C)\wedge H\in\Sem{P}G) \vee (G\Sat\FAIL(C)\wedge H\in\Sem{Q}G))\wedge H\Sat c$ \\
             && $\Leftrightarrow$ &$(\exists H.~G\Sat\SUCCESS(C)\wedge H\in\Sem{P}G\wedge H\Sat c)$\\
             &&& $\vee (\exists H.~G\Sat\FAIL(C)\wedge H\in\Sem{Q}G))\wedge H\Sat c)$\\
             && $\myeq{\Leftrightarrow}{Ind.}$ & $G\Sat(\text{Success}(C)\wedge\text{Pre}(P,c))\vee(\text{Fail}(C)\wedge\text{Pre}(Q,c))$
        \end{tabular}
        \begin{tabular}{llcp{10.7cm}}
             (c) & \multicolumn{3}{l}{$G\Sat\SUCCESS(S)$}\\
             && $\Leftrightarrow$ & 
             $\exists H. H\in\Sem{S}G$\\
             && $\Leftrightarrow$ & $\exists H.~(G\Sat\SUCCESS(C)\wedge H\in\Sem{P}G) \vee (G\Sat\FAIL(C)\wedge H\in\Sem{Q}G)$ \\
             && $\Leftrightarrow$ &$(\exists H.~G\Sat\SUCCESS(C)\wedge H\in\Sem{P}G)\vee (\exists H.~G\Sat\FAIL(C)\wedge H\in\Sem{Q}G)))$\\
             && $\myeq{\Leftrightarrow}{Ind.}$ & $G\Sat(\text{Success}(C)\wedge\text{Success}(P))\vee(\text{Fail}(C)\wedge\text{Success}(Q))$
        \end{tabular}
        \begin{tabular}{llcp{10.7cm}}
             (d) & \multicolumn{3}{l}{$G\Sat\FAIL(S)$}\\
             && $\Leftrightarrow$ & 
             $\text{fail}\in\Sem{S}G$\\
             && $\Leftrightarrow$ & $(G\Sat\SUCCESS(C)\wedge \text{fail}\in\Sem{P}G) \vee (G\Sat\FAIL(C)\wedge \text{fail}\in\Sem{Q}G)$ \\
             && $\myeq{\Leftrightarrow}{Ind.}$ & $G\Sat(\text{Success}(C)\wedge\text{Fail}(P))\vee(\text{Fail}(C)\wedge\text{Fail}(Q))$
        \end{tabular}
    \item If $S=\mtt{try\,}C\mtt{\,then\,}P\mtt{\,else\,}Q$,\\
    \begin{tabular}{llcp{10.7cm}}
            (a) & \multicolumn{3}{l}{$H\Sat\SLP(c,S)$}\\
            && $\myeq{\Leftrightarrow}{P\ref{prop:slpP}}$ &  $G\Sat\SLP(\SLP(c,C),P)\vee\SLP(c\wedge\FAIL(C),Q)$\\
             && $\myeq{\Leftrightarrow}{Ind.}$ & $G\Sat\text{Slp}(\text{Slp}(c,C),P)\vee\text{Slp}(c\wedge\text{Fail}(C),Q)$
        \end{tabular}
        \begin{tabular}{llcp{10.7cm}}
             (b) & \multicolumn{3}{l}{$\exists H. H\in\Sem{S}G\wedge H\Sat c$}\\
             && $\Leftrightarrow$ & $(\exists H,G'.~H\Sat c\wedge G'\in\Sem{C}G\wedge H\in\Sem{P}G') \vee (\exists H.~ H\Sat c\wedge\FAIL(C)\wedge H\in\Sem{Q}G)$\\
             && $\myeq{\Leftrightarrow}{Ind.}$ &$(\exists G'.~G'\Sat\text{Pre}(P,c)\wedge G'\in\Sem{C}G) \vee (\exists H.~G\Sat\text{Fail}(C)\wedge H\in\Sem{Q}G))\wedge H\Sat c)$\\
             && $\myeq{\Leftrightarrow}{Ind.}$ & $G\Sat\text{Pre}(C,\text{Pre}(P,c))\vee(\text{Fail}(C)\wedge\text{Pre}(Q,c))$
        \end{tabular}
        \begin{tabular}{llcp{10.7cm}}
             (c) & \multicolumn{3}{l}{$G\Sat\SUCCESS(S)$}\\
             && $\Leftrightarrow$ & 
             $\exists H. H\in\Sem{S}G$\\
             && $\Leftrightarrow$ & $(\exists H,G'.~G'\in\Sem{C}G\wedge H\in\Sem{P}G') \vee (\exists H.~ \FAIL(C)\wedge H\in\Sem{Q}G)$\\
             && $\myeq{\Leftrightarrow}{D\ref{def:assSE}}$ &$(\exists G'.~G'\in\Sem{C}G\wedge G'\Sat\SUCCESS(P)) \vee (\exists H.~ \FAIL(C)\wedge H\in\Sem{Q}G)$\\
             && $\myeq{\Leftrightarrow}{Ind.}$ & $G\Sat\text{Pre}(C,\text{Success}(P))\vee(\text{Fail}(C)\wedge\text{Success}(Q))$
        \end{tabular}
        \begin{tabular}{llcp{10.7cm}}
             (d) & \multicolumn{3}{l}{$G\Sat\FAIL(S)$}\\
             && $\Leftrightarrow$ & 
             $\text{fail}\in\Sem{S}G$\\
             && $\Leftrightarrow$ & $(\exists G'.~G'\in\Sem{C}G\wedge\text{fail}\in\Sem{P}G)\vee (G\Sat\FAIL(C)\wedge \text{fail}\in\Sem{Q}G)$ \\
             && $\myeq{\Leftrightarrow}{D\ref{def:assSE},\ref{def:assFE}}$ & $G\Sat(\SUCCESS(C)\wedge\FAIL(P))\vee(\FAIL(C)\wedge \FAIL(Q))$ \\
             && $\myeq{\Leftrightarrow}{Ind.}$ & $G\Sat(\text{Success}(C)\wedge\text{Fail}(P))\vee(\text{Fail}(C)\wedge\text{Fail}(Q))$\qed
        \end{tabular}
\end{enumerate}\end{proof}

For any loop-free program $P$, we now can find the first order formula of SLP, SUCCESS, and FAIL. However, constructing SLP and SUCCESS of a loop is a challenging task because a loop may diverge. However, constructing a FO formula for FAIL of a graph program with loops is not as challenging if we only consider some forms of graph programs. In \cite{Bak15a}, Bak introduced a class of commands that cannot fail. Hence, we can always conclude that Fail$(P)=\mrm{false}$ if $P$ is a command that cannot fail. Here, we introduce the class of non-failing commands.

\Def{Non-failing commands}{def:nofail}{
The class of \emph{non-failing commands} is inductively defined as follows:

\noindent Base case:
 \vspace{-\topsep}\begin{enumerate}
    \item $\mtt{break}$ and $\mtt{skip}$ are non-failing commands
    \item Every call of a rule schema with the empty graph as its left-hand graph is a non-failing command
    \item Every rule set call $\{r_1,\ldots,r_n\}$ for $n\geq 1$ where each $r_i$ has the empty graph as its left-hand graph, is a non-failing command
    \item Every command P! is a non-failing command
\end{enumerate}
\noindent Inductive case:
 \vspace{-\topsep}\begin{enumerate}
    \item $P;Q$ is a non-failing command if $P$ and $Q$ are non-failing commands
    \item $\mtt{if\,}C\mtt{\,then\,}P\mtt{\,else\,}Q$ is a non-failing command if $P$ and $Q$ are non-failing commands
    \item $\mtt{try\,}C\mtt{\,then\,}P\mtt{\,else\,}Q$ is a non-failing command if $P$ and $Q$ are non-failing commands
\end{enumerate}
}

Recall the inference rule [alap] of \textsf{SEM}. To obtain a triple $\{c\}~P!~\{d\}$ for some precondition $c$, postcondition $d$, and a graph program $P!$, we need to find Fail$(P)$. We now can construct Fail($P$) if $P$ is a loo-free program as in Definition \ref{def:slpnoloop}, or if $P$ is a non-failing command. 

Now, let us consider $P$ in the form $C;Q$. For any host graph $G$, $\text{fail}\in\Sem{C;Q}G$ iff $\text{fail}\in\Sem{C}G$ or $H\in\Sem{C}G\wedge\text{fail}\in\Sem{Q}H$ for some host graph $H$, which means $G\Sat \FAIL(C)\vee(\SUCCESS(C)\wedge\FAIL(Q))$. We can construct both Fail$(C)$ and Success$(C)$ if $C$ is a loop-free program, and we can construct Fail$(Q)$ if $Q$ is a loop-free program or a non-failing command. Here, we introduce the class of \textit{iteration commands} which is the class of commands where we can obtain Fail of the commands.

\Def{Iteration commands}{def:iteration}{
The class of iteration commands is inductively defined as follows:
 \vspace{-\topsep}\begin{enumerate}
    \item Every loop-free program is an iteration command
    \item Every non-failing command is an iteration command
    \item A command of the form $C;P$ is an iteration command if $C$ is a loop-free program and $P$ is an iteration command 
\end{enumerate}
}

If $S$ is a loop-free program, we can construct Fail$(S)$ as defined in Definition \ref{def:slpnoloop}. Meanwhile, if $S$ is a non-failing command, there is no graph $G$ such that fail$\in\Sem{S}G$ such that we can conclude that Fail$(S)\equiv\mrm{false}$. Finally, if $S$ is in the form of $C;P$ for a loop-free program $C$ and a non-failing program $P$, fail$\in\Sem{S}G$ for a graph $G$ only if fail$\in\Sem{C}G$ (because $P$ cannot fail), so that Fail$(S)\equiv$fail$(C)$.

\begin{definition}[Fail of iteration commands]\label{def:Failit}\normalfont
Let Fail$_{\text{lf}}(C)$ denotes the formula Fail$(C)$ for a loop-free program $C$ as defined in Definition \ref{def:slpnoloop}. For any iteration command $S$,\\
\footnotesize{Fail$(S)=\begin{cases}
\mrm{false} & \text{if $S$ is a non-failing command}\\
\text{Fail}_{\text{lf}}(S) & \text{if $S$ is a loop-free program}\\
\text{Fail}(C) & {\text{if $S=C;P$ for a loop-free program $C$, a non-failing program $P$}}\\
\end{cases}$\qed}
\end{definition}

\begin{theorem}\normalfont\label{theo:nofail}
Given an iteration command $S$. Then,
\[\text{$G\Sat\text{Fail}(S)$ if and only if $G\Sat\text{FAIL}(S)$}.\]
\end{theorem}

\begin{proof}
Here, we prove the theorem case by case.
 \vspace{-\topsep}\begin{enumerate}
    \item It is obvious that if $S$ is a non-failing command, then for any host graph $G$, $\text{fail}\notin\Sem{S}G$. Hence, there is no graph satisfying FAIL$(S)$ such that $G\Sat\mrm{false}$ iff $G\Sat\FAIL(S)$ holds.
    \item If $S$ is a loop-free program, $G\Sat\text{Fail}(S)$ iff $G\Sat\FAIL(S)$ holds based on Theorem \ref{theo:FOL}.
    \item If $S$ is in the form $C;P$ for a loop-free program $C$ and non-failing command $P$, then\\
    \begin{tabular}{lcl}
    $G\Sat\FAIL(C;P)$ & iff & fail$\in\Sem{C;P}G$ \\
     & iff & fail$\in\Sem{C}G \vee \exists G'. G'\in\Sem{C}G\wedge$ fail$\in\Sem{P}G'$\\
     & iff & fail$\in\Sem{C}G$\\
     & iff & $G\Sat\FAIL(C)$\\
     & iff & $G\Sat\text{Fail}(C)$
\end{tabular}
\end{enumerate}
\end{proof}

Now let us consider the proof calculus \textsf{SEM}. There is the assertion $\SUCCESS(C)$ and $\FAIL(C)$ where $C$ is the condition of a branching statement, and FAIL$(S)$ for a loop body $S$. Since we are only able to construct Success$(C)$ for a loop-free program $C$ and FAIL$(S)$ for an iteration command $S$, we do not define the syntactic version for arbitrary graph programs. Hence, we require a loop-free program as the condition of every branching statement and an iteration command as every loop body. For the axiom [ruleapp]$_\text{wlp}$, we follow the construction in \cite{HP09} where a weakest liberal precondition can be constructed using the construction of Slp.

\Def{Control programs}{def:controlprogram}{
A \emph{control command} is a command where the condition of every branching command is loop-free and every loop body is an iteration command. Similarly, a graph program is a \emph{control program} if all its command are control commands.}

\begin{lemma}\label{lemma:wlpr}\normalfont
Given a conditional rule schema $r$ and a closed first-order formula $d$. Let Wlp$(r,d)=\neg$Slp$(\neg d,r^{-1})$. Then for all host graphs $G$,\\
\[G\Sat \text{Wlp}(r,d) \text{ if and only if } G\Sat\text{WLP}(r,d).\]
\end{lemma}

\begin{proof}$~$\\
\begin{tabular}{lclr}
     $G\Sat\text{Wlp}(r,d)$ & iff & $G\Sat\neg\text{Post}(\neg d,(r^\vee)^{-1})$ &\\
     & iff & $\neg(\exists H,g,g^*. H\Rightarrow_{(r^\vee)^{-1},g^*,g}G \wedge H\Sat \neg d)$ & (Lemma \ref{lemma:slp})\\
     & iff & $\neg(\exists H,g,g^*. G\Rightarrow_{(r^\vee),g,g^*}H \wedge H\Sat \neg d)$ & (Lemma \ref{lemma:inverse})\\
     & iff & $\neg(\exists H. G\Rightarrow_{r}H \wedge H\Sat \neg d)$ & (Def. \ref{def:generalisedrPO})\\
     & iff & $\forall H. G\Rightarrow_{r}H \text{ implies } H\Sat d)$ & (Def. implication)\\
     & iff & $G\Sat$WLP$(r,d)$ & (Lemma \ref{lemma:slpP})\qed
\end{tabular}
\end{proof}

Now we know the FO formula for WLP$(r,c)$, SLP$(c,r)$, also SUCCESS$(P)$ and FAIL$(P)$ for some form of $P$. Finally, we define a syntactic partial correctness proof for control programs.

\begin{figure}
    \centering
    \footnotesize
\def\arraystretch{3}\tabcolsep=2pt
~[ruleapp]$_{\text{slp}}~\displaystyle\frac{}{\{c\}~r~\{\text{Slp}(c,r)\}}$\\$~$\\
~[ruleapp]$_{\text{wlp}}~\displaystyle\frac{}{\{\neg \text{Slp}(\neg d,r^{-1})\}~r~\{d\}}$\\$~$\\
~[ruleset]$~\displaystyle\frac{\{c\}~r~\{d\}\text{ for each }r\in\mathcal{R}}{\{c\}~\mathcal{R}~\{d\}}$\\$~$\\
~[comp]$\displaystyle\frac{\{c\}~P~\{e\}~~~~\{e\}~P~\{d\}}{\{c\}~P;Q~\{d\}}$
\\$~$\\
~[cons]~$\displaystyle\frac{c\text{ implies }c'~~~\{c'\}~P~\{d'\}~~~d'\text{ implies }d}{\{c\}~P~\{d\}}$\\$~$\\
{~[if]$~\displaystyle\frac{\{c\wedge\text{Success}(C)\}~P~\{d\}~~~\{c\wedge\text{Fail}(C)\}~Q~\{d\}}{\{c\}~\mtt{if~}C\mtt{~then~}P\mtt{~else~}Q~\{d\}}$}\\$~$\\
{~[try]$~\displaystyle\frac{\{c\wedge\text{Success}(C)\}~C;P~\{d\}~~~\{c\wedge\text{Fail}(C)\}~Q~\{d\}}{\{c\}~\mtt{try~}C\mtt{~then~}P\mtt{~else~}Q~\{d\}}$}
\\$~$\\
~[alap]$~\displaystyle\frac{\{c\}~S~\{c\}~~~~~~~~~\text{Break}(c,S,d)}{\{c\}~S!~\{(c\wedge\text{Fail}(S))\vee d\}}$
    \caption{Calculus \textsf{SYN} of syntactic partial correctness proof rules}
    \label{fig:synprules}
\end{figure}

\Def{Syntactic partial correctness proof rules}{def:synproofrule}{
The syntactic partial correctness proof rules, denoted by \textsf{SYN}, is defined in \figurename~\ref{fig:synprules}, where $c, d,$ and $d'$ are any conditions, $r$ is any conditional rule schema, $\mathcal{R}$ is any set of rule schemata, $C$ is any loop-free program, $P$ and $Q$ are any control commands, and $S$ is any iteration command. Outside a loop, we treat the command $\mtt{break}$ as a $\mtt{skip}$.}

In the following section, we give an example of graph program verification using the calculus \textsf{SYN} we define above.

%% file: input/6_Completeness.tex
In this section, we show that our proof calculi are sound, in the sense that if some triple can be proven in a calculus, then the triple must be partially correct. In addition, we also show the relative completeness of the proof calculi. 

\subsection{Soundness}

To proof the soundness, we use structural induction on proof tree as defined in Definition \ref{def:indprooftree}.

\begin{definition}[Structural induction on proof trees]\normalfont\label{def:indprooftree}
Given a property \textit{Prop}. To prove that \textit{Prop} holds for all proof trees (that are created from some proof rules) by \textit{structural induction on proof tree} is done by:
 \vspace{-\topsep}\begin{enumerate}
    \item Show that \textit{Prop} holds for each axiom in the proof rules
    \item Assuming that \textit{Prop} holds for each premise $T$ of inference rules in the proof rules, show that \textit{Prop} holds for the conclusion of each inference rules in the proof rules.\qed
\end{enumerate}
\end{definition}

When we prove that a triple $\{c\}P\{d\}$ for assertions $c, d$ and a graph program $P$ is partially correct by showing that SLP$(c,P)$ implies $d$, it is obviously sound because of the definition of a strongest liberal postcondition itself. Then if $c$ and $d$ are first-order formulas and $P$ is a loop-free program, showing that Slp$(c,P)$ implies $d$ implies that $\{c\}P\{d\}$ is partially correct from Theorem \ref{theo:FOL}. Then we also need to prove the soundness of proof calculus as summarised in \figurename~\ref{fig:semprules} and \figurename~\ref{fig:synprules}. 

\Theo{Soundness of \textsf{SEM}}{theo:semsound}{
Given a graph program $P$ and assertions $c,d$. Then,
\[\vdash_{\mathsf{SEM}}\{c\}~P~\{d\} \text{~implies~} \vDash \{c\}~P~\{d\}.\]
}

\begin{proof}
To prove the soundness, we show that the implication holds for each axiom and inference rule in the proof rule w.r.t. the semantics of graph programs by structural induction on proof trees.

 \vspace{-\topsep}\begin{enumerate}
    \item Base case : 
    \begin{enumerate}
        \item[(a)] [ruleapp]$_\text{slp}$. Suppose that $\vdash_{\mathsf{SEM}}\{c\}~r~\{d\}$ for a (conditional) rule schema $r$ where for all graphs $H$, $H\Sat d$ iff $H\Sat\SLP(c,r)$. Suppose that $G\Sat c$. From Definition \ref{def:slp}, $G\Rightarrow_rH$ implies $H\Sat d$ so that $\vDash\{c\}~P~\{d\}$. 
    \item[(b)] [ruleapp]$_\text{slp}$. Suppose that $\vdash_{\mathsf{SEM}}\{c\}~r~\{d\}$ for a (conditional) rule schema $r$ where for all graphs $G$, $G\Sat c$ iff $H\Sat\text{WLP}(r,c)$. Suppose that $G\Sat c$. From Definition \ref{def:wlpP}, $G\Rightarrow_rH$ implies $H\Sat d$ so that $\vDash\{c\}~P~\{d\}$. 
    \end{enumerate}
    \item Inductive case.\\
    Assume that \textit{Prop} holds for each premise of inference rules in Definition \ref{def:semproofrule} for a set of rule schemata $\mathcal{R}$, assertions $c,d,e,c',d',inv$, host graphs $G,G',H,H'$, and graph programs $C,P,Q$.
    \begin{enumerate}
        \item[(a)] [ruleset]. Suppose that $\vdash_{\mathsf{SEM}}\{c\}~\mathcal{R}~\{d\}$ and $G\Sat c$. Since we can have a proof tree where $\{c\}~\mathcal{R}~\{d\}$ is the root, then $\vdash_{\mathsf{SEM}}\{c\}~r~\{d\}$ for all $r\in\mathcal{R}$. From point 1, this means that $\vDash\{c\}~r~\{d\}$ for all $r\in\mathcal{R}$. From the semantics of graph programs, $H\in\Sem{R}G$ iff $H\in\Sem{r}G$ for some $r\in\mathcal{R}$. Since for any $r\in\mathcal{R}$, $H\in\Sem{r}G$ implies $H\Sat d,$ $H\in\Sem{\mathcal{R}}G$ implies $H\Sat d$ as well so that $\vDash\{c\}~\mathcal{R}~\{d\}$.
        
        \item[(b)] [comp]. Suppose that $\vdash_{\mathsf{SEM}}\{c\}~P;Q~\{d\}$ and $G\Sat c$. $\vdash_{\mathsf{SEM}}\{c\}~P;Q~\{d\}$, implies $\vdash_{\mathsf{SEM}}\{c\}~P~\{e\}$ and $\vdash_{\mathsf{SEM}}\{e\}~Q~\{d\}$. From the semantic of graph programs, $H\in\Sem{P;Q}G$ iff there exists $G'$ such that $G'\in\Sem{P}G$ and $H\in\Sem{Q}G'$. In addition to the assumption, $\vdash_{\mathsf{SEM}}\{c\}~P~\{e\}$ implies $G'\Sat e$, and $\vdash_{\mathsf{SEM}}\{e\}~Q~\{d\}$ implies $H\Sat d$ so that $\vDash \{c\}~P;Q~\{d\}$.
        
        \item[(c)] [cons]. Suppose that $\vdash_{\mathsf{SEM}}\{c\}~P~\{d\}$ and $G\Sat c$. From the inference rule, we know that $\vdash \{c'\}~P~\{d'\}$, $c$ implies $c'$ (so that $G\Sat c')$, and $d'$ implies $d$. From $\vdash \{c'\}~P~\{d'\}$, we get that for all host graphs $H$, $H\in\Sem{P}G$ implies $H\Sat d'$ so that $H\Sat d$. Hence, $\vDash \{c\}~P~\{d\}$.
    
        \item[(d)] [if]. Suppose that $\vdash_{\mathsf{SEM}}\{c\}~\mtt{if\,}C\mtt{\,then\,}P\mtt{\,else\,}Q~\{d\}$ and $G\Sat c$. From $\vdash_{\mathsf{SEM}}\{c\}~\mtt{if\,}C\mtt{\,then\,}P\mtt{\,else\,}Q~\{d\}$, we get $\vdash_{\mathsf{SEM}}\{c\wedge\SUCCESS(C)\}~P~\{d\}$ and $\vdash_{\mathsf{SEM}}\{c\wedge\FAIL(C)\}~Q~\{d\}$. From the former we know that for all host graphs $H$, if $G\Sat\SUCCESS(C)$ and $H\in\Sem{P}G$ then $H\Sat d$, while from the latter we know that for all host graphs $H$, if $G\Sat\FAIL(C)$ and $H\in\Sem{Q}G$ then $H\Sat d$. Recall that from the semantic of graph programs, $H\in\Sem{\mtt{if\,}C\mtt{\,then\,}P\mtt{\,else\,}Q}G$ iff $G\Sat\SUCCESS(C)\wedge H\in\Sem{P}G$ or $G\Sat\FAIL(C)\wedge H\in\Sem{Q}G$. Since both $G\Sat\SUCCESS(C)\wedge H\in\Sem{P}G$ and $G\Sat\FAIL(C)\wedge H\in\Sem{Q}G$ implies $H\Sat d$, $H\in\Sem{\mtt{if\,}C\mtt{\,then\,}P\mtt{\,else\,}Q}G$ implies $H\Sat d$ such that $\vDash \{c\}~\mtt{if\,}C\mtt{\,then\,}P\mtt{\,else\,}Q~\{d\}$.
    
        \item[(e)] [try]. Suppose that $\vdash_{\mathsf{SEM}}\{c\}~\mtt{try\,}C\mtt{\,then\,}P\mtt{\,else\,}Q~\{d\}$ and $G\Sat c$. $\vdash_{\mathsf{SEM}}\{c\}~\mtt{try\,}C\mtt{\,then\,}P\mtt{\,else\,}Q~\{d\}$ implies $\vdash_{\mathsf{SEM}}\{c\wedge\SUCCESS(C)\}~C;P~\{d\}$ and $\vdash_{\mathsf{SEM}}\{c\wedge\FAIL(C)\}~Q~\{d\}$. From the former we know that for all host graphs $H$, if $G\Sat\SUCCESS(C)$ and $H\in\Sem{C;P}G$ then $H\Sat d$, while from the latter we know that for all host graphs $H$, if $G\Sat\FAIL(C)$ and $H\in\Sem{Q}G$ then $H\Sat d$. Recall that from the semantic of graph programs, $H\in\Sem{\mtt{if\,}C\mtt{\,then\,}P\mtt{\,else\,}Q}G$ iff $G\Sat\SUCCESS(C)\wedge H\in\Sem{C;P}G$ or $G\Sat\FAIL(C)\wedge H\in\Sem{Q}G$. Since both $G\Sat\SUCCESS(C)\wedge H\in\Sem{C;P}G$ and $G\Sat\FAIL(C)\wedge H\in\Sem{Q}G$ implies $H\Sat d$, $H\in\Sem{\mtt{try\,}C\mtt{\,then\,}P\mtt{\,else\,}Q}G$ implies $H\Sat d$ such that $\vDash \{c\}~\mtt{try\,}C\mtt{\,then\,}P\mtt{\,else\,}Q~\{d\}$.
        
         \item[(f)] [alap]. Suppose that $\vdash_{\mathsf{SEM}}\{c\}~P!~\{d\}$ and $G\Sat c$. From $\vdash_{\mathsf{SEM}}\{c\}~P!~\{d\}$, we know that $\vdash_{\mathsf{SEM}}\{c\}~P~\{c\}$ and Break$(c,P,d)$ holds. From $\vdash_{\mathsf{SEM}}\{c\}~P~\{c\}$, we get that for all host graph $H$, $H\in\Sem{P}G$ implies $H\Sat c$, while from Definition \ref{def:Break} and the true value of Break$(c,P,d)$ we know that for all hsot graphs $H$, $G\Sat c$ and $\tuple{P,G}\rightarrow^*\tuple{\mtt{break;H}}$ implies $H\Sat d$.
         From the semantic of graph programs, $H\in\Sem{P!}G$ iff there exist derivation $\tuple{P,G}\rightarrow^*\tuple{\mtt{break;H}}$ or $\tuple{P!,G}\rightarrow^*\tuple{P!,H}$ and $\tuple{P!,H}\rightarrow^+\mtt{fail}$. The first case yields $H\Sat d$ because of Break$(c,P,d)$. Note that $\tuple{P!,G}\rightarrow^*\tuple{P!,H}$ is done by having (probably) multiple execution of $P$ on host graphs, so that from $\vdash_{\mathsf{SEM}}\{c\}~P~\{c\}$ we know that $H\Sat c$. Then since $\tuple{P!,H}\rightarrow^+\mtt{fail}$, $H\Sat\FAIL(P)$ so that $H\Sat c\wedge\FAIL(C)$. Hence, $H\in\Sem{P!}G$ implies $H\Sat d\vee(c\wedge\FAIL(P))$ so that $\vDash\{c\}~P!~\{d\}$.\qed
    \end{enumerate}
\end{enumerate}
\end{proof}

\Theo{Soundness of \textsf{SYN}}{theo:synsound}{
Let $P$ be a restricted graph program i.e. graph programs where for every subprogram in the form $\mtt{if\,}C\mtt{\,then\,}P\mtt{\,else\,}Q$, $\mtt{try\,}C\mtt{\,then\,}P\mtt{\,else\,}Q$, or $C!$, $C$ is a loop-free program. Let also $c$ and $d$ be first-order formulas. Then,
\[\vdash_{\mathsf{SYN}}\{c\}~P~\{d\} \text{~implies~} \vDash \{c\}~P~\{d\}.\]
}

\begin{proof}
The soundness of [ruleapp]$_\text{slp}$ follows from Theorem \ref{theo:slp} and Theorem \ref{theo:semsound}, while the soundness of [ruleapp]$_\text{wlp}$ follows from Theorem \ref{theo:slp} and Lemma \ref{lemma:wlpr}. The soundness of [ruleset], [comp], [cons], [if], and [try] follows from Theorem \ref{theo:semsound} and Theorem \ref{theo:FOL} about defining SUCCESS and FAIL in first-order formulas. Finally, the soundness of the inference rule [alap] follows from Theorem \ref{theo:semsound} and Theorem \ref{theo:nofail}.\qed
\end{proof}

\subsection{Relative completeness}

A proof calculus is complete when anytime we can denote a triple is valid according to partial correctness, then we can prove the correctness with the proof calculus. However, a completeness really depends on assertions we used because the ability to prove that $d$ can be implied by $c$ for some assertions $c$ and $d$ depends on language of the assertions. Hence, here we show the relative completeness instead of completeness, where we can separate the incompleteness due to the axioms and inference rules from any incompleteness in deducing valid assertions \cite{Cook78}. 

Before showing that \textsf{SEM} is relative complete, we first show that for any postcondition $d$ and graph program $P$, we can show that $\vdash_{\mathsf{SEM}}\{\text{WLP}(P,d)\}~P~\{d\}$.

\Lemma{lemma:provwlp}{
Given a graph program $S$ and a postcondition $d$. Then,\\
\[\vdash_{\mathsf{SEM}}\{\text{WLP}(S,d)\}~S~\{d\}\]
}

\begin{proof}
Here we prove the lemma by induction on graph programs.\\
\begin{longtable}{lcp{11cm}}
\multicolumn{3}{l}{Base case. If $S$ is a (conditional) rule schema $r$,}\\
    ~~ & \multicolumn{2}{p{11.7cm}}{$\vdash_{\mathsf{SEM}}\{\text{WLP}(S,d)\}~S~\{d\}$ automatically follows from the axiom [ruleapp]$_{\text{wlp}}$.}\\
    \multicolumn{3}{p{11.7cm}}{Inductive case.}\\
    &\multicolumn{2}{p{11.7cm}}{Assume that for graph programs $C, P,$ and $Q$, $\vdash_{\mathsf{SEM}}\{\text{WLP}(C,d)\}~C~\{d\}$, $\vdash_{\mathsf{SEM}}\{\text{WLP}(P,d)\}~P~\{d\}$, and $\vdash_{\mathsf{SEM}}\{\text{WLP}(Q,d)\}~Q~\{d\}$.}\\
    & (a) & If $S=\mathcal{R}$.\\
    && If $\mathcal{R}=\{\}$, then there is no premise to prove so that we can deduce $\vdash_{\mathsf{SEM}}\{\text{WLP}(\mathcal{R},d)\}~\mathcal{R}~\{d\}$ automatically. If $\mathcal{R}=\{r_1,\ldots,r_n\}$ for $n>0$, $\vdash_{\mathsf{SEM}}\{\text{WLP}(r_1,d)\}~r_1~\{d\},\ldots,\vdash_{\mathsf{SEM}}\{\text{WLP}(r_n,d)\}~r_n~\{d\}$ from [ruleapp]$_{\text{slp}}$. Let $e$ be the assertion $\text{WLP}(r_1,d)\wedge\ldots\wedge\text{WLP}(r_n,d)$, so that by [cons], $\vdash_{\mathsf{SEM}}\{e\}~r_1~\{d\},\ldots,\vdash_{\mathsf{SEM}}\{e\}~r_n~\{d\}$. By [ruleset] we then get that $\vdash_{\mathsf{SEM}}\{e\}~\mathcal{R}~\{d\}$. Then by [cons],
    $\vdash_{\mathsf{SEM}}\{\text{WLP}(\mathcal{R},d)\}~\mathcal{R}~\{d\}$ because\\
    &&\begin{tabular}{lcl}
        \multicolumn{3}{l}{$G\Sat\text{WLP}(\mathcal{R},d)$}\\
        $~~~$& $\myeq{\Leftrightarrow}{L\ref{lemma:wlpP}}$ & $\forall H. H\in\Sem{\mathcal{R}}G\Rightarrow H\Sat d$\\
        & $\Leftrightarrow$ & $\forall H. (H\in\Sem{r_1}G\vee\ldots\vee H\in\Sem{r_n}G)\Rightarrow H\Sat d$\\
        & $\Leftrightarrow$ & $\forall H. (H\in\Sem{r_1}G\Rightarrow H\Sat d)\wedge\ldots\wedge (H\in\Sem{r_n}G\Rightarrow H\Sat d)$\\
        & $\myeq{\Leftrightarrow}{L\ref{lemma:wlpP}}$ & $G\Sat \text{WLP}(r_1,d)\wedge\ldots\wedge\text{WLP}(r_n,d)$
    \end{tabular}\\
    & (b) & If $S=P;Q$,\\
    && From the assumption, $\vdash_{\mathsf{SEM}}\{\text{WLP}(P,\text{WLP}(Q,d))\}\,P\,\{\text{WLP}(Q,d)\}$ and $\vdash_{\mathsf{SEM}}\{\text{WLP}(Q,d)\}\,Q\,\{d\}$. Then by the inference rule [comp], we get that $\vdash_{\mathsf{SEM}}\{\text{WLP}(P,\text{WLP}(Q,d))\}~P;Q~\{d\}$. Finally by [cons], we have $\vdash_{\mathsf{SEM}}\{\text{WLP}(P;Q,d)\}~P;Q~\{d\}$ because\\
    &&\begin{tabular}{lcl}
        \multicolumn{3}{l}{$G\Sat\text{WLP}(P;Q,d)$}\\
        $~~~$& $\myeq{\Leftrightarrow}{L\ref{lemma:wlpP}}$ & $\forall H. H\in\Sem{P;Q}G\Rightarrow H\Sat d$\\
        & $\Leftrightarrow$ & $\forall H,G'. (G'\in\Sem{P}G \wedge H\in\Sem{Q}G'\Rightarrow H\Sat d$\\
        & $\Leftrightarrow$ & $\forall G'. (G'\in\Sem{P}G \Rightarrow (\forall H. H\in\Sem{Q}G'\Rightarrow H\Sat d)$\\
        & $\myeq{\Leftrightarrow}{L\ref{lemma:wlpP}}$ & $\forall G'. (G'\in\Sem{P}G \Rightarrow G'\Sat\text{WLP}(Q,d)$\\
        & $\myeq{\Leftrightarrow}{L\ref{lemma:wlpP}}$ & $G\Sat \text{WLP}(P,\text{WLP}(Q,d))$
    \end{tabular}\\
    & (c) & If $S=\mtt{if\,}C\mtt{\,then\,}P\mtt{\,else\,}Q$,\\
    && Both $\vdash_{\mathsf{SEM}}\{\text{WLP}(P,d)\}\,P\,\{d\}$ and $\vdash_{\mathsf{SEM}}\{\text{WLP}(Q,d)\}\,Q\,\{d\}$ follow from the assumption. By [cons], we have:\\
    &&$\vdash_{\mathsf{SEM}}\{\text{WLP}(P,d)\wedge(\FAIL(C)\Rightarrow\text{WLP}(Q,d))\}\,P\,\{d\}$ and \\
    &&$\vdash_{\mathsf{SEM}}\{\text{WLP}(Q,d)\wedge(\SUCCESS(C)\Rightarrow\text{WLP}(P,d))\}\,Q\,\{d\}$.\\
    && Let $e$ denotes $(\SUCCESS(C)\Rightarrow\text{WLP}(P,d))\wedge(\FAIL(C)\Rightarrow\text{WLP}(Q,d))$ so that by [cons], we have:\\
    &&$\vdash_{\mathsf{SEM}}\{e\wedge\SUCCESS(C)\}\,P\,\{d\}$ and $\vdash_{\mathsf{SEM}}\{e\wedge\FAIL(C)\}\,Q\,\{d\}$.\\
    &&By [if] we then get that $\vdash_{\mathsf{SEM}}\{e\}\,S\,\{d\}$, and finally by [cons] we have $\vdash_{\mathsf{SEM}}\{\text{WLP}(S,d)\}\,S\,\{d\}$ because\\
    &&\small{\begin{tabular}{lcl}
        \multicolumn{3}{l}{$G\Sat\text{WLP}(\mtt{if\,}C\mtt{\,then\,}P\mtt{\,else\,}Q,d)$}\\
        & $\myeq{\Leftrightarrow}{L\ref{lemma:wlpP}}$ & $\forall H.~ H\in\Sem{\mtt{if\,}C\mtt{\,then\,}P\mtt{\,else\,}Q}G\Rightarrow H\Sat d$\\
        & $\Leftrightarrow$ & $\forall H.~ ((G\Sat\SUCCESS(C)\wedge H\in\Sem{P}G)\vee(G\Sat\FAIL(C)\wedge H\in\Sem{Q}G))\Rightarrow H\Sat d$\\
        & $\Leftrightarrow$ & $(\forall H.~ (G\Sat\SUCCESS(C)\wedge H\in\Sem{P}G)\Rightarrow H\Sat d)$\\
        &&$\wedge(\forall H.~(G\Sat\FAIL(C)\wedge H\in\Sem{Q}G)\Rightarrow H\Sat d)$
        \\ 
        & $\Leftrightarrow$ & $G\Sat\SUCCESS(C)\Rightarrow(\forall H.~  H\in\Sem{P}G\Rightarrow H\Sat d)$\\
        &&$\wedge G\Sat\FAIL(C)\Rightarrow(\forall H.~H\in\Sem{Q}G\Rightarrow H\Sat d)$\\
        & $\myeq{\Leftrightarrow}{L\ref{lemma:wlpP}}$ & $G\Sat(\SUCCESS(C)\Rightarrow\text{WLP}(P,d)\wedge(\FAIL(C)\Rightarrow\text{WLP}(Q,d)$
    \end{tabular}}\\
    & (d) & If $S=\mtt{try\,}C\mtt{\,then\,}P\mtt{\,else\,}Q$,\\
    && Let $e$ denotes \begin{small}$\SUCCESS(C)\Rightarrow\text{WLP}(C;P,d)\wedge\FAIL(C)\Rightarrow\text{WLP}(Q,d)$\end{small}. Similar to point (c), from the assumption we have $\vdash_{\mathsf{SEM}}\{\text{WLP}(Q,d)\}\,Q\,\{d\}$, which imply $\vdash_{\mathsf{SEM}}\{e\wedge\FAIL(C)\}\,Q\,\{d\}$. Also from the assumption, we have both $\vdash_{\mathsf{SEM}}\{\text{WLP}(C,\text{WLP}(P,d))\}\,P\,\{\text{WLP}(P,d)\}$ and also $\vdash_{\mathsf{SEM}}\{\text{WLP}(P,d)\}\,P\,\{d\}$. By [comp] and [cons] as case $S=P;Q$,
    $\vdash_{\mathsf{SEM}}\{\text{WLP}(C;P,d)\}\,C;P\,\{d\}$. Then by [cons] as in $\mtt{if-then-try}$ case, $\vdash_{\mathsf{SEM}}\{e\wedge\SUCCESS(C)\}\,C;P\,\{d\}$ such that by the inference rule [try] we have $\vdash_{\mathsf{SEM}}\{e\}\,S\,\{d\}$. Finnaly by [cons], $\vdash_{\mathsf{SEM}}\{\text{WLP}(S,d)\}\,S\,\{d\}$ because\\
    &&\small{\begin{tabular}{lcl}
        \multicolumn{3}{l}{$G\Sat\text{WLP}(\mtt{try\,}C\mtt{\,then\,}P\mtt{\,else\,}Q,d)$}\\
        & $\myeq{\Leftrightarrow}{L\ref{lemma:wlpP}}$ & $\forall H.~ H\in\Sem{\mtt{try\,}C\mtt{\,then\,}P\mtt{\,else\,}Q}G\Rightarrow H\Sat d$\\
        & $\Leftrightarrow$ & $\forall H.~ ((G\Sat\SUCCESS(C)\wedge H\in\Sem{C;P}G)\vee(G\Sat\FAIL(C)\wedge H\in\Sem{Q}G))$\\
        &&$~~~~~~~~\Rightarrow H\Sat d$\\
        & $\Leftrightarrow$ & $(\forall H.~ (G\Sat\SUCCESS(C)\wedge H\in\Sem{C;P}G)\Rightarrow H\Sat d)$\\
        &&$\wedge(\forall H.~(G\Sat\FAIL(C)\wedge H\in\Sem{Q}G)\Rightarrow H\Sat d)$
        \\ 
        & $\Leftrightarrow$ & $G\Sat\SUCCESS(C)\Rightarrow(\forall H.~  H\in\Sem{C;P}G\Rightarrow H\Sat d)$\\
        &&$\wedge G\Sat\FAIL(C)\Rightarrow(\forall H.~H\in\Sem{Q}G\Rightarrow H\Sat d)$\\
        & $\myeq{\Leftrightarrow}{L\ref{lemma:wlpP}}$ & $G\Sat(\SUCCESS(C)\Rightarrow\text{WLP}(C;P,d)\wedge(\FAIL(C)\Rightarrow\text{WLP}(Q,d)$
    \end{tabular}}\\
     & (d) & If $S=P!$,\\
     &&From the assumption, $\vdash_{\mathsf{SEM}}\{\text{WLP}(P,\text{WLP}(P!,d))\}\,P\,\{\text{WLP}(P!,d)\}$. By [cons] as in $P;Q$ case, we get $\vdash_{\mathsf{SEM}}\{\text{WLP}(P;P!,d)\}\,P\,\{\text{WLP}(P!,d)\}$ such that by [cons] we know that $\vdash_{\mathsf{SEM}}\{\text{WLP}(P!,d)\}\,P\,\{\text{WLP}(P!,d)\}$. Note that from Theorem \ref{theo:semsound}, this implies $\vDash \{\text{WLP}(P!,d)\}\,P\,\{\text{WLP}(P!,d)\}$ such that for all host graphs $G_1,\ldots,G_n,$ and $H$ where $G_2\in\Sem{P}G_1,\ldots,G_n\in\Sem{P}G_{n-1},$ and $\tuple{P,G_n}\rightarrow^*\tuple{\mtt{break},H}$, $G\Sat\text{WLP}(P!.d)$ implies $G'\Sat \text{WLP}(P!.d)$ and $H\Sat d$. Hence, Break($\text{WLP}(P!.d),P,d$) holds. Then by the inference rule [alap], we have $\vdash_{\mathsf{SEM}}\{\text{WLP}(P!,d)\}\,P\,\{(\text{WLP}(P!,d)\wedge\FAIL(P))\vee d\}$ such that by [cons], $\vdash_{\mathsf{SEM}}\{\text{WLP}(P!,d)\}\,P\,\{d\}$ because\\
     &&\small{\begin{tabular}{lcl}
        $H\Sat\text{WLP}(P!,d)\wedge\FAIL(P)$
        & $\myeq{\Leftrightarrow}{L\ref{lemma:wlpP}}$ & $\text{fail}\in\Sem{P}H\wedge\forall H'.~ H'\in\Sem{P!}H\Rightarrow H'\Sat d$\\
        & $\Rightarrow$ & $H\in\Sem{P!}H\wedge\forall H'.~ H'\in\Sem{P!}H\Rightarrow H'\Sat d$\\
        & $\Rightarrow$ & $H\Sat d$.
    \end{tabular}}\qed
\end{longtable}
\end{proof}

\Theo{Relative completeness of \textsf{SEM}}{theo:syncomplete}{
Given a graph program $P$ and assertions $c,d$. Then,
\[\vDash\{c\}~P~\{d\} \text{~implies~} \vdash_{\mathsf{SEM}} \{c\}~P~\{d\}.\]
}

\begin{proof}
From Lemma \ref{lemma:provwlp}, we know that for all $\vdash_{\mathsf{SEM}} \{\text{WLP}(P,d)\}~P~\{d\}$ and from Theorem \ref{theo:semsound}, we get that $\text{WLP}(P,d)$ is a weakest liberal precondition over $P$ and $d$. Hence, if $\vDash\{c\}~P~\{d\}$, $c$ must imply $\text{WLP}(P,d)$ so that by [cons] we get that $\vdash_{\mathsf{SEM}} \{c\}~P~\{d\}$.\qed
\end{proof}

\begin{conjecture}
The proof calculus \textsf{SYN} is not relative complete.
\end{conjecture}

In Theorem \ref{theo:syncomplete}, we show the relative completeness of our semantic partial correctness calculus. This proof, however, assumes that the assertion language was expressive, i.e. able to express strongest liberal postcondition relative to arbitrary programs and preconditions. However, there are limitations in properties that can be expressed by first-order logic. We believe that FO logic can not express that a graph has an even number of nodes \cite{Libkin04}. Although we do not have proof of the incompleteness of \textsf{SYN}, we believe that the calculus is not relative complete due to the expressiveness of FO formulas.

There is strong evidence that this is impossible. For example, consider the triple $\{c\}~P~\{d\}$ with $c=\mrm{\A{V}x(m_V(x)=none\land\neg\E{E}y(s(y)=x\vee t(y)=x))}$ 
(all nodes are unmarked and isolated), $d=\mrm{\A{V}x(false)}$ (the graph is empty), and the following program:\\
\begin{samepage}
\begin{scriptsize}
$~~~~~~~~~~~~~~~~~~\mtt{Main = duplicate!; delete!}$\\[1ex]
%\noindent\adjustbox{max width=0.9\textwidth}{
$~~~~~~~~~~~~~~~~~~~~~~~~$\begin{tabular}{lll}
\begin{tikzpicture}[remember picture,
  inner/.style={circle,draw,minimum size=16pt},
  outer/.style={inner sep=2pt}, scale=0.7
  ]
  \node[outer] (AA) at (0,1) {$\tiny\mtt{duplicate(a:list)}$};
  \node[outer] (A) at (-1,0) {
   \begin{tikzpicture}[scale=0.7, transform shape]
		\node[inner, label=below:\tiny 1] (Aa) at (0,0) {$\tiny\mtt{a}$};
		\end{tikzpicture}};
  \node[outer] (B) at (0,0) {$\Rightarrow$};
  \node[outer] (C) at (1.5,0) {
   \begin{tikzpicture}[scale=0.7, transform shape]
		\node[inner, label=below:\tiny 1, fill=gray!50] (Aa) at (0,0) {$\tiny\mtt{a}$};	
		\node[inner, fill=gray!50] (Aa) at (1,0) {$\tiny\mtt{a}$};	
		\end{tikzpicture}};
	\end{tikzpicture}
&
$~~~~~~~~~~~~~~~~~~~~~~~~$&
\begin{tikzpicture}[remember picture,
  inner/.style={circle,draw,minimum size=16pt},
  outer/.style={inner sep=2pt}, scale=0.7
  ]
  \node[outer] (AA) at (-1,1) {$\tiny\mtt{delete(a:list)}$};
  \node[outer] (A) at (-1.5,0) {
   \begin{tikzpicture}[scale=0.7, transform shape]
		\node[inner, label=below:\tiny 1, fill=gray!50] (Aa) at (0,0) {$\tiny\mtt{a}$};	
		\node[inner, label=below:\tiny 2, fill=gray!50] (Aa) at (1,0) {$\tiny\mtt{a}$};	
		\end{tikzpicture}};
  \node[outer] (B) at (0,0) {$\Rightarrow$};
  \node[outer] (C) at (1,0) {$\emptyset$};
	\end{tikzpicture}

\end{tabular}\end{scriptsize}
\end{samepage}

It is obvious that $\vDash \{c\}~\mtt{duplicate!; delete!}~\{d\}$ holds: $\mtt{duplicate!}$ duplicates the number of nodes while marking the nodes grey, hence its result graph consists of an even number of isolated grey nodes. Then $\mtt{delete!}$ deletes pairs of grey nodes as long as possible, so the overall result is the empty graph. Note that ``consists of an even number of isolated grey nodes" is both the strongest postcondition with respect to $c$ and \ttt{duplicate!}, and the weakest precondition with respect to \ttt{delete!} and $d$.

Using \textsf{SYN} one can prove $\vdash \{c\}~\mtt{duplicate!}~\{e\}$ where $e$ expresses that all nodes are grey and isolated. However, we believe that our logic cannot express that a graph has an even number of nodes. This is because pure first-order logic (without built-in operations) cannot express this property \cite{Libkin04} and it is likely that this inexpressiveness carries over to our logic. As a consequence, one can only prove $\vdash \{e\}~\mtt{delete!}~\{f\}$ where $f$ expresses that the graph contains at most one node (because otherwise \ttt{delete} would be applicable). But we cannot use \textsf{SYN} to prove  $\vdash \{c\}~\mtt{duplicate!; delete!}~\{d\}$.

%% file: input/7_Examples.tex
In this section, we show an example of graph program verification with first-order logic. Here we consider the program $\mtt{2colouring}$ that can be seen in \figurename~\ref{fig:2col-prog}. Given a host graph where all nodes are unmarked and unrooted and all edges are unmarked. If the input graph is two-colourable, then all nodes in the resulting graph should marked with blue or red such that no two adjacent nodes have the same colour.

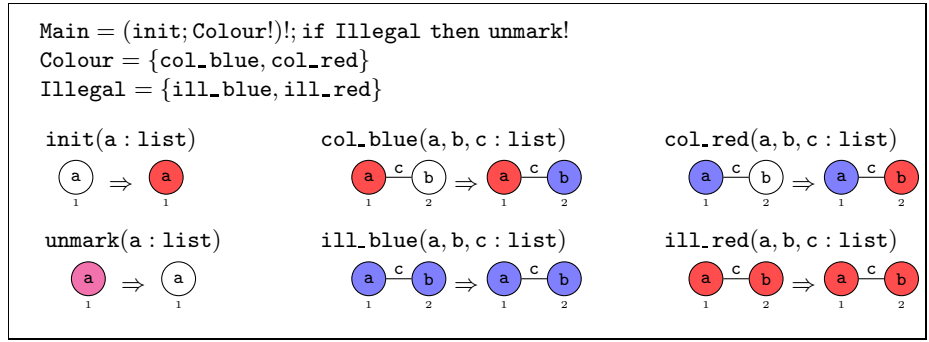
\begin{figure}
\begin{mdframed}
\begin{tabular}{l}
\small$\mtt{Main = (init; Colour!)!; if~Illegal~then~unmark!}$\\
$\mtt{Colour = \{col\_\,blue, col\_\,red\}}$\\
$\mtt{Illegal = \{ill\_\,blue, ill\_\,red\}}$\\
\end{tabular}\\[1.5ex]
%\noindent\adjustbox{max width=0.9\textwidth}{
\begin{tabular}{lllll}
\begin{tikzpicture}[remember picture,
  inner/.style={circle,draw,minimum size=16pt},
  outer/.style={inner sep=2pt}, scale=0.6
  ]
  \node[outer] (AA) at (0,1) {$\tiny\mtt{init(a:list)}$};
  \node[outer] (A) at (-1,0) {
   \begin{tikzpicture}[scale=0.8, transform shape]
		\node[inner, label=below:\tiny 1] (Aa) at (0,0) {$\mtt{a}$};
		\end{tikzpicture}};
  \node[outer] (B) at (0,0) {$\Rightarrow$};
  \node[outer] (C) at (1,0) {
   \begin{tikzpicture}[scale=0.8, transform shape]
		\node[inner, label=below:\tiny 1, fill=red!70] (Aa) at (0,0) {$\mtt{a}$};	
		\end{tikzpicture}};
	\end{tikzpicture}
	
	&$~~~~~~~~~$&\begin{tikzpicture}[remember picture,
  inner/.style={circle,draw,minimum size=16pt},
  outer/.style={inner sep=2pt}, scale=0.6
  ]
  \node[outer] (AA) at (0,1) {$\tiny\mtt{col\_\,blue(a,b,c:list)}$};
  \node[outer] (A) at (-1,0) {
   \begin{tikzpicture}[scale=0.8, transform shape]
		\node[inner, label=below:\tiny 1, fill=red!70] (Aa) at (0,0) {$\mtt{a}$};
		\node[inner, label=below:\tiny 2] (Ab) at (1,0) {$\mtt{b}$};
		\draw (Aa) to node[above] {$\mtt{c}$} (Ab);
		\end{tikzpicture}};
  \node[outer] (B) at (0.5,0) {$\Rightarrow$};
  \node[outer] (C) at (2,0) {
   \begin{tikzpicture}[scale=0.8, transform shape]
		\node[inner, label=below:\tiny 1, fill=red!70] (Aa) at (0,0) {$\mtt{a}$};
		\node[inner, label=below:\tiny 2, fill=blue!50] (Ab) at (1,0) {$\mtt{b}$};
		\draw (Aa) to node[above] {$\mtt{c}$} (Ab);	
		\end{tikzpicture}};
	\end{tikzpicture}
&$~~~~~~~$&
\begin{tikzpicture}[remember picture,
  inner/.style={circle,draw,minimum size=16pt},
  outer/.style={inner sep=2pt}, scale=0.6
  ]
  \node[outer] (AA) at (0,1) {$\tiny\mtt{col\_\,red(a,b,c:list)}$};
  \node[outer] (A) at (-1,0) {
   \begin{tikzpicture}[scale=0.8, transform shape]
		\node[inner, label=below:\tiny 1, fill=blue!50] (Aa) at (0,0) {$\mtt{a}$};
		\node[inner, label=below:\tiny 2] (Ab) at (1,0) {$\mtt{b}$};
		\draw (Aa) to node[above] {$\mtt{c}$} (Ab);
		\end{tikzpicture}};
  \node[outer] (B) at (0.5,0) {$\Rightarrow$};
  \node[outer] (C) at (2,0) {
   \begin{tikzpicture}[scale=0.8, transform shape]
		\node[inner, label=below:\tiny 1, fill=blue!50] (Aa) at (0,0) {$\mtt{a}$};
		\node[inner, label=below:\tiny 2, fill=red!70] (Ab) at (1,0) {$\mtt{b}$};
		\draw (Aa) to node[above] {$\mtt{c}$} (Ab);	
		\end{tikzpicture}};
	\end{tikzpicture}\\

\begin{tikzpicture}[remember picture,
  inner/.style={circle,draw,minimum size=16pt},
  outer/.style={inner sep=2pt}, scale=0.6
  ]
  \node[outer] (AA) at (0,1) {$\tiny\mtt{unmark(a:list)}$};
  \node[outer] (A) at (-1,0) {
   \begin{tikzpicture}[scale=0.8, transform shape]
		\node[inner, label=below:\tiny 1, fill=magenta!70] (Aa) at (0,0) {$\mtt{a}$};
		\end{tikzpicture}};
  \node[outer] (B) at (0,0) {$\Rightarrow$};
  \node[outer] (C) at (1,0) {
   \begin{tikzpicture}[scale=0.8, transform shape]
		\node[inner, label=below:\tiny 1] (Aa) at (0,0) {$\mtt{a}$};
		\end{tikzpicture}};
	\end{tikzpicture}
	&&
\begin{tikzpicture}[remember picture,
  inner/.style={circle,draw,minimum size=16pt},
  outer/.style={inner sep=2pt}, scale=0.6
  ]
  \node[outer] (AA) at (0,1) {$\tiny\mtt{ill\_\,blue(a,b,c:list)}$};
  \node[outer] (A) at (-1,0) {
   \begin{tikzpicture}[scale=0.8, transform shape]
		\node[inner, label=below:\tiny 1, fill=blue!50] (Aa) at (0,0) {$\mtt{a}$};
		\node[inner, label=below:\tiny 2, fill=blue!50] (Ab) at (1,0) {$\mtt{b}$};
		\draw (Aa) to node[above] {$\mtt{c}$} (Ab);
		\end{tikzpicture}};
  \node[outer] (B) at (0.5,0) {$\Rightarrow$};
  \node[outer] (C) at (2,0) {
   \begin{tikzpicture}[scale=0.8, transform shape]
		\node[inner, label=below:\tiny 1, fill=blue!50] (Aa) at (0,0) {$\mtt{a}$};
		\node[inner, label=below:\tiny 2, fill=blue!50] (Ab) at (1,0) {$\mtt{b}$};
		\draw (Aa) to node[above] {$\mtt{c}$} (Ab);	
		\end{tikzpicture}};
	\end{tikzpicture}
&&
\begin{tikzpicture}[remember picture,
  inner/.style={circle,draw,minimum size=16pt},
  outer/.style={inner sep=2pt}, scale=0.6
  ]
  \node[outer] (AA) at (0,1) {$\tiny\mtt{ill\_\,red(a,b,c:list)}$};
  \node[outer] (A) at (-1,0) {
   \begin{tikzpicture}[scale=0.8, transform shape]
		\node[inner, label=below:\tiny 1, fill=red!70] (Aa) at (0,0) {$\mtt{a}$};
		\node[inner, label=below:\tiny 2, fill=red!70] (Ab) at (1,0) {$\mtt{b}$};
		\draw (Aa) to node[above] {$\mtt{c}$} (Ab);
		\end{tikzpicture}};
  \node[outer] (B) at (0.5,0) {$\Rightarrow$};
  \node[outer] (C) at (2,0) {
   \begin{tikzpicture}[scale=0.8, transform shape]
		\node[inner, label=below:\tiny 1, fill=red!70] (Aa) at (0,0) {$\mtt{a}$};
		\node[inner, label=below:\tiny 2, fill=red!70] (Ab) at (1,0) {$\mtt{b}$};
		\draw (Aa) to node[above] {$\mtt{c}$} (Ab);	
		\end{tikzpicture}};
	\end{tikzpicture}
	
\end{tabular}
\end{mdframed}
\caption{Graph program $\mtt{2colouring}$ for computing a 2-colouring graph}
\label{fig:2col-prog}
\end{figure}

Then, let us consider the following pre- and postcondition:\\
\noindent\textbf{Precondition} ``every node and edge is unmarked and every node is unrooted"\\
\textbf{Postcondition} ``the precondition holds or every node is marked with blue or red, and no two adjacent nodes marked with the same colour"

Let $c$ and $c\vee d$ be the FO formulas expressing pre- and postcondition respectively. We define $c$ and $d$ as follows:\\
\begin{footnotesize}\begin{tabular}{rcl}
    $c$ & $\equiv$ & $\mrm{\A{V}x(\mV(x)= none\wedge\neg root(x))\wedge\A{E}x(\mE(x)=none)}$\\
    $d$ & $\equiv$ & $\mrm{\A{V}x((\mV(x)=red\vee\mV(x)=blue))\wedge\neg\E{E}x(s(x)\neq t(x)\wedge\mV(s(x))=\mV(t(x)))}$
\end{tabular}\end{footnotesize}

A proof tree for the partial correctness for $\mtt{2-colouring}$ with respect to $c$ and $c\vee d$ is provided in \figurename~\ref{fig:2col-proof}. The conditions in the tree are defined in Table \ref{tab:2col-cond}. 

\begin{figure}
    \centering\begin{scriptsize}
    \begin{mdframed}

\begin{prooftree}
\AxiomC{Subtree I}
\AxiomC{Subtree II}
\LeftLabel{[comp]}
\BinaryInfC{\{$f$\}$~\mtt{2colouring}~$\{$c\vee d$\}}
\LeftLabel{[cons]}
\UnaryInfC{\{$c$\}$~\mtt{2colouring}~$\{$c\vee d$\}}
\end{prooftree}
$~$\\$~$\\

where subtree I is:
\begin{prooftree}
\AxiomC{}
\LeftLabel{[ruleapp]$_\text{slp}$}
\UnaryInfC{\{$f$\}$\mtt{init}$\{Slp$(f,\mtt{init})$\}}
\LeftLabel{[cons]}
\UnaryInfC{\{$f$\}$\mtt{init}$\{$f$\}}

\AxiomC{subtree I.a}
\LeftLabel{[comp]}
\BinaryInfC{\{$f$\}$~\mtt{init;Colour!}~$\{$f$\}}
\LeftLabel{[alap]}
\UnaryInfC{\{$f$\}$~\mtt{(init; Colour!)!~}$\{$f\wedge\text{Fail}(\mtt{init;Colour!})$\}}
\LeftLabel{[cons]}
\UnaryInfC{\{$f$\}$~\mtt{(init;Colour!)!}~$\{$e$\}}

\end{prooftree}
with subtree I.a:
\begin{prooftree}
\AxiomC{}
\LeftLabel{[ruleapp]$_\text{slp}$}
\UnaryInfC{\{$f$\}$~\mtt{c\_\,blue}~$\{Slp$(f,\mtt{c\_\,blue})$\}}
\LeftLabel{[cons]}
\UnaryInfC{\{$f$\}$~\mtt{c\_\,blue}~$\{$f$\}}

\AxiomC{}
\LeftLabel{[ruleapp]$_\text{slp}$}
\UnaryInfC{\{$f$\}$~\mtt{c\_\,red}~$\{Slp$(f,\mtt{c\_\,red})$\}}
\LeftLabel{[cons]}
\UnaryInfC{\{$f$\}$~\mtt{c\_\,red}~$\{$f$\}}

\LeftLabel{[cons]}
\BinaryInfC{\{$f$\}$~\mtt{Colour}~$\{$f$\}}
\LeftLabel{[alap]}
\UnaryInfC{\{$f$\}$~\mtt{Colour!}~$\{$f\wedge\text{Fail}(\mtt{Colour})$\}}
\LeftLabel{[cons]}
\UnaryInfC{\{$f$\}$~\mtt{Colour!}~$\{$f$\}}
\end{prooftree}

$~$\\$~$\\
and subtree II is:
\begin{prooftree}
\AxiomC{}
\LeftLabel{[ruleapp]$_\text{slp}$}
\UnaryInfC{\{$f$\}$~\mtt{unmark}~$\{Slp$(f,\mtt{unmark})$\}}
\LeftLabel{[cons]}
\UnaryInfC{\{$f$\}$~\mtt{unmark}~$\{$f$\}}
\LeftLabel{[alap]}
\UnaryInfC{\{$f$\}$~\mtt{unmark!}~$\{$f\wedge\text{Fail}(\mtt{unmark})$\}}
\LeftLabel{[cons]}
\UnaryInfC{\{$e\wedge\text{Success}(\mtt{Illegal})$\}$~\mtt{unmark!}~$\{$c\vee d$\}}

\AxiomC{}
\LeftLabel{[ruleapp]$_\text{slp}$}
\UnaryInfC{\{$d$\}$~\mtt{skip}~$\{$d$\})}
\LeftLabel{[cons]}
\UnaryInfC{\{$e\wedge\text{Fail}(\mtt{Illegal})$\}$~\mtt{skip}~$\{$c\vee d$\}}

\LeftLabel{[if]}
\BinaryInfC{\{$e$\}$~\mtt{if~Illegal~then~umark!}~$\{$c\vee d$\}}
\end{prooftree}
\end{mdframed}\end{scriptsize}
    \caption{Proof tree for partial correctness of $\mtt{2colouring}$}
    \label{fig:2col-proof}
\end{figure}

\begin{table}[]
    \caption{Assertions inside proof tree of $\mtt{2-colouring}$}
    \label{tab:2col-cond}
    \centering\begin{scriptsize}
    \begin{tabular}{|p{11.5cm}|}
        \hline
        \multicolumn{1}{|c|}{\textbf{symbol and its first-order formulas}}\\\hline
        
        $c
        \equiv
        \mrm{\A{V}x(\mV(x)=none\wedge\neg root(x))\wedge\A{E}x(\mE(x)=none)}$\\\hline
        
        $d\equiv
        \mrm{\A{V}x((\mV(x)=red\vee\mV(x)=blue))\wedge\neg\E{E}x(s(x)\neq t(x)\wedge\mV(s(x))=\mV(t(x)))}$\\\hline
        
        $e
        \equiv
        \mrm{\A{V}x((\mV(x)=red\vee\mV(x)=blue)\wedge\neg root(x))\wedge\A{E}x(\mE(x)=none)}$\\\hline
        
        $f
        \equiv
        \mrm{\A{V}x((\mV(x)=red\vee\mV(x)=blue\vee\mV(x)=none))\wedge\neg root(x))\wedge\A{E}x(\mE(x)=none)}$\\\hline
       
        Slp$(f,\mtt{init})$\\
        $\equiv
        \mrm{\E{V}y(\A{V}x(x=y\vee((\mV(x)=red\vee\mV(x)=blue\vee\mV(x)=none)\wedge\neg root(x)))}$\\
        $\mrm{~~~~~~~~~\wedge \mV(y)=red\wedge\neg root(y))\wedge\A{E}x(\mE(x)=none)}$\\\hline
        
        Slp$(f,\mtt{c\_blue})=$Slp$(f,\mtt{c\_red})$\\
        $\equiv
        \mrm{\E{V}u,v(\A{V}x(x=u\vee x=v\vee((\mV(x)=red\vee\mV(x)=blue\vee\mV(x)=none)\wedge\neg root(x)))}$\\
       $\mrm{~~~~~~~~~~~~\wedge \mV(u)=red\wedge\mV(v)=blue\wedge\neg root(u)\wedge\neg root(v)}$\\
        $\mrm{~~~~~~~~~~~~\wedge\E{E}y((s(y)=u\wedge t(y)=v)\vee(t(y)=u\wedge s(y)=v)))\wedge\A{E}x(\mE(x)=none)}$\\\hline
        
        Slp$(f,\mtt{unmark})$\\
        $\equiv
        \mrm{\E{V}y(\A{V}x(x=y\vee((\mV(x)=red\vee\mV(x)=blue\vee\mV(x)=none)\wedge\neg root(x)))}$\\
        $\mrm{~~~~~~~~~\wedge\mV(y)=none\wedge\neg root(y))\wedge\A{E}x(\mE(x)=none)}$\\\hline
        
        Fail$(\mtt{Colour})$\\
        $\equiv
        \mrm{\neg\E{E}x((((\mV(s(x))=red\vee\mV(s(x))=blue)\wedge \mV(t(x))=none)}$\\
        $\mrm{~~~~~~~~~~~~\vee((\mV(t(x))=red\vee\mV(t(x))=blue)\wedge \mV(s(x))=none))}$\\
        $\mrm{~~~~~~~~~~~~\wedge \neg root(s(x))\wedge\neg root(t(x))})$\\\hline
        
        Fail$(\mtt{init;Colour!})
        \equiv
        \mrm{\neg\E{V}x(\mV(x)=none\wedge\neg root(x))}$\\\hline

        Fail$(\mtt{unmark})
        \equiv
        \mrm{\neg\E{V}x(\mV(x)\neq none\wedge\neg root(x))}$\\\hline

        Fail$(\mtt{Illegal})$\\
        $\equiv
        \mrm{\neg\E{E}x(s(x)\neq t(x)}$\\
        $\mrm{~~~~~~~~~~~~\wedge((\mV(s(x))=red\wedge\mV(t(x))=red)\vee(\mV(s(x))=blue\wedge\mV(t(x))=blue)))}$\\\hline

        Success$(\mtt{Illegal})$\\
        $\equiv
        \mrm{\E{E}x(s(x)\neq t(x)}$\\
        $\mrm{~~~~~~~~~~\wedge((\mV(s(x))=red\wedge\mV(t(x))=red)\vee(\mV(s(x))=blue\wedge\mV(t(x))=blue)))}$\\\hline

    \end{tabular}\end{scriptsize}
\end{table}

Note that there is no command $\mtt{break}$ in the program, so Break$(c,P,\mrm{false})$ holds for any precondition $c$ and sub-command $P$ of the program $\mtt{2colouring}$. For this reason and for simplicity, we omit premise Break$(c,P,\mrm{false})$ in the inference rule [alap] of the proof tree.

As we can see in the proof tree of \figurename~\ref{fig:2col-proof}, we apply some inference rule [cons] which means we need to give proof of implications applied to the rules. Some implications are obvious, e.g. $c$ implies $c\vee d$, so that for those obvious implications, we do not give any argument about them. Otherwise, we show that the implications hold:

 \vspace{-\topsep}\begin{enumerate}
    \item \textit{Proof of} $c$ \textit{implies} $f$.\\
    \begin{small}
    \begin{tabular}{lcl}
        $G\Sat c$ & $\Leftrightarrow$ & $G\Sat\mrm{\A{V}x(\mV(x)=none\wedge\neg root(x))\wedge\A{E}x(\mE(x)=none)}$  \\
        & $\Rightarrow$ &  $G\Sat\mrm{\A{V}x((\mV(x)=none\vee\mV(x)=blue\vee\mV(x)=red)\wedge\neg root(x))}$\\
        && $\mrm{~~~~~\wedge\A{E}x(\mE(x)=none)}$
    \end{tabular}\end{small}\\
    
    \item \textit{Proof of} Slp$(f,\mtt{init})$ \textit{implies} $f$.\\
    \begin{small}
    \begin{tabular}{lcl}
        \multicolumn{3}{l}{$G\Sat$Slp$(f,\mtt{init})$}\\
        & $\Leftrightarrow$ & $G\Sat\mrm{\E{V}y(\A{V}x(x=y\vee(\mV(x)=red\vee\mV(x)=blue\vee\mV(x)=none))}$\\
`       &&$\mrm{~~~~~~~~~~~~\wedge \mV(y)=red\wedge\neg root(y))}$  \\
        && $\mrm{~~~~~\wedge\A{E}x(\mE(x)=none)}$\\
        & $\Rightarrow$ & $G\Sat\mrm{\A{V}x((\mV(x)=red\wedge\neg root(x))}$\\
        &&$\mrm{~~~~~~~~~~~~~\vee(\mV(x)=red\vee\mV(x)=blue\vee\mV(x)=none))}$ \\
        && $\mrm{~~~~~\wedge\A{E}x(\mE(x)=none)}$\\
        & $\Rightarrow$ & 
        $G\Sat\mrm{\A{V}x((\mV(x)=none\vee\mV(x)=blue\vee\mV(x)=red)\wedge\neg root(x))}$\\
        && $\mrm{~~~~~\wedge\A{E}x(\mE(x)=none)}$
    \end{tabular}\end{small}\\
    
    \item \textit{Proof of} Slp$(f,\mtt{c\_blue})$ \textit{implies} $f$.\\
    \begin{small}
    \begin{tabular}{lcl}
        \multicolumn{3}{l}{$G\Sat$Slp$(f,\mtt{c\_blue})$}\\
        & $\Leftrightarrow$ & $G\Sat\mrm{\E{V}u,v(u\neq v\wedge\A{V}x(x=u\vee x=v}$\\
        &&$\mrm{~~~~~~~~~~~~~~~~~~~~~~\vee((\mV(x)=none\vee\mV(x)=blue\vee\mV(x)=red)\wedge\neg root(x)))}$\\
        &&$\mrm{~~~~~~~~~~~~~~~~\wedge \mV(u)=red\wedge\mV(v)=blue\wedge\neg root(u)\wedge\neg root(v)}$\\
        &&$\mrm{~~~~~~~~~~~~~~~~\wedge\E{E}y((s(y)=u\wedge t(y)=v)\vee(t(y)=u\wedge s(y)=v)))}$  \\
        && $\mrm{~~~~~\wedge\A{E}x(\mE(x)=none)}$\\
        & $\Rightarrow$ & $G\Sat\mrm{\A{V}x((\mV(x)=red\wedge\neg root(x))\vee (\mV(x)=blue\wedge\neg root(x))}$\\
        &&$\mrm{~~~~~~~~~~~~~\vee(\mV(x)=none\vee\mV(x)=blue\vee\mV(x)=red)\wedge\neg root(x))}$\\
        && $\mrm{~~~~~\wedge\A{E}x(\mE(x)=none)}$\\
        & $\Rightarrow$ & 
        $G\Sat\mrm{\A{V}x((\mV(x)=none\vee\mV(x)=blue\vee\mV(x)=red)\wedge\neg root(x))}$\\
        && $\mrm{~~~~~\wedge\A{E}x(\mE(x)=none)}$
    \end{tabular}\end{small}\\
    
    \item \textit{Proof of} Slp$(f,\mtt{col\_red})$ \textit{implies} $f$.\\
    \begin{small}
    \begin{tabular}{lcl}
        \multicolumn{3}{l}{$G\Sat$Slp$(f,\mtt{c\_red})$}\\
        & $\Leftrightarrow$ & $G\Sat$Slp$(f,\mtt{c\_blue})$\\
        & $\Rightarrow$ & 
        $G\Sat\mrm{\A{V}x((\mV(x)=none\vee\mV(x)=blue\vee\mV(x)=red)\wedge\neg root(x))}$\\
        && $\mrm{~~~~~\wedge\A{E}x(\mE(x)=none)}$
    \end{tabular}\end{small}\\

    \item \textit{Proof of} Slp$(f,\mtt{unmark})$ \textit{implies} $f$.\\
    \begin{small}
    \begin{tabular}{lcl}
        \multicolumn{3}{l}{$G\Sat$Slp$(f,\mtt{unmark})$}\\
        & $\Leftrightarrow$ & $G\Sat\mrm{\E{V}y(\A{V}x(x=y\vee((\mV(x)=red\vee\mV(x)=blue\vee\mV(x)=none)\wedge\neg root(x)))}$\\
        &&$\mrm{~~~~~~~~~~~~~\wedge\mV(y)=none\wedge\neg root(y))}$  \\
        && $\mrm{~~~~~\wedge\A{E}x(\mE(x)=none)}$\\
        & $\Rightarrow$ & 
        $G\Sat\mrm{\A{V}x((\mV(x)=none\wedge\neg root(x))}$\\
        &&$\mrm{~~~~~~~~~~~~~\vee((\mV(x)=red\vee\mV(x)=blue\vee\mV(x)=none)\wedge\neg root(x)))}$\\
        && $\mrm{~~~~~\wedge\A{E}x(\mE(x)=none)}$\\
        & $\Rightarrow$ & 
        $G\Sat\mrm{\A{V}x((\mV(x)=none\vee\mV(x)=blue\vee\mV(x)=red)\wedge\neg root(x))}$ \\
        && $\mrm{~~~~~\wedge\A{E}x(\mE(x)=none)}$
    \end{tabular}\end{small}\\
    
    \item \textit{Proof of} $f\wedge$Fail$(\mtt{init;Colour!})$ \textit{implies} $e$.\\
    \begin{small}
    \begin{tabular}{lcl}
        \multicolumn{3}{l}{$G\Sat f\wedge$Fail$(\mtt{init;Colour!})$}\\
        & $\Leftrightarrow$ & $G\Sat\mrm{\A{V}x((\mV(x)=none\vee\mV(x)=blue\vee\mV(x)=red)\wedge\neg root(x))}$\\
        &&$\mrm{~~~~~\wedge\neg\E{V}x(\mV(x)=none\wedge\neg root(x))\wedge\A{E}x(\mE(x)=none)}$ \\
        & $\Rightarrow$ & $G\Sat\mrm{\A{V}x((\mV(x)=blue\vee\mV(x)=red)\wedge\neg root(x))\wedge\A{E}x(\mE(x)=none)}$
    \end{tabular}\end{small}\\
    
    \item \textit{Proof of} $f\wedge$Fail$(\mtt{unmark})$ \textit{implies} $c\vee d$.\\
    \begin{small}
    \begin{tabular}{lcl}
        \multicolumn{3}{l}{$G\Sat f\wedge$Fail$(\mtt{unmark})$}\\
        & $\Leftrightarrow$ & $G\Sat\mrm{\A{V}x((\mV(x)=none\vee\mV(x)=blue\vee\mV(x)=red)\wedge\neg root(x))}$\\
        &&$\mrm{~~~~~\wedge\neg\E{V}x(\mV(x)\neq none\wedge\neg root(x))\wedge\A{E}x(\mE(x)=none)}$ \\
        & $\Rightarrow$ & $G\Sat\mrm{\A{V}x(\mV(x)=none\wedge\neg root(x))\wedge\A{E}x(\mE(x)=none)}$\\
        & $\Rightarrow$ & $G\Sat (\mrm{\A{V}x(\mV(x)=none\wedge\neg root(x))\wedge\A{E}x(\mE(x)=none)})\vee d$
    \end{tabular}\end{small}\\
    
    \item \textit{Proof of} $e\wedge$Fail$(\mtt{Illegal})$ \textit{implies} $d$.\\
    \begin{small}
    \begin{tabular}{lcl}
        \multicolumn{3}{l}{$G\Sat e\wedge$Fail$(\mtt{Illegal})$}\\
        & $\Leftrightarrow$ & $G\Sat\mrm{\A{V}x((\mV(x)=blue\vee\mV(x)=red)\wedge\neg root(x))\wedge\A{E}x(\mE(x)=none)}$\\
        &&$\mrm{~~~~~~\wedge\neg\E{E}x(((\mV(s(x))=red\wedge\mV(t(x))=red)}$\\
        &&$\mrm{~~~~~~~~~~~~~~~~\vee(\mV(s(x))=blue\wedge\mV(t(x))=blue))\wedge s(x)\neq t(x))}$ \\
        & $\Rightarrow$ & $G\Sat\mrm{\A{V}x(\mV(x)=blue\vee\mV(x)=red)}$\\
        &&$\mrm{~~~~~~\wedge\neg\E{E}x(((\mV(s(x))=red\wedge\mV(t(x))=red)}$\\
        &&$\mrm{~~~~~~~~~~~~~~~~\vee(\mV(s(x))=blue\wedge\mV(t(x))=blue))\wedge s(x)\neq t(x))}$ \\
        & $\Rightarrow$ & $G\Sat \mrm{\A{V}x((\mV(x)=red\vee\mV(x)=blue))}$\\
        &&$~~~~~~~\mrm{\wedge\neg\E{E}x(\mV(s(x))=\mV(t(x))\wedge s(x)\neq t(x))}$
    \end{tabular}\end{small}\\
    
    \item \textit{Proof of} $e\wedge$Success$(\mtt{Illegal})$ \textit{implies} $f$.\\
    \begin{small}
    \begin{tabular}{lcl}
        \multicolumn{3}{l}{$G\Sat e\wedge$Success$(\mtt{Illegal})$}\\
        & $\Rightarrow$ &
        $G\Sat e$\\
        & $\Rightarrow$ &$G\Sat\mrm{\A{V}x((\mV(x)=blue\vee\mV(x)=red)\wedge\neg root(x))\wedge\A{E}x(\mE(x)=none)}$\\
        & $\Rightarrow$ & $G\Sat\mrm{\A{V}x((\mV(x)=none\vee\mV(x)=blue\vee\mV(x)=red)\wedge\neg root(x))}$\\
        &&$\mrm{~~~~~~\wedge\A{E}x(\mE(x)=none)}$
    \end{tabular}\end{small}\\
\end{enumerate}

%% file: input/8_RelatedWork.tex
Hoare-style verification of graph programs with attributed rules was introduced in \cite{PoskittP12,Poskitt13}, using E-conditions which generalise the nested graph conditions of Habel and Pennemann \cite{HP09,Pennemann09}. E-conditions do not cover rooted rules or the $\mtt{break}$ command, which are considered in our first-order formulas. More importantly, the approach of \cite{PoskittP12,Poskitt13} can only handle programs in which the conditions of branching commands and loop bodies are rule set calls. Our syntactic calculus \textsf{SYN} covers a larger class of graph programs, viz.\ programs where the condition of each branching command is a loop-free program, and each loop body is an iteration command. This allows us, in particular, to verify many programs with nested loops. Besides this increased power, we believe that assertions in the form of first-order formulas are easier to comprehend by programmers than nested graph conditions of some form. 

% The ability to specify marks in the text is also an advantage to express a property in a straightforward way. For an example, if we want to express ``all nodes are unmarked", we can express it as $\mrm{\forall_Vx(\mV(x)=none)}$ by our first-order formula. However in \cite{PoskittP12,Poskitt13}, the simplest way to express it is: $\neg\exists$(\begin{tikzpicture}[scale=0.6, transform shape, minimum size=.1cm,baseline,thick]
%   \node[circle, draw, fill=grey!40] (a) at (0, 0) {$\mtt{a}$};
% \end{tikzpicture}) $\wedge \neg\exists$(\begin{tikzpicture}[scale=0.6, transform shape, minimum size=.1cm,baseline,thick]
%   \node[circle, draw, fill=red!70] (a) at (0, 0) {$\mtt{a}$};
% \end{tikzpicture}) $\wedge\neg\exists$(\begin{tikzpicture}[scale=0.6, transform shape, minimum size=.1cm,baseline,thick]
%   \node[circle, draw, fill=green!70] (a) at (0, 0) {$\mtt{a}$};
% \end{tikzpicture}) $\wedge\neg\exists$(\begin{tikzpicture}[scale=0.6, transform shape, minimum size=.1cm,baseline,thick]
%   \node[circle, draw, fill=blue!50] (a) at (0, 0) {$\mtt{a}$};
% \end{tikzpicture}).

As argued at the end of the previous section, we cannot express SLP$(c,P)$ or WLP$(P,c)$ for arbitrary assertions $c$ and graph programs $P$ as first-order formulas. In \cite{HP09,Pennemann09}, there is a construction of Wlp$(c,P!)$ by using an infinite formula. Here, we do not use a similar trick but stick to standard finitary logic. The papers \cite{DijkstraS90,Jones-Roscoe-Wood10a} do not give constructions for syntactic strongest liberal postconditions or weakest liberal postconditions either. Instead, similar to the consequent of our inference rule [alap], the conjunction of a loop invariant and a negated loop condition is considered as an ``approximate" strongest liberal postcondition. 

In \cite{Brenas-Echahed-Strecker18b}, the authors design an imperative programming language for manipulating graphs and give a Hoare calculus based on weakest preconditions. Programs manipulate the graph structure only and do not contain arithmetic. Assertions are formulas of the so-called guarded fragment of first-order logic, which is decidable. This relatively weak logic makes the correctness of programs decidable. 

Our goal is different in that we want a powerful assertion language that can specify many practical algorithms on graphs. (In fact, we plan to extend our logic to monadic second-order logic in order to express non-local properties such as connectedness, colourability, etc.) In our setting, it is easily seen that correctness is undecidable in general, even for trivial programs. For example, consider Hoare triples of the form $\{\mrm{true}\} \mtt{skip} \{d\}$ where d is an arithmetic formula (without references to nodes or edges). Such a triple is partially (and totally) correct if and only if d is true on the integers. But our formulas include Peano arithmetic and hence are undecidable in general \cite{Monk76a}. Thus, even for triples of the restricted form above, correctness is undecidable.

%% file: input/9_Conclusion.tex
We have shown how to construct a strongest liberal postcondition for a given conditional rule schema and a precondition in the form of a first-order formula. Using this construction, we have shown that we can obtain a strongest liberal postcondition over a loop-free program, and construct a first-order formula for SUCCESS$(C)$ for a loop-free program $C$. Moreover, we can construct a first-order formula for FAIL$(P)$ for an iteration command $P$. Altogether, this gives us a proof calculus that can handle more programs than previous calculi in the literature, in particular we can now handle certain nested loops.

However, the expressiveness of first-order formulas over the domain of graphs is quite limited. For example, one cannot specify that a graph is connected by a first-order formula. Hence, in the near future, we will extend our formulas to monadic second-order formulas to overcome such limitations \cite{Cou12}. 

Another limitation in current approaches to graph program verification is the inability to specify isomorphisms between the initial and final graphs \cite{WP18}. Monadic second-order transductions can link initial and final states by expressing the final state through elements of the initial state \cite{Cou12}. We plan to adopt this technique for graph program verification in the future. 